\documentclass[11pt]{article} 

\usepackage[utf8]{inputenc}
\usepackage[T1]{fontenc}
\usepackage[pdftex,left=1in,top=1in,bottom=1in,right=1in]{geometry}

\usepackage{amsmath, amssymb, dsfont, mathrsfs, amsthm}

\usepackage{graphicx}
\usepackage{caption}
\usepackage{bm}
\usepackage{xcolor}
\usepackage[colorlinks=true,citecolor=blue,linkcolor=blue]{hyperref}
\hypersetup{
    colorlinks,
    linkcolor={red!50!black},
    citecolor={blue!50!black},
    urlcolor={blue!80!black}
}
\usepackage[nameinlink, noabbrev, capitalize]{cleveref}

\allowdisplaybreaks

\newcommand{\N}{\mathds{N}}
\newcommand{\F}{\mathds{F}}
\newcommand{\Z}{\mathbb{Z}}

\newcommand{\E}{\mathds{E}}

\newcommand{\cC}{\mathcal{C}}

\newcommand{\cS}{\mathcal{S}}
\newcommand{\cT}{\mathcal{T}}

\newcommand{\cX}{\mathcal{X}}
\newcommand{\cY}{\mathcal{Y}}

\newcommand{\loglog}{\operatorname{loglog}}

\newcommand{\set}[1]{\left\{{#1}\right\}}
\newcommand{\B}{\set{0,1}}

\newcommand{\Cond}{\mathsf{Cond}}
\newcommand{\BExt}{\mathsf{BExt}}

\newcommand{\supp}{\mathsf{supp}}

\newcommand{\eps}{\varepsilon}

\newtheorem{theorem-intro}{Theorem}
\newtheorem{lemma}{Lemma}[section]
\newtheorem{theorem}[lemma]{Theorem}
\newtheorem{corollary}[lemma]{Corollary}
\newtheorem{definition}[lemma]{Definition}
\newtheorem{claim}[lemma]{Claim}

\newtheorem{remark}[lemma]{Remark}

\AddToHook{env/theorem/begin}{\crefalias{lemma}{theorem}}
\AddToHook{env/corollary/begin}{\crefalias{lemma}{corollary}}
\AddToHook{env/definition/begin}{\crefalias{lemma}{definition}}
\AddToHook{env/claim/begin}{\crefalias{lemma}{claim}}
\AddToHook{env/proposition/begin}{\crefalias{lemma}{proposition}}
\AddToHook{env/observation/begin}{\crefalias{lemma}{observation}}
\AddToHook{env/remark/begin}{\crefalias{lemma}{remark}}

\newcommand{\poly}{\operatorname{poly}}
\newcommand{\polylog}{\operatorname{polylog}}
\newcommand{\Otilde}{\widetilde{O}}

\newcommand{\minH}{\mathbf{H}_\infty}
\newcommand{\aminH}{\widetilde{\mathbf{H}}_\infty}
\newcommand{\bits}{\{0,1\}}

\newcommand{\Ext}{\mathsf{Ext}}
\newcommand{\KTCond}{\mathsf{KTCond}}
\newcommand{\RSCond}{\mathsf{RSCond}}
\newcommand{\Samp}{\mathsf{Samp}}
\newcommand{\Tre}{\mathsf{Tre}}

\newcommand{\Mod}[1]{\ \mathrm{mod}\ #1}

\let\originalleft\left
\let\originalright\right
\renewcommand{\left}{\mathopen{}\mathclose\bgroup\originalleft}
\renewcommand{\right}{\aftergroup\egroup\originalright}

\title{Nearly-Linear Time Seeded Extractors with Short Seeds
}

\author{Dean Doron\thanks{Ben-Gurion University. \texttt{deand@bgu.ac.il}. Part of this work was done while visiting Instituto de Telecomunicações and the Simons Institute for the Theory of Computing.} \and João Ribeiro\thanks{Instituto de Telecomunicações and Departamento de Matemática, Instituto Superior Técnico, Universidade de Lisboa. \texttt{jribeiro@tecnico.ulisboa.pt}. Part of this work was done while at NOVA LINCS and NOVA School of Science and Technology, and while visiting the Simons Institute for the Theory of Computing.}}

\date{}

\begin{document}
	
\maketitle
	
\begin{abstract}
    Seeded extractors are fundamental objects in pseudorandomness and cryptography, and a deep line of work has designed polynomial-time seeded extractors with nearly-optimal parameters.
    However, existing constructions of seeded extractors with short seed length and large output length run in time $\Omega(n \log(1/\eps))$ and often slower, where $n$ is the input source length and $\eps$ is the error of the extractor.
    Since cryptographic applications of extractors require $\eps$ to be small, the resulting runtime makes these extractors impractical.
    
    Motivated by this, we explore constructions of strong seeded extractors with short seeds computable in nearly-linear time $O(n \log^c n)$, for any error $\eps$.
    We show that an appropriate combination of modern condensers and classical approaches for constructing seeded extractors for high min-entropy sources yields such extractors.
    More precisely, we obtain strong extractors for $n$-bit sources with any min-entropy $k$ and any target error $\eps$ with seed length $d=O(\log(n/\eps))$ and output length $m=(1-\eta)k$ for an arbitrarily small constant $\eta>0$, running in nearly-linear time.
    When $k$ or $\eps$ are very small, our construction requires a reasonable one-time preprocessing step.
    These extractors directly yield privacy amplification protocols with nearly-linear time complexity (possibly after a one-time preprocessing step), large output length, and low communication complexity.
    As a second contribution, we give an instantiation of Trevisan's extractor that can be evaluated in \emph{truly} linear time in the RAM model, as long as the number of output bits is at most $\frac{n}{\log(1/\eps)\polylog(n)}$.
    Previous fast implementations of Trevisan’s extractor ran in $\widetilde{O}(n)$ time in this setting.
\end{abstract}

\newpage
\tableofcontents
\newpage

\section{Introduction}\label{sec:intro}

Seeded randomness extractors are central objects in the theory of pseudorandomness.
A strong $(k,\eps)$-seeded extractor, first introduced by Nisan and Zuckerman~\cite{NZ96}, is a deterministic function $\Ext\colon \B^n\times\B^d\to\B^m$ that receives as input an $n$-bit source of randomness $X$ with $k$ bits of min-entropy\footnote{A random variable $X$ has $k$ bits of min-entropy if $\Pr[X=x]\leq 2^{-k}$ for all $x$. Min-entropy has been the most common measure for the quality of a weak source of randomness since the work of Chor and Goldreich~\cite{CG88}.} and a $d$-bit independent and uniformly random seed $Y$, and outputs an $m$-bit string $\Ext(X,Y)$ that is $\eps$-close in statistical distance to the uniform distribution over $\B^m$, where $\eps$ is an error term, even when the seed $Y$ is revealed.
Besides their most direct application to the generation of nearly-perfect randomness from imperfect physical sources of randomness (and their early applications to derandomizing space-bounded computation~\cite{NZ96} and privacy amplification~\cite{BBR88,BBCM95}), seeded extractors have also found many other surprising applications throughout computer science, particularly in cryptography (specifically, in leakage-resilient cryptography~\cite{SV19, QWW21} and non-malleable cryptography~\cite{BGW19,CGL20,AKOOS22}).

For most applications, it is important to minimize the \emph{seed length} of the extractor.
A standard application of the probabilistic method shows the existence of strong $(k,\eps)$-seeded extractors with seed length $d=\log(n-k)+2\log(1/\eps)+O(1)$ and output length $m=k-2\log(1/\eps)-O(1)$, and we also know that these parameters are optimal up to the $O(1)$ terms~\cite{RT00}.
This motivated a deep line of research devising explicit constructions of seeded extractors with seed length as small as possible spanning more than a decade (e.g., \cite{NZ96,SZ99,NT99,Tre01,TZS06,SU05}) and culminating in extractors with essentially optimal seed length \cite{LRVW03,GUV09}. In particular, the beautiful work of Guruswami, Umans, and Vadhan~\cite{GUV09} gives explicit strong extractors with order-optimal seed length $d=O(\log(n/\eps))$ and output length $m=(1-\eta)k$ for any constant $\eta>0$, and follow-up work~\cite{DKSS13,TU12} further improved $m$ to $(1-o(1))k$ (at the expense of higher error).
The extractors constructed in these works are explicit, in the sense that there is an algorithm that given $x$ and $y$ computes the corresponding output $\Ext(x,y)$ in time polynomial in the input length.

A closer look shows that the short-seed constructions presented in the literature all run in time $\Omega(n\log(1/\eps))$, and often significantly slower.
In cryptographic applications of extractors we want the error guarantee $\eps$ to be small, which means that implementations running in time $\Omega(n\log(1/\eps))$ are often impractical.
If we insist on nearly-linear runtime for arbitrary error $\eps$, we can use strong seeded extractors based on universal hash functions that can be implemented in $O(n\log n)$ time (e.g., see~\cite{HT16}) and have essentially optimal output length, but have the severe drawback of requiring a very large seed length $d=\Omega(m)$. 

These limitations have been noted in a series of works studying concrete implementations of seeded extractors, with practical applications in quantum cryptography in mind~\cite{MPS12,FWEBC23,FYEC24}.
For example, Foreman, Yeung, Edgington, and Curchod~\cite{FYEC24} implement a version of Trevisan's extractor~\cite{Tre01,RRV02} with its standard instantiation of Reed--Solomon codes concatenated with the Hadmadard code, and emphasize its excessive running time as a major reason towards non-adoption.\footnote{The reason why these works focus on Trevisan's extractor is that this is the best seeded extractor (in terms of asymptotic seed length) that is known to be secure against quantum adversaries~\cite{DPVR12}.}
Instead, they have to rely on extractors based on universal hash functions, which, as mentioned above, are fast but require very large seeds.

This state of affairs motivates the following question, which is the main focus of this work:
\begin{quote}
    \em
    Can we construct strong $(k,\eps)$-seeded extractors with seed length $d=O(\log(n/\eps))$ and output length $m=(1-\eta)k$ computable in nearly-linear time, for arbitrary error $\eps$?
\end{quote}
\noindent
Progress on this problem would immediately lead to faster implementations of many cryptographic protocols that use seeded extractors, like those mentioned above -- the most significant performance gains would be in the context of privacy amplification.

\subsection{Our Contributions}

We make progress on the construction of nearly-linear time seeded extractors.

\paragraph{Seeded extractors with order-optimal seed length and large output length.}
We construct nearly-linear time strong seeded extractors with order-optimal seed length and large output length for any $k$ and $\eps$, with the caveat that they require a reasonable one-time preprocessing step whenever $k$ or $\eps$ are small.
More precisely, we have the following result.
\begin{theorem-intro}\label{thm:rec-ext-intro}
    For any constant $\eta>0$ there exists a constant $C>0$ such that the following holds.
    For any positive integers $n$ and $k\leq n$ and any $\eps>0$ satisfying $k\geq C \log(n/\eps)$ there exists a strong $(k,\eps)$-seeded extractor \[\Ext \colon \bits^n\times\bits^d\to\bits^m\] with seed length $d\leq C\log(n/\eps)$ and output length $m\geq (1-\eta)k$.
    Furthermore,
    \begin{enumerate}
        \item If $k \geq 2^{C\log^*\!n}\cdot \log^2(n/\eps)$ and $\eps\geq 2^{-Cn^{0.1}}$, then $\Ext$ is computable in time $\Otilde(n)$, where $\Otilde(\cdot)$ hides polylogarithmic factors in its argument and $\log^*\!$ denotes the iterated logarithm;
        \label{it:case1}

        \item If $k\geq 2^{C\log^*\!n}\cdot \log^2(n/\eps)$ and $\eps< 2^{-Cn^{0.1}}$, then $\Ext$ is computable in time $\Otilde(n)$ after a preprocessing step, corresponding to generating $O(\log^*\!n)$ primes $q\leq \poly(n/\eps)$; \label{it:case2}

        \item If $k < 2^{C\log^*\!n}\cdot \log^2(n/\eps)$, then $\Ext$ is computable in time $\Otilde(n)$ after a preprocessing step, corresponding to generating $O(\log\log n)$ primes $q\leq \poly(n/\eps)$ and a primitive element for each field $\F_q$. \label{it:case3}      
    \end{enumerate}

    The one-time preprocessing steps above can be implemented in time $\polylog(n/\eps)$ using randomness.\footnote{At least naively, deterministic algorithms for the above one-time preprocessing steps require time $\poly(n/\eps)$, by testing primality of all numbers in an interval of length $\poly(n/\eps)$. 
    The preprocessing step in \cref{it:case2} can be replaced by a different known procedure that can be implemented deterministically in time $\Otilde(n+\sqrt{n}\log(1/\eps)^{4+\delta})$ for any constant $\delta>0$.
    Note that this is $\Otilde(n)$ when, say, $\eps\geq 2^{-Cn^{0.1}}$, meaning that in this case its runtime can be absorbed into the runtime of the extractor, yielding \cref{it:case1}.
    This constraint can be weakened to roughly $\eps\geq 2^{-Cn^{0.16}}$ by using an alternative approach that we sketch in \cref{sec:remove-preproc}.
    To avoid overloading \cref{thm:rec-ext-intro}, we leave these discussions to \cref{remark:KT-preproc} and \cref{sec:remove-preproc}, respectively.}

\end{theorem-intro}

\cref{thm:rec-ext-intro}
follows from combining modern condensers with short seeds (namely, the lossless condenser of Kalev and Ta-Shma~\cite{KT22} and the lossy Reed-Solomon-based condenser of Guruswami, Umans, and Vadhan~\cite{GUV09}) with a careful combination and instantiation of classical recursive approaches developed by Srinivasan and Zuckerman~\cite{SZ99} and in~\cite{GUV09}.
It readily implies, among other things, an $\Otilde(n)$-time privacy amplification protocol where only $O(\log(n/\eps))$ bits need to be communicated over the one-way authenticated public channel and almost all the min-entropy can be extracted (after a reasonable one-time preprocessing step if $k$ or $\eps$ are very small).

\begin{remark}[complexity of the preprocessing steps]
    \em
    Both of the one-time preprocessing steps in \cref{thm:rec-ext-intro} are well-studied, and can be implemented in time $\polylog(n/\eps)$ using randomness.
    They are related to the condensers we use in the construction, and we discuss their complexity in more detail in \cref{remark:KT-preproc,remark:preproc3}, 
    and some approaches (and barriers) towards derandomizing them in \cref{sec:remove-preproc}.
\end{remark}

\paragraph{A new non-recursive construction.}
As a conceptual contribution which may be of independent interest, we present a new ``non-recursive'' construction of extractors with seed length $O(\log(n/\eps))$ and output length $(1-\eta)k$ that is computable in nearly-linear time when $k>\polylog(1/\eps)$ and avoids the complicated recursive procedures from~\cite{SZ99,GUV09} and \cref{thm:rec-ext-intro}.
We believe this to be a conceptually better approach towards constructing seeded extractors, and we discuss it in more detail in the technical overview and fully in \cref{sec:non-rec}.

\paragraph{Faster instantiations of Trevisan's extractor.}
One of the most widely-used explicit seeded extractors is Trevisan's extractor~\cite{Tre01,RRV02}.
While by now we have extractors with better parameters, one of its main advantages is that it is one of the few examples of extractors, and in a sense the best one, which are known to be \emph{quantum-proof}.\footnote{An extractor is quantum-proof if its output is close to uniform even in the presence of a quantum adversary that has some (bounded) correlation with $X$. A bit more formally, $\Ext$ is quantum-proof if for all classical-quantum states $\rho_{XE}$ (where $E$ is a quantum state correlated with $X$) with $H_{\infty}(X|E) \ge k$, and a uniform seed $Y$, it holds that $\rho_{\Ext(X,Y)YE} \approx_{\eps} \rho_{U_m} \otimes \rho_{Y} \otimes \rho_{E}$. See \cite{DPVR12} for more details.}

Trevisan's extractor uses two basic primitives: combinatorial designs (when more than one output bit is desired), and binary list-decodable codes. 
A standard instantiation of such suitable codes goes by concatenating a Reed-Solomon code with a Hadamard code, and this is also what is considered in~\cite{FWEBC23,FYEC24}.
As they also observe, this gives a nearly-linear time construction when the output length $m=1$.
In fact, by leveraging fast multipoint evaluation, one can also get a nearly-linear time construction for any output length $m\leq \frac{n}{\log(1/\eps)}$, although this was not noted in previous works.\footnote{For a rigorous statement on fast multipoint evaluation, see \cref{lemma:fast}.}

We present an alternative instantiation of Trevisan's extractor that can be computed in truly linear time on a RAM in the logarithmic cost model, for any output length $m\leq \frac{n}{\log(1/\eps)\cdot \polylog(n)}$.
While the underlying technical details are simple, we opt to present this result here  
because it may be of interest to some readers due to wide interest in efficient implementations of Trevisan’s extractor, and because it was not observed in prior works.
\begin{theorem-intro}
There exists an instantiation of Trevisan's extractor, set to extract $m$ bits with any error $\eps > 0$, that is computable in: 
\begin{enumerate}
    \item Time $O(n)+m\log(1/\eps) \cdot \polylog(n)$ after a preprocessing step\footnote{This preprocessing step corresponds to precomputing the design, and is \emph{not} the same preprocessing step as in \cref{thm:rec-ext-intro}.} running in time $\widetilde{O}(m\log(n/\eps))$, on a RAM in the
    logarithmic cost model. In particular, there exists a universal constant $c$, such that whenever $m \le \frac{n}{\log(1/\eps)\cdot \log^{c}(n)}$, the instantiation runs in time $O(n)$, without the need for a preprocessing step. 
    \item Time $\widetilde{O}(n+m\log(1/\eps))$ in the Turing model.
\end{enumerate}
\end{theorem-intro}

We note that one interesting instantiation of the above theorem is when Trevisan's extractor is set to output $k^{\Omega(1)}$ bits for $k = n^{\Omega(1)}$. In this setting, Trevisan's extractor requires a seed of length $O\left( \frac{\log^{2}(n/\eps)}{\log(1/\eps)} \right)$, and, as long as $\eps$ is not too tiny, we get truly-linear runtime.

\subsection{Other Related Work}

Besides the long line of work focusing on improved constructions of explicit seeded extractors and mentioned in the introduction above, other works have studied randomness extraction in a variety of restricted computational models.
These include extractors computable by streaming algorithms~\cite{BRST02}, local algorithms~\cite{Lu02,Vad04,BG13,CL18}, AC$^0$ circuits~\cite{GVW15,CL18,CW24},
AC$^0$ circuits with a layer of parity gates~\cite{HIV22},
NC$^1$ circuits~\cite{CW24}, and low-degree polynomials~\cite{ACGLR22,AGMR24,GGHNY24}.
Moreover, some works have independently explored implementations of other fundamental pseudorandomness primitives  in various restricted computational models. These include $k$-wise and $\eps$-biased generators, which often play a key role in constructions of various types of extractors. See  \cite{healy2006constant,healy2008randomness,celis2013balls,meka2014fast} for a very partial list.

As mentioned briefly above, some works have also focused on constructing seeded extractors computable in time $O(n\log n)$, motivated by applications in privacy amplification for quantum key distribution.
Such constructions are based on hash functions, and are thus far restricted to $\Omega(m)$ seed length.
The work of Hayashi and Tsurumaru~\cite{HT16} presents an extensive discussion of such efforts. We also mention that nearly-linear time extractors with very short seed, in the regime $k=n^{\Omega(1)}$ and $\eps=n^{-o(1)}$, were given in \cite{DMOZ22}, with applications in derandomization.\footnote{Our extractor can replace the one constructed in \cite{DMOZ22}, and it is indeed more efficient. However, due to other bottlenecks, they need to work with high error $\eps=n^{-o(1)}$ and relatively large min-entropy, and so the difference between the two extractors is not significant.} 

The techniques in this paper build on, and extend, block-source extraction techniques \cite{NZ96,Zuc96,Zuc97,SZ99,GUV09}. 
Another line of work, notably including \cite{NT99,LRVW03,DW08,DKSS13,TU12}, utilizes  \emph{mergers} to construct seeded extractors.\footnote{For the definition of mergers, see, e.g., \cite{DKSS13}. We note that those works \emph{do} often use block source conversion techniques, but usually in a different manner.} 
However, when restricted to constructions that get optimal seed length, they generally do not support low error (let alone run in nearly-linear time). 
On the positive side, the state-of-the-art mergers-based constructions get sub-linear entropy loss \cite{DKSS13,TU12}, whereas the constructions in this paper do not. We stress that it is still a very interesting open problem to construct a low-error, optimal seed length extractor with sub-linear entropy loss.

\subsection{Technical Overview}

In a nutshell, we obtain \cref{thm:rec-ext-intro} by following two standard high-level steps:
\begin{enumerate}
    \item We apply a randomness condenser with small seed length $O(\log(n/\eps))$ to the original $n$-bit weak source $X$ to obtain an output $X'$ that is $\eps$-close to a high min-entropy source. \label{it:condensing-step}

    \item We apply a seeded extractor tailored to high min-entropy sources with small seed length $O(\log(n/\eps))$ to $X'$ to obtain a long output that is $\eps$-close to uniform. \label{it:high-entropy-ext}
\end{enumerate}
To realize this approach, we need to implement each of these steps in nearly-linear time $\Otilde(n)$ (possibly after a reasonable one-time preprocessing step).
We briefly discuss how we achieve this, and some pitfalls we encounter along the way.

\paragraph{Observations about nearly-linear time condensers.}

In order to implement \cref{it:condensing-step}, we need to use \emph{fast} condensers with short seeds.
Luckily for us, some existing state-of-the-art constructions of condensers can already be computed in nearly-linear time, although, to the best of our knowledge, this has not been observed before.
We argue this carefully in \cref{sec:fast-condensers}.

For example, the ``lossy Reed-Solomon condenser'' from~\cite{GUV09} interprets the source as a polynomial $f\in\F_q[x]$ of degree $d\leq n/\log q$ and the seed $y$ as an element of $\F_q$, and outputs  $\RSCond(f,y)=(f(y),f(\zeta y),\dots,f(\zeta^{m'-1} y))$, for an appropriate $m'$ and field size $q$, with $\zeta$ a primitive element of $\F_q$.
Evaluating $\RSCond(f,y)$ corresponds to evaluating the same polynomial $f$ on multiple points in $\F_q$.
This is an instance of the classical problem of multipoint evaluation in computational algebra, for which we know fast and practical algorithms (e.g., see~\cite[Chapter 10]{vZGG13} or \cref{lemma:fast}) running in time $\Otilde((d+m')\log q)=\Otilde(n)$, since $d\leq n/\log q$, and if $m'\leq n/\log q$.

A downside of this condenser is that it requires knowing a primitive element $\zeta$ of $\F_q$ with $q=\poly(n/\eps)$.
But note that finding this primitive element only needs to be done once for a given set of parameters $(n,k,\eps,m)$ independently of the actual seed and input source, and so we leave it as a one-time preprocessing step.
As discussed in \cref{remark:preproc3}, we have the freedom of choosing $q$ to be prime, and in that case we can find this primitive element in randomized time $\polylog(q)=\polylog(n/\eps)$.

The lossless ``KT condenser'' from~\cite{KT22} has a similar flavor. 
It interprets the source as a polynomial $f\in\F_q[x]$ and the seed $y$ as an evaluation point, and outputs $\KTCond(f,y)=(f(y),f'(y),\dots,f^{(m'-1)}(y))$, for some appropriate $m'$.
The problem of evaluating several derivatives of the same polynomial $f$ on the same point $y$ (sometimes referred to as Hermite evaluation) is closely related to the multipoint evaluation problem above, and can also be solved in time $\Otilde(n)$.\footnote{Interestingly, recent works used other useful computational properties of the KT condenser. Cheng and Wu~\cite{CW24} crucially use the fact that the KT condenser can be computed in NC$^1$. Doron and Tell \cite{DT23} use the fact that the KT condenser is logspace computable for applications in space-bounded derandomization.}
Furthermore, evaluating the KT condenser only requires preprocessing when $\eps$ is very small.
On the other hand, it only works when the min-entropy $k\geq C\log^2(n/\eps)$ for a large constant $C>0$, where $n$ is the source length and $\eps$ the target error of the condenser.

\paragraph{The ``ideal'' approach to seeded extraction from high min-entropy sources.}

We have seen that there are fast condensers with short seeds.
It remains to realize \cref{it:high-entropy-ext}.
Because of the initial condensing step, we may essentially assume that our $n$-bit weak source $X$ has min-entropy $k\geq (1-\delta)n$, for an arbitrarily small constant $\delta>0$.
In this case, we would like to realize in time $\Otilde(n)$ and with overall seed length $O(\log(n/\eps))$ what we see as the most natural approach to seeded extraction from high min-entropy sources:
\begin{enumerate}
    \item Use a fresh short seed to transform $X$ into a \emph{block source} $Z=(Z_1, Z_2,\dots, Z_t)$ with geometrically decreasing blocks.
    A block source has the property that each block $Z_i$ has good min-entropy even conditioned on the values of blocks $Z_1,\dots,Z_{i-1}$.\label{it:create-block}

    \item Perform \emph{block source extraction} on $Z$ using another fresh short seed. Due to its special structure, we can extract a long random string from $Z$ using only the (small) seed length associated with extracting randomness from the smallest block $Z_t$.\label{it:block-ext}
\end{enumerate}
Similar approaches were taken in \cite{NZ96,Zuc96}, but they do not support logarithmic seed and low-error (see \cref{sec:non-rec}). 
The classical approach to \cref{it:block-ext} where we iteratively apply extractors based on universal hash functions with increasing output lengths to the blocks of $Z$ from right to left is easily seen to run in time $\Otilde(n)$ and requires a seed of length $O(\log(n/\eps))$ if, e.g., we use the practical extractors of~\cite{TSSR11,HT16}.
Therefore, we only need to worry about realizing \cref{it:create-block}.

A standard approach to \cref{it:create-block} would be to use an \emph{averaging sampler} to iteratively sample subsequences of $X$ as the successive blocks of the block source $Z$, following a classical strategy of Nisan and Zuckerman~\cite{NZ96} (improved by \cite{RSW06,Vad04}).
We do know averaging samplers running in time $\Otilde(n)$ (such as those based on random walks on a carefully chosen expander graph).
However, this approach requires a fresh seed of length $\Theta(\log(n/\eps))$ \emph{per block of $Z$}.
Since $Z$ will have roughly $\log n$ blocks, this leads to an overall seed of length $\Theta(\log^2 n+\log(1/\eps))$, which is too much for us.

Instead, we provide a new analysis of a sampler based on bounded independence, that runs in time $\Otilde(n)$ and only requires a seed of length $O(\log(n/\eps))$ to create the \emph{entire} desired block source. 
However, this block source has blocks of \emph{increasing} lengths, whereas we need decreasing blocks to perform the block source extraction. We remedy that by sub-sampling from each block using a standard expander random walk sampler.

We give the construction, which may be of independent interest, in \cref{sec:bounded-indep-sampler}.
The caveat of this construction is that it only works as desired when the target error $\eps\geq 2^{-k^c}$ for some small constant $c>0$. See \cref{sec:non-rec} for the formal analysis.

\paragraph{Getting around the limitation of the ideal approach.}
We saw above that combining the ideal approach to seeded extraction from high min-entropy sources with the new analysis of the bounded independence sampler yields a conceptually simple construction with the desired properties when the error is not too small.
However, we would like to have $\Otilde(n)$-time seeded extraction with $O(\log(n/\eps))$ seed length and large output length for all ranges of parameters.

To get around this limitation of our first construction, it is natural to turn to other classical approaches for constructing nearly-optimal extractors for high min-entropy sources, such as those of Srinivasan and Zuckerman~\cite{SZ99} or Guruswami, Umans, and Vadhan~\cite{GUV09}.
These approaches consist of intricate recursive procedures combining a variety of combinatorial objects, and require a careful analysis.\footnote{In our view, these approaches are much less conceptually appealing than the ``ideal'' approach above. We believe that obtaining conceptually simpler constructions of fast nearly-optimal extractors that work for all errors is a worthwhile research direction, even if one does not improve on the best existing parameters.}
However, we could not find such an approach that works as is, even when instantiated with $\Otilde(n)$-time condensers and $\Otilde(n)$-time hash-based extractors.
In particular:
\begin{itemize}
    \item The GUV approach~\cite{GUV09} gives explicit seeded extractors with large output length and order-optimal seed length for \emph{any} min-entropy requirement $k$ and error $\eps$.
    However, its overall runtime is significantly larger than $\Otilde(n)$ whenever $\eps$ is not extremely small (for example, $\eps= 2^{-k^\alpha}$ for some $\alpha\in (0,1/2)$ is not small enough).

    \item The SZ approach~\cite{SZ99} can be made to run in time $\Otilde(n)$ and have large output length when instantiated with fast condensers, samplers, and hash-based extractors, but it is constrained to error $\eps\geq 2^{-ck/\log^{*}\! n}$, where $\log^{*}$ is the iterated logarithm.
\end{itemize}
Fortunately, the pros and cons of the GUV and SZ approaches complement each other.
Therefore, we can obtain our desired result by applying appropriately instantiated versions of the GUV and SZ approaches depending on the regime of $\eps$ we are targeting.

\subsection{Future Work}

We list here some directions for future work:
\begin{itemize}
    \item Remove the preprocessing step that our constructions behind \cref{thm:rec-ext-intro} require when $k$ or $\eps$ are small.
    We expand on some promising approaches suggested by Jesse Goodman and also some barriers we face in \cref{sec:remove-preproc}.

    \item On the practical side, develop software implementations of seeded extractors with near-optimal seed length and large output length. In particular, we think that our non-recursive construction in \cref{sec:non-rec} holds promise in this direction.
\end{itemize}

\subsection{Acknowledgements}

We thank Jesse Goodman for many valuable comments and suggestions that greatly improved this work, and in particular for suggesting the approach outlined in \cref{sec:remove-preproc}.

Part of this research was done while the authors were visiting the Simons Institute for the Theory of Computing, supported by DOE grant \#DE-SC0024124.
D.\ Doron's research was also supported by Instituto de Telecomunicações (ref.\ UIDB/50008/2020) with the financial support of FCT - Fundação para a Ciência e a Tecnologia and by NSF-BSF grant \#2022644.
J.\ Ribeiro's research was also supported by Instituto de Telecomunicações (ref.\ UIDB/50008/2020) and NOVA LINCS (ref.\ UIDB/04516/2020) with the financial support of FCT - Fundação para a Ciência e a Tecnologia.

\section{Preliminaries}

\subsection{Notation}
We often use uppercase Roman letters to denote sets and random variables -- the distinction will be clear from context.
We denote the support of a random variable $X$ by $\supp(X)$, and for a random variable $X$ and set $S$, we also write $X\sim S$ to mean that $X$ is supported on $S$.
For a random variable $X$, we write $x\sim X$ to mean that $x$ is sampled according to the distribution of $X$.
We use $U_d$ to denote a random variable that is uniformly distributed over $\B^d$.
For two strings $x$ and $y$, we may write $(x,y)$ for their concatenation.
Given two random variables $X$ and $Y$, we denote their product distribution by $X\times Y$ (i.e., $\Pr[X\times Y=(x,y)]=\Pr[X=x]\cdot \Pr[Y=y]$).
Given a positive integer $n$, we write $[n]=\{1,\dots,n\}$.
For a prime power $q$, we denote the finite field of order $q$ by $\F_q$.
We denote the base-$2$ logarithm by $\log$.

\subsection{Model of Computation}
We work in the standard, multi-tape, Turing machine model with some fixed number of
work tapes. In particular, there exists a constant $C$ such that all our claimed time bounds hold whenever we work with
at most $C$ work tapes. This also implies that our results hold in the RAM model, wherein each machine word can store integers
up to some fixed length, and standard word operations take constant time. In \cref{sec:fast-trevisan} we will give, in addition to the standard Turing machine model bounds, an improved runtime bound that is dedicated to the logarithmic-cost RAM model.

\subsection{Fast Finite Field Operations}

Understanding the complexity of operations in finite fields 
will be useful for the analysis of the complexity of the condensers from~\cite{GUV09,KT22} in \cref{sec:fast-condensers}.
For a prime power $q = p^{\ell}$, we let $M_{q}(d)$ be the number of field operations required to multiply two univariate polynomials over $\F_q$ of degree less than $d$, and $M_{q}^{\mathsf{b}}(d)$ be the bit complexity of such a multiplication, so $M_{q}^{\mathsf{b}}(d) \le M_{q}(d) \cdot T(q)$, where we denote by $T(q)$ an upper bound on the bit complexity of arithmetic operations in $\F_q$.
Harvey and van der Hoeven \cite{HH19} (see also \cite{HH22}) showed that
\[
M^{\mathsf{b}}_{q}(d) = O(d \log q\cdot \log(d\log q) \cdot 4^{\log^{*}(d\log q)}).
\]
Overall, when $d \le q$, we have that $M_{q}^{\mathsf{b}}(d) = d\log d \cdot \widetilde{O}(\log q)$.\footnote{A similar bound can be obtained using simpler methods. If $\F_q$ contains a $d$-th root of unity, one can get $M_{q}(d)=d\log d$ from the classic FFT algorithm \cite{CT65}. For a simpler algorithm attaining the bound $M_{q}(d)=d\log d \loglog d$, see \cite[Sections 8, 10]{vZGG13}. When $p=2$, $M_{q}(d)=d\log d \loglog d$ also follows from Sch{\"o}nhage’s algorithm \cite{Sch77}. Now, since $M_{q}^{\mathsf{b}}(d) = M_{q}(d) \cdot M_{p}(\ell) \cdot M_{p}^{\mathsf{b}}(0)$, a bound of $d\log d \cdot \widetilde{O}(\log q)$ also follows.}

We will use fast multi-point evaluation and fast computation of derivatives (together with the preceding bounds on $M_{q}^{\mathsf{b}}$).

\begin{lemma}[\protect{\cite{BM74}, see also \cite[Chapter 10]{vZGG13}}]\label{lemma:fast}
Let $d \in \mathbb{N}$, and let $q=p^r$ be a prime power. Then,
given
a polynomial $f \in \F_q[X]$  of degree at most $d$ (together with a representation of $\F_q$ via an irreducible polynomial over $\F_p$ of degree $r$), the following holds.
\begin{enumerate}
    \item Given a set $\set{\alpha_1,\ldots,\alpha_t} \subseteq \F_q$, where $t \le d$, one can compute $f(\alpha_1),\ldots,f(\alpha_t)$ in time 
    $O(M^{\mathsf{b}}_{q}(d) \cdot \log d) = d\log^{2}d \cdot \widetilde{O}(\log q)$. \label{arithmetic:it1}
    \item For $t \le d$ and $\alpha \in \F_q$, one can compute the derivatives $f(\alpha),f'(\alpha),\ldots,f^{(t)}(\alpha)$ in time  $O(M_{q}^{\mathsf{b}}(d) \cdot \log d)=d\log^{2}d \cdot \widetilde{O}(\log q)$. 
\end{enumerate}
Note that when $q \le 2^d$, we can bound $O(M_{q}(d) \cdot \log d)$ by $\widetilde{O}(d) \cdot \log q$.
\end{lemma}

For a comprehensive discussion of fast polynomial arithmetic, see Von Zur Gathen and Gerhard's book~\cite{vZGG13} (and the more recent important developments~\cite{HH21}).

\subsection{Statistical Distance, Entropy}

We present some relevant definitions and lemmas about statistical distance and min-entropy.

\begin{definition}[statistical distance]
    The \emph{statistical distance} between two random variables $X$ and $Y$ supported on $\cS$, denoted by $\Delta(X,Y)$, is defined as
    \begin{equation*}
        \Delta(X,Y) = \max_{\cT\subseteq \cS}|\Pr[X\in \cT]-\Pr[Y\in \cT]| = \frac{1}{2}\sum_{x\in \cS}|\Pr[X=x]-\Pr[Y=x]|.
    \end{equation*}
    We say that $X$ and $Y$ are \emph{$\eps$-close}, and write $X\approx_\eps Y$, if $\Delta(X,Y)\leq \eps$.
\end{definition}

\begin{definition}[min-entropy]
    The \emph{min-entropy} of a random variable $X$ supported on $\cX$, denoted by $\minH(X)$, is defined as
    \begin{equation*}
        \minH(X) = -\log\left(\max_{x\in\cX} \Pr[X=x]\right).
    \end{equation*} 
The min-entropy \emph{rate} of $X$ is given by $\frac{\minH(X)}{\log|\cX|}$.
\end{definition}

\begin{definition}[average conditional min-entropy]
    Let $X$ and $Y$ be two random variables supported on $\cX$ and $\cY$, respectively.
    The \emph{average conditional min-entropy} of $X$ given $Y$, denoted by $\aminH(X|Y)$, is defined as
    \begin{equation*}
        \aminH(X|Y) = -\log\left(\E_{y\sim Y}[2^{-\minH(X|Y=y)}]\right).
    \end{equation*}
\end{definition}

The following standard lemma gives a chain rule for min-entropy.
\begin{lemma}[see, e.g., \cite{DORS08}]\label{lem:chainrule}
    Let $X$, $Y$, and $Z$ be arbitrary random variables such that $|\supp(Y)|\leq 2^\ell$.
    Then,
    \begin{equation*}
        \aminH(X|Y,Z)  \geq  \aminH(X|Z)-\ell.
    \end{equation*}
\end{lemma}

We can turn the chain rule above into a high probability statement.
\begin{lemma}[see, e.g., \cite{MW97}]\label{lem:tailboundchainrule}
    Let $X$, $Y$, and $Z$ be random variables such that $|\supp(Y)|\leq 2^\ell$.
    Then,
    \begin{equation*}
        \Pr_{y\sim Y}[\aminH(X|Y=y,Z)\geq \aminH(X|Z)-\ell-\log(1/\delta)]\geq 1-\delta
    \end{equation*}
    for any $\delta>0$.
\end{lemma}

\subsection{Extractors and Condensers}

\begin{definition}[$(n,k)$-source]
    We say that a random variable $X$ is an \emph{$(n,k)$-source} if $X\sim\B^n$ and $\minH(X)\geq k$.
\end{definition}

\begin{definition}[block source]
    A random variable $X$ is an \emph{$((n_1,n_2,\dots,n_t),(k_1,k_2,\dots,k_t))$-block source} if we can write $X=(X_1,X_2,\dots, X_t)$, each $X_i\in \B^{n_i}$, where $\aminH(X_i | X_1,\ldots,X_{i-1})\geq k_i$ for all $i\in[s]$.
    In the special case where $k_i=\alpha n_i$ for all $i\in[t]$, we say that $X$ is an \emph{$((n_1,n_2,\dots,n_t),\alpha)$-block source}.

    We say that $X$ is an \emph{exact} block source if $\minH(X_i|X_1=x_1,\ldots,X_{i-1}=x_{i-1}) \ge k_i$ for \emph{any} prefix $x_1,\dots,x_{i-1}$. \cref{lem:tailboundchainrule} tells
    us that any $((n_1,n_2,\ldots,n_t),\alpha)$-block-source is $\eps$-close to an exact 
    $((n_1,n_2,\ldots,n_t),(1-\zeta)\alpha)$-block-source, where $\eps = \sum_{i=1}^{t}2^{-\alpha\zeta n_i}$.
\end{definition}

\begin{definition}[seeded extractor]
A function $\Ext \colon \B^n \times \B^d \rightarrow \B^m$ is a 
\emph{$(k,\eps)$ seeded extractor} if the following holds. For every $(n,k)$-source $X$,
\begin{equation*}
    \Ext(X,Y)\approx_\eps U_{m},
\end{equation*}
where $Y$ is uniformly
distributed over $\B^d$ and is independent of $X$ and $U_{m}$ is uniformly distributed over $\B^{m}$.
We say that $\Ext$ is \emph{strong} if $(\Ext(X,Y), Y)\approx_\eps U_{m+d}$.

Furthermore, $\Ext$ is said to be an \emph{average-case $(k,\eps)$ (strong seeded) extractor} if for all correlated random variables $X$ and $W$ such that $X$ is supported on $\B^n$ and $\aminH(X|W)\geq k$ we have
\begin{equation*}
    (\Ext(X,Y), Y, W)\approx_\eps (U_{m+d}, W),
\end{equation*}
where $Y$ is uniformly
distributed over $\B^d$ and is independent of $X$, and $U_{m+d}$ is uniformly distributed over $\B^{m+d}$ and independent of $W$.
\end{definition}

\begin{remark}\label{rem:strong}
    \em
    By \cref{lem:tailboundchainrule}, every strong $(k,\eps)$-seeded extractor $\Ext \colon \B^n\times \B^d\to\B^m$ is also an average-case strong $(k'=k+\log(1/\eps),\eps'=2\eps)$-seeded extractor.
\end{remark}

\begin{definition}[condenser]
A function $\Cond \colon \B^n \times \B^d \rightarrow \B^m$ is a 
\emph{$(k,k',\eps)$ (seeded) condenser} if the following holds. For every $(n,k)$-source $X$, it holds that $Z = \Cond(X,Y)$ is $\eps$-close to some $Z'$ with $\minH(Z') \ge k'$, where $Y$ is uniformly
distributed over $\B^d$ and is independent of $X$. 

We say that $\Cond$ is \emph{strong} if
$(Y, \Cond(X,Y))$ is $\eps$-close to some distribution $(Y, Z)$ with min-entropy
$k'$ (and note that here, necessarily, $d$ bits of entropy come from the seed).
Finally, we say that $\Cond$ is
\emph{lossless} if $k' = k+d$.
\end{definition}

We also define seeded extractors tailored to block sources.
\begin{definition}[block source extractor]
    A function $\BExt\colon \B^{n_1}\times\cdots\times \B^{n_t}\times\B^d\to\B^m$ is a \emph{$(k_1,\dots,k_t,\eps)$ strong block-source extractor} if for any $((n_1,\dots,n_t),(k_1,\dots,k_t))$-block-source $X$,
    \begin{equation*}
        (\BExt(X,Y), Y)\approx_\eps U_{m+d},
    \end{equation*}
    where $Y$ is uniformly
    distributed over $\B^d$ and is independent of $X$ and $U_{m+d}$ is uniformly distributed over $\B^{m+d}$.
\end{definition}

We will also require the following extractors based on the leftover hash lemma and fast hash functions.
We state a result from~\cite{TSSR11} which requires seed length $d\approx 2m$, where $m$ is the output length.
\begin{lemma}[\protect{fast hash-based extractors~\cite[Theorem 10]{TSSR11}, adapted. See also~\cite[Table I]{HT16}}]\label{lem:fasthash}
    For any positive integers $n$, $k$, and $m$ and any $\eps>0$ such that $k\leq n$ and $m\leq k-4\log(16/\eps)$ there exists a $(k,\eps)$-strong seeded extractor $\Ext \colon \bits^n\times\bits^d\to\bits^m$ with seed length $d\leq 2(m+\log n+2\log(1/\eps)+4)$.
    Moreover, $\Ext$ can be computed in time $O(n\log n)$.

    Note that by appending the seed to the output of the extractor, we can get the following: There exists a constant $c$ such that for any constant $\theta \le \frac{1}{3}$, $d \ge c\log(n/\eps)$ and $k \ge \theta d+c\log(1/\eps)$, there exists a $(k,\eps)$-seeded extractor $\Ext \colon \B^{n} \times \B^{d} \rightarrow \B^{(1+\theta)d}$.
\end{lemma}

\subsection{Averaging Samplers}

In this section we define averaging samplers and state some useful related results and constructions.

\begin{definition}[averaging sampler]
    We say that $\Samp\colon\B^r\to [n]^m$ is a \emph{$(\gamma,\theta)$-averaging sampler} if
    \begin{equation*}
        \Pr_{(i_1,\dots,i_m)\sim \Samp(U_r)}\left[\left|\frac{1}{t}\sum_{j=1}^m f(i_j) - \E[f] \right| \ge \theta \right] < \gamma
    \end{equation*}
    for every function $f\colon[n]\to[0,1]$, where $\E[f]=\frac{1}{n}\sum_{i\in[n]}f(i)$.
    We say that $\Samp$ has \emph{distinct samples} if $\Samp(x)$ outputs $m$ distinct elements of $[n]$ for every input $x$. The parameter $\theta$ is often referred to as the \emph{accuracy} of the sampler, and $\gamma$ is its \emph{confidence} parameter. Moreover, we sometimes refer to $\Samp(U_r) \sim [n]^m$ as a $(\gamma,\theta)$ \emph{sampling distribution}.
\end{definition}

It is worth noting the equivalence between averaging samplers and extractors \cite{Zuc97}, where we think of the samples obtained by enumerating over the seeds of the extractor. Specifically, a $(k,\eps)$-seeded extractor over inputs of length $r$ is a $(2^{k-r+1},\eps)$-averaging sampler, and a $(\gamma,\theta)$-averaging sampler with input length $r$ is a $(k,2\theta)$-seeded extractor for $k = n-\log(1/\gamma)+\log(1/\theta)$.

The following lemma gives guarantees on sub-sampling coordinates from a weak source using an averaging sampler.
\begin{lemma}[\protect{\cite[Lemma 6.2]{Vad04}}]\label{lem:sample-entropy}
    Let $\delta,\gamma,\tau\in (0,1)$ be such that $\delta\geq 3\tau$ and let $\Samp \colon \bits^r\to [n]^m$ be a $(\gamma,\theta = \tau/\log(1/\tau))$-averaging sampler with distinct samples.
    Then, for any $(n,k=\delta n)$-source $X$ and $Y$ uniformly distributed over $\B^r$ we have that
    \begin{equation*}
        \left( Y, X_{\Samp(Y)} \right) \approx_{\gamma + 2^{-\Omega(\tau n)}} (Y, W),
    \end{equation*}
    where $(W|Y=y)$ is an $(m,k'=(\delta-3\tau)m)$-source for every $y$.
\end{lemma}

\paragraph{The ``expander random walk'' sampler.}
We will need the following well-known averaging sampler based on random walks on expanders (see, e.g., \cite{Gil98,Zuc07}).
Let $G$ be a $D$-regular graph with vertex set $[n]$.
We assume that the neighborhood of each vertex is ordered in some predetermined way.
Then, the associated averaging sampler parses its input $x$ as $(i_1,b_1,b_2,\dots,b_{t-1})$, where $i_1\in[n]$ and $b_1,\dots,b_{t-1}\in [D]$, and outputs $\Samp(x) = (i_1,\dots,i_t)$,
where $i_j$ is the $b_{j-1}$-th neighbor of $i_{j-1}$ when $j>1$.

The performance of $\Samp$ as an averaging sampler is determined by the spectral expansion of $G$.\footnote{We say that an undirected graph $G$ over $n$ vertices has spectral expansion $\lambda$ if $\lambda \le \max_{i \ge 2}|\lambda_i|$, where $\lambda_n \le \ldots \le \lambda_2 \le \lambda_1 = 1$ are the eigenvalues of $G$'s random walk matrix.}
In fact, if $G$ has spectral expansion $\theta/2$ then a direct application of the expander Chernoff bound \cite{Gil98} gives that $\Samp$ is a $(\gamma,\theta)$-averaging sampler with $t=O(\log(1/\gamma)/\theta^2)$ and $r=\log n + O(t\log(1/\theta))$~\cite[Section 8.2]{Vad04}.

To ensure distinct samples while maintaining accuracy, we follow \cite{Vad04} and modify the standard random walk sampler as follows. As a ``base'' sampler, we use the above walk, but only take the first $(1-\theta/2)t$ distinct vertices (if there are such). Letting $r_0 = \log n + O(t\log(1/\theta))$ be the corresponding randomness complexity, we then take a random walk of length $\ell = O(\log(1/\gamma))$ over an expander $G_0$ with $2^{r_0}$ vertices and constant spectral expansion. Each vertex of $G_0$ corresponds to a \emph{random walk on $G$}, and out of those $\ell$ chosen walks, we take the first one that indeed has $(1-\theta/2)t$ distinct vertices (and output some arbitrary value if the process failed). In \cite[Lemma 8.2]{Vad04}, it is shown that setting the spectral expansion of $G$ to be $\Theta(\theta)$ (and suitably choosing the constants inside the  $O()$ notation), one still gets a $(\gamma,\theta)$-averaging sampler.

We instantiate $G$ and $G_0$ with the
regular expander graphs from the following result of Alon~\cite{Alo21}.
\begin{lemma}[\protect{\cite[Theorem 1.2]{Alo21}, adapted}]\label{lem:exp-any-vertex}
    Fix any prime $p$ such that $p\equiv 1\mod 4$.
    Then, there is a constant $C_p$ such that for every integer $n\geq C_p$ there exists a $(D=p+2)$-regular graph $G_n$ on $n$ vertices with spectral expansion $\lambda \leq \frac{(1+\sqrt{2})\sqrt{D-1}+o(1)}{D}$, where the $o(1)$ tends to $0$ as $n\to\infty$.
    Furthermore, the family $(G_n)_n$ is strongly explicit.

    In particular, for any $\theta>0$ 
    there exist constants $C_\theta > 0$ and $D_\theta = O(\theta^{-2})$ and a strongly explicit family of $D_\theta$-regular graphs $(G_n)_{n}$ with spectral expansion $\lambda\leq \theta$ for any $n\geq C_\theta$.
\end{lemma}

Given a graph $G$, its $t$-th power $G^t$ is a graph over the same number of vertices, and each edge corresponds to a $t$-step walk over $G$. It is well-known that
taking the $t$-th power of a $\lambda$-spectral expander improves its expansion to $\lambda^t$. This readily gives us the following corollary.
\begin{corollary}\label{cor:expander}
For every large enough $n$, and any $\lambda = \lambda(n) > 0$, there exists a $D$-regular graph $G = (V = [n],E)$ with spectral expansion $\lambda$, where $D= \poly(1/\lambda)$, and given $x \in [n]$ and $i \in [D]$, the $i$-th neighbor of $x$ can be computed in time $\log(1/\lambda) \cdot \polylog(n) = \polylog(n)$. 
\end{corollary}

Combining the discussion above with \cref{lem:exp-any-vertex} (or \cref{cor:expander}) immediately yields the following, observing that the runtime is dominated by $\ell \cdot t \log(1/\lambda)  \polylog(n)$.

\begin{lemma}[\protect{\cite[Lemma 8.2]{Vad04}, appropriately instantiated}]\label{lem:baseSampler}
    For every large enough integer $n$ and every $\theta,\gamma\in(0,1)$, there exists a $(\gamma,\theta)$-averaging sampler $\Samp \colon \bits^r\to[n]^t$ with distinct samples with $t = O(\log(1/\gamma)/\theta^2)$ and $r = \log n + O(t\log(1/\theta))$.
    Furthermore, $\Samp$ is computable in time $t\log(1/\gamma)\cdot \polylog n$.
\end{lemma}

We can extend \cref{lem:baseSampler} to output more distinct samples while not increasing $r$ via the following simple lemma.
\begin{lemma}[\protect{\cite[Lemma 8.3]{Vad04}}]\label{lem:increaseSamplesDomain}
    Suppose that $\Samp_0 \colon \bits^r \to [n]^{t}$ is a $(\gamma,\theta)$-averaging sampler with distinct samples.
    Then, for every integer $m\geq 1$ there exists a $(\gamma,\theta)$-averaging sampler $\Samp \colon \bits^r \to [m\cdot n]^{m\cdot t}$ with distinct samples.
\end{lemma}
\cref{lem:increaseSamplesDomain} follows easily by parsing $[m\cdot t]=[m]\times [t]$ and considering the sampler $\Samp(x)_{i,j}=(i,\Samp_0(x)_j)$ for $i\in[m]$ and $j\in[t]$.
If we can compute $\Samp_0(x)$ in time $T$, then we can compute $\Samp(x)$ in time $T+O(mt\log(mn))$.
The following is an easy consequence of \cref{lem:baseSampler,lem:increaseSamplesDomain}.
\begin{lemma}[\protect{\cite[Lemma 8.4]{Vad04}, with additional complexity claim}]\label{lem:extendedSampler}
    There exists a constant $C>0$ such that the following holds.
    For every large enough $n$ and $\theta,\gamma\in(0,1)$, there exists a $(\gamma,\theta)$-averaging sampler $\Samp \colon \bits^r\to[n]^t$ with distinct samples for any $t\in[t_0,n]$ with $t_0 \leq C\log(1/\gamma)/\theta^2$  and $r = \log(n/t) + \log(1/\gamma)\cdot\poly(1/\theta)$.
    Furthermore, $\Samp$ is computable in time $t_0 \log(1/\gamma) \cdot \polylog n + O(t\log n)$.

    In particular, if $\theta$ is constant then $t_0=O(\log(1/\gamma))$, $r=\log (n/t)+O(\log(1/\gamma))$, and $\Samp$ is computable in time $\log^{2}(1/\gamma) \cdot\polylog n + O(t\log n)$.
\end{lemma}

\subsection{Standard Composition Techniques for Extractors}

We collect some useful classical techniques for composing seeded extractors.

\begin{lemma}[boosting the output length~\cite{WZ99,RRV02}]\label{lem:boost-output}
    Suppose that for $i\in\{1,2\}$ there exist strong $(k_i,\eps_i)$-seeded extractors $\Ext_i \colon \bits^{n}\times\bits^{d_i}\to\bits^{m_i}$ running in time $T_i$, with $k_2\leq k_1-m_1-g$.
    Then, $\Ext \colon \bits^n\times\bits^{d_1+d_2}\to\bits^{m_1+m_2}$ given by $\Ext(X,(Y_1,Y_2))=(\Ext_1(X,Y_1), \Ext_2(X,Y_2))$ is a strong $(k_1,\frac{\eps_1}{1-2^{-g}}+\eps_2)$-seeded extractor running in time $O(T_1+T_2)$.
\end{lemma}

\begin{lemma}[block source extraction]\label{lem:block-ext}
    Let $X=(X_1, \dots, X_t)$ be an $((n_1,\dots,n_t),(k_1,\dots,k_t))$-block-source, and let $\Ext_i \colon \bits^{n_i}\times \bits^{d_i}\to\bits^{m_i}$ be average-case $(k_i,\eps_i)$-seeded extractors running in time $T_i$ with output length $m_i\geq d_{i-1}$ for $i\geq 2$, \emph{that output their seed}.
    Then, there exists a strong $(k_1,\dots,k_t,\eps=\sum_{i\in[t]}\eps_i)$-block-source extractor $\BExt \colon \bits^{n_1}\times\cdots\times\bits^{n_t}\times\bits^{d_t}\to\bits^{m}$ with output length $m=m_1-d_t$ that runs in time $O(\sum_{i\in[t]}T_i)$.
    If $X$ is an exact block source, then the $\Ext_i$'s do not need to be average-case.
\end{lemma}

We discuss how the fast hash-based extractor from \cref{lem:fasthash} can be used to construct a fast extractor with seed length any constant factor smaller than its output length for high min-entropy sources.
We need the following lemma, which is an easy consequence of the chain rule for min-entropy.
\begin{lemma}[\protect{\cite[Corollary 4.16]{GUV09}}]
\label{lem:block-chain-rule}
    Let $X$ be an $(n,k=n-\Delta)$-source and let $X_1,\dots,X_t$ correspond to disjoint subsets of coordinates of $X$, with each $X_i$ of length $n_i\geq n'$.
    Then, $(X_1,\dots,X_t)$ is $t\eps$-close to an exact $((n_1,\dots,n_t),(k',\dots,k'))$-block-source for $k'=n'-\Delta-\log(1/\eps)$.\footnote{\cite[Corollary 4.16]{GUV09} is originally stated only for the special case where $X_1,\dots,X_t$ partition the coordinates of $X$. This extends easily to the statement presented here by noting that the puncturing of $X$ to its $\widetilde{n}=\sum_{i=1}^t n_i$ bits corresponding to $X_1,\dots,X_t$ is an $(\widetilde{n},\widetilde{n}-\Delta)$-source, and applying the original statement in~\cite{GUV09}.}
\end{lemma}

The following appears in~\cite{GUV09} without the time complexity bound.
We appropriately instantiate their approach and analyze the time complexity below.
\begin{lemma}[fast extractors with seed shorter than output \protect{\cite[Lemma 4.11]{GUV09}}]\label{lem:fasthash-short}
    For every integer $t\geq 1$ there exists a constant $C>0$ such that
    for any positive integer $n$ and $\eps>2^{-\frac{n}{50t}}$ there exists a strong $(k=(1-\frac{1}{20t})n,\eps)$-seeded extractor $\Ext \colon \B^n\times\B^d\to\B^m$ with $m\geq k/2$ and $d \leq k/t + C\log(n/\eps)$ computable in time $O(tn\log n)$.
\end{lemma}
\begin{proof}
    Let $X$ be an $(n,k=(1-\frac{1}{20t})n)$-source and $\eps'=\frac{\eps}{2t}$.
    Let $X_1, \dots, X_t$ correspond to disjoint subsets of coordinates of $X$, with $|X_i|= \lfloor n/t\rfloor=n'$ for all $i$.
    Then, \cref{lem:block-chain-rule} guarantees that $(X_1,\ldots,X_t)$ is $(t\eps')$-close to an exact $((n_1=n',\dots,n_t=n'),k'=n'-\frac{n}{20t}-\log(1/\eps'))$-block-source $X'$.
    Note that
    \begin{equation*}
        k'=n'-\frac{n}{20t} - \log(1/\eps')\geq \frac{19n}{20t} - 1 - \log(1/\eps') \geq 0.9n',
    \end{equation*}
    since we have assumed that $\eps> 2^{-\frac{n}{50t}}$.
    Now, let $\Ext' \colon \bits^{n'}\times\bits^d\to\bits^m$ be the strong $(k',\eps')$-seeded extractor from \cref{lem:fasthash} with output length $m=\left\lceil\frac{k}{2t}\right\rceil \leq k-4\log(16/\eps')$ and corresponding seed length $d\leq k/t + 4\log(n'/\eps')+9\leq k/t + C\log(n/\eps)$ for a large enough constant $C>0$ depending only on $t$.
    Then, we apply block source extraction (\cref{lem:block-ext}) to $X'$ with $\Ext_1=\cdots=\Ext_t=\Ext'$ to get the desired strong $(k,2t\eps'=\eps)$-extractor $\Ext$ with output length $t\cdot m\geq k/2$ and seed length $d$.
    Since $\Ext'$ is computable in time $O(n\log n)$, then $\Ext$ is computable in time $O(tn\log n)$.
\end{proof}

In addition to \cref{lem:boost-output}, one can potentially boost the output length of a high min-entropy extractor by first treating the source as a block source and then performing a simple block source extraction. The next corollary appears in \cite[Lemma 6.27]{Vad12}.
\begin{lemma}\label{cor:boosting2}
Let $\Ext_{\mathsf{in}} \colon \B^{n/2} \times \B^{\ell}  \rightarrow  \B^{d}$
and $\Ext_{\mathsf{out}} \colon \B^{n/2} \times \B^{d} \rightarrow \B^{m}$ be $(k',\eps)$-extractors. 
Then, for any $(n,k = \delta n)$-source 
$(X_1, X_2)$, each $X_i \sim \B^{n/2}$, and an independent uniform $Y \sim \B^{\ell}$, we have that
\[
\Ext((X_1,X_2),Y) = \Ext_{\mathsf{out}}(X_1,\Ext_{\mathsf{in}}(X_2,Y))
\]
is $4\eps$-close to uniform, assuming that $k' \le (\delta-\frac{3}{4})n$ and $\eps \ge 2^{-n/4}$.
In other words, $\Ext$ is a $(k,4\eps)$-extractor.
Moreover, if $\Ext_{\mathsf{in}}$ is strong then $\Ext$ is also strong, and if $\Ext_{\mathsf{in}}$ and $\Ext_{\mathsf{out}}$ run in time $T_1$ and $T_2$, respectively, then $\Ext$ runs in time $O(T_1+T_2)$. 
\end{lemma}

\section{Additional Building Blocks}

\subsection{Fast Generation of Small-Bias Sets}\label{sec:biased}

A distribution $S \sim \B^n$ is \emph{$\eps$-biased} if it is indistinguishable from uniform by every linear test. Namely, if for every nonempty $T \subseteq [n]$ it holds that $\Pr_{s \sim S}[\bigoplus_{i \in T} s_i=1] \in [\frac{1-\eps}{2},\frac{1+\eps}{2}]$. We say that a \emph{set} $S \subseteq \B^n$ is $\eps$-biased if the flat distribution over its support is $\eps$-biased.
We say that a linear code $\cC \subseteq \B^n$ is $\eps$-balanced if the Hamming weight of each nonzero codeword lies in $[\frac{1-\eps}{2}n,\frac{1+\eps}{2}n]$. It is known that these two objects are essentially the same: $S$ is $\eps$-biased if and only if the $|S| \times n$ matrix whose rows are the elements of $S$ is a generator matrix of an $\eps$-balanced code.

One prominent way of constructing $\eps$-balanced codes is via
\emph{distance amplification}, namely, starting with a code of some
bias  $\eps_0 \gg \eps$ and, using a parity sampler, amplify its bias. We will use a specific, simple, instantiation of a parity sampler -- the random walk sampler.
\begin{lemma}[\protect{RWs amplify bias \cite[Theorem 3.1]{TS17}\footnote{The argument for $t=2$ was suggested already by Rozenman and Wigderson (see \cite{bogdanov}). Note that the goal of the RW is to \emph{reduce} the bias, $\eps \ll \eps_0$, but we choose to use ``amplify'' and follow the existing nomenclature.}}]\label{lem:rws-code}
Let $\cC_0 \subseteq \B^n$ be an $\eps_0$-balanced code, and let 
$G = (V = [n],E)$ be a $D$-regular $\lambda$-spectral expander, and for an even $t \in \mathbb{N}$, let $\mathcal{W}_t = \set{w_1,\ldots,w_{\bar{n}}} \subseteq [n]^t$ be the set of walks of length $t$ on $G$, noting that $\bar{n}=n \cdot D^t$. Define
$\cC \subseteq \B^{\bar{n}}$ such that
\[
\mathcal{C} = \set{\mathrm{dsum}_{\mathcal{W}_t}(c_0) : c_0 \in \cC_0},
\]
where $y = \mathrm{dsum}_{\mathcal{W}_t}(x)$ at location $i \in [\bar{n}]$ is given by $\bigoplus_{j \in w_i} x_j$.

Then, $\cC$ is $\eps$-balanced, for
\[
\eps = \left( \eps_0 + 2\lambda \right)^{t/2}.
\]
\end{lemma}

For $\cC_0$, we will use the Justesen code. 
\begin{lemma}[\cite{Jus72}]\label{lem:jus}
There exist constants $R \in (0,1)$ and $\eps_0 \in (0,1)$ such that there exists an explicit family of codes $\set{\mathrm{Jus}_n}$ parameterized by block length $n$, with rate $R$ and $\eps_0$-balanced. 
Moreover, given $x \in \B^{k = Rn}$, $\mathrm{Jus}_{n}(x)$ can be computed in $\widetilde{O}(n)$.
\end{lemma}
\begin{proof}
The parameters of the codes follow from the original construction (and specifically, the lemma holds, say, with $R = \frac{1}{8}$ and $\eps_0 = \frac{37}{40}$), so we will just show that the code is efficiently computable. Given a message $x$, we first encode it with a full Reed--Solomon code of constant relative distance over a field $\F_q$ of characteristic $2$, where $q\log q = O(n)$. By \cref{lemma:fast}, this can be done in time $\widetilde{O}(q) = \widetilde{O}(n)$. Then,
we replace each Reed--Solomon symbol $p_{x}(\alpha)$,
for $\alpha \in \F_q$, with the binary representation of $(p(\alpha),\alpha \cdot p(\alpha))$. (In other words, we concatenate Reed--Solomon with the Wozencraft ensemble.) This takes $\widetilde{O}(q)$ time as well.
\end{proof}

\begin{corollary}\label{cor:biased}
There exist a constant $c > 1$, and an explicit family of balanced codes, such that for every $\bar{n} \in \mathbb{N}$ and any $\eps > 0$, $\cC \subseteq \F_2^{\bar{n}}$ is $\eps$-balanced of rate $R = \eps^{c}$, and given $x \in \F_2^{k = R\bar{n}}$, any $m$ bits of
$\cC(x)$ can be computed in time $\widetilde{O}(k)+m\log(1/\eps) \cdot \polylog(k)$.

Moreover, for every $k \in \mathbb{N}$ and any $\eps > 0$ there exists an explicit $\eps$-biased set over $\F_2^k$ generated by a function $\mathsf{SmallBias} \colon [\bar{n}] \rightarrow \B^k$ computable in time $(k + \log(1/\eps)) \cdot  \polylog(k)$.
\end{corollary}
\begin{proof}
Let $\cC_0 \colon \F_2^{k = \Theta(n)} \rightarrow \F_2^{n}$ be the $\eps_0$-balanced code guaranteed to us by \cref{lem:jus}, and let
$G = (V = [n],E)$ be the $D$-regular $\lambda$-spectral expander of \cref{cor:expander}, instantiated with $\lambda = \frac{1-\eps_0}{4}$ (so $D=D(\eps_0)$). Letting $\cC \colon \F_2^{k} \rightarrow \F_2^{\bar{n}}$ be the amplified code of \cref{lem:rws-code} set with \[t = \frac{2\log(1/\eps)}{\log\left(\frac{2}{1+\eps_0}\right)} = O(\log(1/\eps)),\] the lemma tells us that it is
$
( \eps_0 + 2\lambda )^{t/2} \le \eps
$
balanced. Given $x \in \F_2^n$ and $i \in [\bar{n}]$, computing $\cC(x)_i$ amounts to XORing $t$ coordinates of $\cC_0(x)$ determined by $i = (v,i_1,\ldots,i_t)$, which indexes a random walk over $G$.
Computing $\cC_0(x)$ takes $\widetilde{O}(n)$ time, and computing a length-$t$ walk over $G$ takes $t \cdot \log(1/\lambda) \cdot \polylog(n)$ time. The corollary then follows, observing that $\bar{n} = n \cdot D^t = n \cdot \poly(1/\eps)$.

For the ``Moreover'' part, recall that we can take the rows of the generator matrix of $\cC$ as our $\eps$-biased set $S$. Thus, for any
$i \in [\bar{n}]$, we can compute $\mathsf{SmallBias}(i)$ as follows: Compute the corresponding random walk on $G$, and then, for any $j \in [k]$, 
$\mathsf{SmallBias}(i)_j$ is obtained by XORing the bits of $\cC_0(e_j)$ indexed by the $i$-th random walk. Observing that each bit of $\cC_0(e_j)$ can be computed in time $\widetilde{O}(\log n)$,\footnote{Indeed, each coordinate of $\cC_0(e_j)$ is a bit in the encoding of $(\alpha^{j},\alpha^{j+1})$ for some $\alpha \in \F_q$, where $q\log q = O(n)$.} the runtime of $\mathsf{SmallBias}$ is
\[
t \cdot \log(1/\lambda) \cdot \polylog(n) + k \cdot \widetilde{O}(\log n) = (k + \log(1/\eps)) \cdot  \polylog(k), 
\]
recalling that $k = \Theta(n)$.
\end{proof}

\begin{remark}\label{remark:spielman}
\em
Instead of using Justesen's code from \cref{lem:jus} as our inner code $\mathcal{C}_0$, we can instead use the linear-time encodable code of Spielman \cite{Spi96}. While not stated as \emph{balanced} codes, but rather as constant relative distance codes, one can verify that the distance can also be bounded by above.
The construction is more involved than Justesen's one. 
However, in the logarithmic-cost RAM model, in which basic register operations over $O(\log n)$ bit registers count as a single time step, Spielman's code can be implemented in $O(n)$ time.
\end{remark}

\subsection{A Sampler from Bounded Independence}\label{sec:bounded-indep-sampler}

Recall that $X_1,\ldots,X_n \sim \Sigma^n$ is a $b$-wise independent distribution, if for every distinct $i_1,\ldots,i_b \in [n]$ it holds that $(X_{i_1},\ldots,X_{i_b}) = U_{\Sigma^b}$, the uniform distribution over $\Sigma^b$.

\begin{lemma}\label{lem:kwise}
For any prime power $n$, any $m \le n$, and any $b \le m$,  there exists an explicit $b$-wise independent generator 
$\mathsf{BI}_{b} \colon \B^{d} \rightarrow [n]^m$ with $d = O(b\log n)$\footnote{That is, the distribution formed by picking $z \sim U_d$ and outputting $\mathsf{BI}_{b}(z)$ is $b$-wise independent over $[n]^m$.} that satisfies the following.
\begin{enumerate}
    \item Given $z \in \B^d$, $\mathsf{BI}_{b}(z)$ is computable in time $\widetilde{O}(n)$.
    \item\label{it:repeat} For any $\theta > 0$ the following holds. With probability at least $1-2^{-\Omega(\theta b)}$ over $z \sim U_d$,  $\mathsf{BI}_{b}(z)$ has at least $m - (1+\theta)\frac{m^2}{n}$ distinct elements.
\end{enumerate}
\end{lemma}
\begin{proof}
We use the standard polynomial-based construction of bounded independence sample spaces. Concretely, 
given $z \in \F_n^{b}$, we interpret it as a polynomial $f_{z}(x) = \sum_{i=1}^{b}z_{i}x^{i-1}$, and we let $\mathsf{BI}_b(z)$ output $(f_{z}(\alpha_1),\ldots,f_{z}(\alpha_m))$, where the $\alpha_i$-s are distinct elements in $\F_{n}$. This gives us a $b$-wise generator over $[n]^m$ (the correctness can be found, e.g., in \cite[Chapter 3]{Vad12}). By \cref{arithmetic:it1} of \cref{lemma:fast}, 
$\mathsf{BI}_{b}(z)$ is computable in time $\widetilde{O}(n)$ (and this is true for any $b \le n$). The seed length $d$ is given by $O(b\log n)$.

Next, we argue that most samples contain mostly distinct elements. Towards this end, let $X_1,\ldots,X_m$ be our $b$-wise independent distribution $\mathsf{BI}_{b}(U_d)$, and let $Z_i$ denote the indicator random variable that is $1$ if and only if $X_i$ is a duplicate element (namely, there exists $j < i$ such that $X_i = X_j$). We are looking to bound $\sum_{i \in [m]}Z_i$ with high probability. This will follow from the fact that $X$ is $b$-wise independent.
\begin{claim}\label{claim:repeat2}
Assume that $t \le b/2$. Then, for any distinct $i_1,\ldots,i_t \in [m]$, it holds that $\Pr[Z_{i_1}=\ldots=Z_{i_t}=1] \le (m/n)^t$.
\end{claim}
\begin{proof}
Fix indices $j_1,\ldots,j_t$, where each $j_\ell < i_\ell$. The probability that $X_{i_{\ell}} = X_{j_{\ell}}$ for all $\ell \in [t]$ is at most $n^{-t}$, since this event depends on at most $2t \le b$ random variables. Union-bounding over all choices of $j$-s incurs a multiplicative factor of $
\prod_{\ell \in [t]}(i_{\ell}-1) \le m^{t},
$
so overall, 
$\Pr[Z_{i_1}=\ldots=Z_{i_t}=1] \le (m/n)^{t}$.
\end{proof}
Now, \cref{claim:repeat2} is sufficient to give us good tail bounds (see, e.g., \cite[Section 3]{HH15}\footnote{In the notation of \cite{HH15}, the distribution $Z$ is $(0,b)$-growth-bounded. The tail bound then follows from \cite[Theorem 3.2]{HH15}.}). In particular, denoting $\mu = \frac{m}{n}$ there exists a universal constant $c>0$ such that
\[
\Pr\left[ \sum_{i \in [m]}Z_i \ge (1+\theta)\mu m \right] \le 2^{-c\theta b},
\]
which implies \cref{it:repeat}.
\qedhere
\end{proof}

Towards introducing our sampler, we will need the following tail bound for $b$-wise independent random variables.
\begin{lemma}[e.g., \protect{\cite[Chapter 3]{Vad12}}]\label{lem:tail}
Let $X \sim \Sigma^m$ be a $b$-wise independent distribution, and fix some $\eps > 0$. Then, $X$ is also
a $(\delta,\eps)$ sampling distribution, where
$
\delta = \left( \frac{b}{2\eps\sqrt{m}} \right)^{b}.
$
\end{lemma}

While the error in \cref{it:repeat} above is small, it is not small enough for us to simply combine \cref{lem:tail,lem:kwise}, and we will need to do a mild error reduction. We do this via random walks on expanders and discarding repeating symbols, as was also done in \cite[Section 8]{Vad04}. This gives us the following bounded-independence based sampler.
\begin{lemma}\label{lem:k-sampler1}
For any positive integers $m \le n$, any $\delta_{\Gamma} \in (0,1)$, and any constant $\eta \in (0,1)$ such that $m \le \frac{\eta}{8}n$, there exists an explicit $(\delta_{\Gamma},\eps_{\Gamma} = 2\eta)$-averaging sampler 
$\Gamma \colon  \B^{d} \rightarrow [n]^{m}$ with $d = O\left( \frac{\log n}{\log m} \cdot \log\frac{1}{\delta_{\Gamma}} \right)$, that satisfies the following additional properties.
\begin{enumerate}
    \item Every output of $\Gamma$ contains distinct symbols of $[n]$, and
    \item Given $y \in \B^d$, $\Gamma(y)$ is computable in time  $\widetilde{O}(n)+\poly\left(\log\frac{1}{\delta_{\Gamma}} \cdot \frac{\log n}{\log m}\right)$.
\end{enumerate}
\end{lemma}
\begin{proof}
Set $b$ to be the smallest integer such that $b\log\frac{2\eta\sqrt{m'}}{b} \ge \log\frac{4}{\delta_{\Gamma}}$ and set $m' = (1+\eta)m$, $\theta = \eta/2$. Notice that $b = O\left(\frac{\log(1/\delta_{\Gamma})}{\log m}\right)$. Assume first that $n$ is a prime power, and 
let  \[\mathsf{BI}_{b} \colon \B^{d'} \rightarrow [n]^{m'}\] be the
$b$-wise independent generator guaranteed to us by \cref{lem:kwise}, with $d' = O(b\log n)$. By \cref{lem:tail}, $X = \mathsf{BI}_{b}(U_{d'})$ is a $(\delta_{\mathsf{b}},\eta)$ sampling distribution,
where 
\[
\delta_{\mathsf{b}} = \left( \frac{b}{2\eta\sqrt{m'}} \right)^{b}  \le  \frac{\delta_{\Gamma}}{4}.
\]
Also, we know from \cref{lem:kwise} that with probability at 
least $1 - 2^{-\Omega(\theta b)} \triangleq 1-p$, each sample from $X$ has 
at least $m' - (1+\theta)\frac{m'^2}{n} \ge m$ distinct symbols, using the fact that $\frac{n}{m} \ge \frac{(1+\theta)(1+\eta)^2}{\eta}$.
Conditioned on seeing at least $m$ distinct symbols, $X$ as a sampling distribution, when we remove the non-distinct elements, has confidence $\frac{\delta_{\Gamma}/4}{1-p} \le \frac{\delta_{\Gamma}}{2}$ and accuracy $2\eta$ (where the second $\eta$ comes from the fact that $\eta m$ symbols were removed).

Next, in order to improve the probability of sampling a good sequence to match the confidence, let $G = (V = \B^{d'},E)$ be the $D$-regular $\lambda$-spectral expander of \cref{cor:expander}, instantiated with $\lambda = p$, so $D \le p^{-c}$ for some universal constant $c$. 
Write $d =d'+\ell'$ for $\ell' = \ell \cdot \log D$, where $\ell = c_{\ell} \cdot \frac{\log(1/\delta_{\Gamma})}{b}$ for some constant $c_{\ell}$ soon to be determined. 
Given $y = (z,w) \in \B^{d'} \times [D]^{\ell}$, let $z = v_0,v_1,\ldots,v_{\ell}$ denote the corresponding random walk over $G$. 
Our sampler $\Gamma$, on input $y$, computes $\mathsf{BI}_{b}(v_i)$  for every $i \in [\ell]$ and outputs the first $m$ distinct
symbols of the first sequence with at least $m$ distinct symbols. 
If no such sequence was found, $\Gamma$ simply outputs $(1,\ldots,m)$ (in which case we say it \emph{failed}). By the expander hitting property (see, e.g., \cite[Chapter 4]{Vad12}), $\Gamma$ fails with probability at most
\[
(p+\lambda)^{\ell} = (2p)^{\ell} \le \frac{\delta_{\Gamma}}{2}
\]
over $y \sim U_{d}$, upon choosing the appropriate constant $c_{\ell} = c_{\ell}(\eta)$. We then have that $\Gamma(U_{d})$ is indeed a $(\delta_{\Gamma},2\eta)$ sampling distribution, that can be generated using a seed of length $d' + \ell' = O(\log(1/\gamma))$.  In terms of runtime, computing $v_1,\ldots,v_{\ell}$ can be done in time 
\[
\ell \cdot \log\frac{1}{p} \cdot \polylog(d') = \poly\left( \log\frac{1}{\delta_{\Gamma}} \cdot \frac{\log n}{\log m} \right),
\]
and computing the sequences themselves takes $\ell \cdot \widetilde{O}(n)$ time, and observe that $\ell = O(\log m)$.

Finally, we need to argue that we can also handle the case where $n$ is not a prime power. One option to handle arbitrary $n$-s is to resort to almost $b$-wise independence. 
Specifically, we can start with a $\gamma$-biased distribution, with a small enough $\gamma$, over $n_{\mathsf{b}} = \lceil \tau^{-1}\log n \rceil n$ bits, and map each consecutive  $\lceil \tau^{-1}\log n \rceil$ bits to $[n]$, for a small enough $\tau$, by simply taking the corresponding integer modulo $n$. We skip the details (to see the corresponding sampling property, see \cite{XZ24}), and note that the result that uses the bounded-independence sampler, \cref{thm:main-nonrecursive}, only needs to invoke the sampler with $n$ being a prime power.\footnote{In some more detail, in the proof of \cref{thm:block-source} we use the $b$-wise distribution over $[n]$ in order to sub-sample from $X \sim \B^n$.
In the context of \cref{thm:main-nonrecursive}, we can assume that $X$ has entropy rate $1-\alpha$ for an arbitrarily small constant $\alpha>0$, and the goal is to retain its high entropy rate when sub-sampling. Letting $p$ be the largest prime smaller than or equal to $n$ (which satisfies $p\geq n/2$ by Bertrand's postulate), we have that $X'=X_{[1,p]}$ has entropy rate at least $1-2\alpha$, and one can verify that sampling from $X'$ works equally well, up to negligible loss in parameters. Moreover, we can find $p$ in time $\Otilde(n)$ using, e.g., the deterministic AKS primality test~\cite{AKS04} that runs in time $\polylog n$ applied to each integer in $[n/2,n]$.
}
\end{proof}

We will need to somewhat extend \cref{lem:k-sampler1} and use the simple, yet crucial, property of our bounded independence sampling: A subset of the coordinates of a $b$-wise independent distribution with distinct samples is itself a $b$-wise  independent distribution with distinct samples.\footnote{We note that the distinct-samples sampler given in \cite{Vad04}, an instantiation of which we use in \cref{lem:extendedSampler}, does not seem to enjoy a similar property. The advantage of \cref{lem:extendedSampler} over \cref{lem:k-sampler2} that will appear soon, is that it has better seed length.}
Thus, if we wish to sample multiple times, say using $m_1 \le \ldots \le m_t < n$ samples, we can use \emph{one sample} from a sampler that outputs $m_t$ coordinates, and truncate accordingly to create the other samples. 
\begin{lemma}\label{lem:k-sampler2}
For any positive integers $n$ and $m_1 < \ldots < m_t \le n$, any $\delta \in (0,1)$ and any constant $\eps$ such that $m_t \le \frac{\eps}{16}n$, there exists an explicit function $\Gamma \colon \B^d \rightarrow [n]^{m_t}$ with $d = O\left( \frac{\log n}{\log m_1} \cdot \log\frac{1}{\delta} \right)$ that satisfies the following.
\begin{enumerate}
    \item For any $i \in [t]$, the function $\Gamma_i$, that on input $y \in \B^d$ outputs $\Gamma(y)|_{[1,m_i]}$, is a $(\delta,\eps)$-averaging sampler, and each sample contains distinct symbols.
    \item On input $y \in \B^d$, $\Gamma(y)$ can be computed in time $\widetilde{O}(n) + \poly(\log(1/\delta),(\log n)/(\log m_1))$.
\end{enumerate}
\end{lemma}
\begin{proof}[Proof Sketch]
Let $\Gamma \colon \B^d \rightarrow [n]^{m_t}$ be the 
$(\delta,\eps)$-averaging sampler of \cref{lem:k-sampler1}, set with 
$\delta_\Gamma = \delta$ and $\eta = \eps/2$.
Fix some $i \in [t]$, and let $\Gamma_i$ be as described in the lemma's statement. Inspecting the proof of 
\cref{lem:k-sampler1}, and using the same notation, we see that the output of $\Gamma_i$ can be obtained by truncating the output of $\mathsf{BI}_b$ to the first $m'_i = (1+\eta)m_i$ bits and using the same expander random walk in order to get $m_i$ distinct symbols with high probability. Note that:
\begin{enumerate}
    \item (sampling) $X_{[1,m'_i]}$ is a $b$-wise independent distribution, whose sampling properties are determined by $m'_i$. We then need to make sure that the smallest $m'_i$ is large enough, by setting $b = O\left(\frac{\log(1/\delta_{\Gamma})}{\log m_1}\right)$, as in the proof of \cref{lem:k-sampler1}. 
    \item (distinctness)
    In order for the probability that $X_{[1,m'_i]}$ has at least $m_i$ distinct symbols to be large enough (recall that $X$ itself is over $m'_t$ symbols), $m_i \le m_t$ needs to be small enough compared to $n$. And indeed, we take $m_i \le \frac{\eta}{8}n = \frac{\eps}{16}n$.
\end{enumerate}
Once the two properties above are guaranteed, we can amplify the probability to sample a string with sufficiently many distinct symbols via expander random walks, exactly as in the proof of \cref{lem:k-sampler1}.
\end{proof}

\subsection{Nearly-Linear Time Condensers}\label{sec:fast-condensers}

In this section we argue that modern constructions of condensers with nearly-optimal parameters~\cite{GUV09,KT22} are computable in nearly-linear time, possibly after a reasonable one-time preprocessing step.
We will repeatedly use the fact that these condensers allow us to transform, in time $\Otilde(n)$ and using nearly-optimal seed length $O(\log(n/\eps))$, an arbitrary input $(n,k)$-source $X$ into an output $\eps$-close to a source $X'$ of length $m\approx k$ and min-entropy rate $1-\alpha$, for any constant $\alpha\in(0,1)$ of our choice.
We provide formal statements below, without being explicit about the precise relationship between the parameter $\alpha$ and the various constants in the seed length and entropy requirement for the sake of readability.
More precise control over the constants in these condensers can be found, for example, in~\cite[Theorems 1 and 2]{KT22}.

\subsubsection{The Lossless KT Condenser}\label{sec:kt-cond}

We first give the lossless KT condenser based on multiplicity codes over $\F_q$, due to Kalev and Ta-Shma \cite{KT22}.
This condenser works when the min-entropy requirement $k=\Omega(\log^2(n/\eps))$.
Then, we discuss how it can be converted to a condenser over bits with only a negligible loss in parameters.

\begin{theorem}[\protect{the lossless KT condenser over $\F_q$ \cite[adapted]{KT22}}]\label{thm:F_q-ktcond}
For any constant $\alpha \in (0,1)$ there exist constants $C_\alpha,C'_\alpha>0$ such that the following holds
for every $n \in \mathbb{N}$, $\eps>0$,
$k \ge C_\alpha\log^{2}(n/\eps)$, and a prime power $q=p^r$ with $p\geq n$ satisfying $ \frac{C'_\alpha}{2} \log(n/\eps)\leq \log q \leq  C'_\alpha \log(n/\eps)$. 
There exists a strong $(k,k'=k+\log q,\eps)$-condenser
\[
    \mathsf{KTCond}' \colon \F_q^{n'} \times \F_q \rightarrow \F_q^{m'}
\]
where $n' = \left\lceil \frac{n}{\lfloor \log q\rfloor} \right\rceil$ and $m' \leq (1+\alpha)\frac{k}{\log q}$.
Moreover, given $p$, an irreducible polynomial over $\F_p$ of degree $r$, $x \in \F_q^n$, and $y \in \F_q$, the output $\mathsf{KTCond}'(x,y)$ can be computed in $\widetilde{O}(n'\log q)=\Otilde(n)$ time.
\end{theorem}
\begin{proof}
We only need to establish the construction's runtime. 
Given $x \in \F_q^{n'}$ and $y \in \F_q$, 
interpret $x$ as a polynomial $f\in \F_q[X]$ of degree at most $n'-1$.  
The output  $\mathsf{KTCond}(x,y)$ is the sequence of derivatives
\[
    \left(f(y),f'(y),\ldots,f^{(m'-1)}(y) \right).
\]
By \cref{lemma:fast}, computing the derivatives takes $\widetilde{O}(n') \cdot \log q = \widetilde{O}(n'\log q)=\Otilde(n)$ time  since we are also given $p$ and an irreducible polynomial over $\F_p$ of degree $r$. 
The rest of the auxiliary operations 
are negligible compared to computing the derivatives.
\end{proof}

\begin{remark}[the KT condenser preprocessing step]\label{remark:KT-preproc}
    \em
    The KT condenser must be instantiated with an appropriate field size $q=\poly(n/\eps)$ and requires performing operations over $\F_q$.
    Hence, as stated in \cref{thm:F_q-ktcond}, we need to know the characteristic $p$ of $\F_q$ and an irreducible polynomial over $\F_p$ of the appropriate degree.
    Since these objects only need to be generated once for a given set of extractor parameters $(n,k,\eps,m)$, we view this is as a one-time preprocessing step.
    Our choice of $q$ influences the complexity of constructing a model of $\F_q$.
    For example, if we aim for prime field size, then we need to generate a prime $q=\poly(n/\eps)$.
    We know how to do that in randomized time $\polylog(n/\eps)$ (e.g., see~\cite[Section 9.4]{Sho05}).
    We believe that this is already reasonable when seen as a one-time preprocessing step, because generating large primes is a well-studied problem and practical randomized algorithms exist.
    
    Nevertheless, we can do even better and avoid the preprocessing step if $\eps$ is not very small by exploiting the fact that the KT condenser works with any prime power $q=p^r=\poly(n/\eps)$ for any prime $p\geq n$~\cite[Theorem 3]{KT22}.
    First, a prime $p\in [n,2n]$ is guaranteed to exist by Bertrand's postulate, and we can find it deterministically in time $\Otilde(n)$.\footnote{The procedure that sequentially tests the primality of all integers in $[n,2n]$ is already sufficiently quick. Each primality test takes time $\polylog(n)$ using, for example, the AKS primality test~\cite{AKS04}, and so the overall complexity is $\Otilde(n)$.}
    Second, a deterministic algorithm of Shoup~\cite[Theorem 3.2]{Sho90b} finds an irreducible polynomial of degree $r$ over $\F_p$ in time $\Otilde(\sqrt{p}\cdot r^{4+\delta})$ for any constant $\delta>0$.
    Since $p\leq 2n$ and $r=\log_p(\poly(n/\eps))=O(\log(n/\eps))$, we conclude that Shoup's algorithm runs in time $\Otilde(n)$ provided that $\eps\geq 2^{-C n^{\gamma}}$ for any constant $\gamma<1/8$ and a sufficiently large constant $C>0$.
    In particular, it suffices to have, say, $\eps\geq 2^{-C n^{0.1}}$.

    In sum, the preprocessing step for the KT condenser requires either (1) randomized $\polylog(n/\eps)$ time, or (2) deterministic $\Otilde(n+\sqrt{n}\cdot \log^{4+\delta}(1/\eps))$ time.
    In particular, this leads to the distinction between \cref{it:case1,it:case2} in \cref{thm:rec-ext-intro}.
\end{remark}

The condenser from \cref{thm:F_q-ktcond} receives and outputs vectors over $\F_q$.
We would like to have a version of this condenser that works with vectors over bits.
Towards that end, for completeness, we formally state and prove a (standard) transformation.
\begin{lemma}\label{lem:q-to-2}
    Suppose that $\Cond_0\colon \F_q^{n'}\times \F_q\to\F_q^{m'}$ is a strong $(k,k',\eps)$-condenser.
    Then, there exists $\Cond\colon \bits^n\times \bits^\ell \to\bits^m$ with $n=n'\cdot \lfloor \log q\rfloor$, $\ell=\lfloor \log q\rfloor$, and $m=m'\cdot\lceil\log q\rceil$ that is a strong $(k+\ell,k',3\sqrt{\eps})$-condenser.
    Furthermore, if $\Cond_0$ is computable in time $T$ then $\Cond$ is computable in time $T+\Otilde(n)$.
\end{lemma}
\begin{proof}
    The condenser $\Cond$ works as follows.
    Given an input $x \in \B^n$, we interpret it as an element of $\F_q^{n'}$ by mapping each consecutive block of $\lfloor \log q \rfloor$ bits to an $\F_q$-element and discarding up to $\lceil \frac{n}{\lfloor \log q\rfloor}\rceil\cdot \lfloor \log q\rfloor - n \leq \lfloor \log q\rfloor$ bits at the end of $x$.
    Denote this mapping by $\psi\colon \B^n \to \F_q^{n'}$.
    Given a seed $y\in\B^\ell$, we interpret it as an $\F_q$-element, since $\ell=\lfloor \log q\rfloor$.
    Denoting this mapping by $\phi \colon \F_2^{\ell} \rightarrow \F_q$, we have that $\minH(\phi(U_{\ell})) \ge \ell$.

    Equipped with an $\F_q^{n'}$-element $x$, and an $\F_q$-element $y$, we are ready to compute $\Cond_0(x,y)\in\F_q^{m'}$. 
    Mapping the output into bits is done similarly, by writing the binary encoding of each $\F_q$-element using $\lceil \log q \rceil$ bits.
    There are a few possible losses along the way:
\begin{enumerate}
    \item We get an $(n,k)$ source $X$, but feed the condenser a source $X' \sim \F_q^{n'}$ with entropy at least $k-\lfloor \log q \rfloor=k-\ell$.
    \item We do not use a uniform seed, but rather a seed that lacks at most $1$ bit of entropy.
    \item Mapping $\Cond_0(\psi(x),\phi(y))$ to the binary $\Cond(x,y)$ may increase the output by $m'$ non-entropic bits. This simply increases the output length slightly, but not the output entropy.
\end{enumerate}
To handle (1), we simply increase the entropy of the source. To handle (2), we prove that strong condensers work with entropy-deficient random seeds, as long as the error is good enough.
\begin{claim}\label{claim:def-seed}
Let $\Cond \colon \F_q^n \times \F_q \rightarrow \F_q^m$ be a strong $(k,k',\eps)$ condenser, and let $Y$ be a random variable with min-entropy $\log(q) - g$. Then, $(Y, \Cond(X,Y))$ is $\eps'$-close to a random variable with min-entropy $k'$, where $\eps' = (2^{g}+1)\sqrt{\eps}$.
\end{claim}
\begin{proof}
Let $X$ be a random variable over $\F_q^n$ with min-entropy $k$. By an averaging argument, we know that there exists a set $\mathbf{B} \subseteq \F_q$ of size $|\mathbf{B}| \le \sqrt{\eps} \cdot q$ such that for any $y \notin \mathbf{B}$ it holds that $\Cond(X,y)$ is $\sqrt{\eps}$-close to a random variable with min-entropy $k'-\log q$. But now,
\[
\Pr[Y' \in \mathbf{B}] = \sum_{y \in \mathbf{B}}\Pr[Y=y] \le |\mathbf{B}| \cdot 2^{-(\log(q)-g)} \leq \sqrt{\eps} \cdot 2^{g}. \qedhere
\]
\end{proof}
In our case we invoke \cref{claim:def-seed} with $g=1$, so we can simply set the new error parameter $\eps'$ to be $\eps' =  \eps^2/3$, where $\eps$ is the designated distance to high min-entropy. A bit more formally, guaranteeing that $(U_d, \Cond(X,U_d))$ is $\eps'$-close to entropy $k'$ implies that $(\phi(U_d), \Cond(X,\phi(U_d)))$ is $\eps$-close to entropy $k'$ as well. 
Finally, for the running time claim note that computing the mappings $\psi$ and $\phi$ and converting the output in $\F_q^{m'}$ to $\B^m$ can be done in time $\Otilde(n)$.
\end{proof}

Combining \cref{thm:F_q-ktcond} and \cref{lem:q-to-2} yields the following.
\begin{theorem}[the KT condenser over bits \cite{KT22}]\label{thm:ktcond}
    For any constant $\alpha \in (0,1)$ there exist constants $C_\alpha,C'_\alpha>0$ such that the following holds
    for every $n \in \mathbb{N}$, $\eps>0$,
    and $k \ge C_\alpha\log^{2}(n/\eps)$. 
    There exists a strong $(k,k',\eps)$-condenser
    \[
        \mathsf{KTCond} \colon \B^n \times \B^\ell \rightarrow \B^m
    \]
    where $\ell \leq C'_\alpha\log(n/\eps)$, $m \leq (1+\alpha)k$, and $k'=k\geq \ell + (1-\alpha)m$.
    Moreover,
    \begin{itemize}
        \item Given $x \in \B^n$ and $y \in \B^{\ell}$, a prime power $q=p^r$ in $[2^\ell,2n\cdot 2^{\ell}]$ with $p\in [n,2n]$ and an irreducible polynomial over $\F_p$ of degree $r$, the output $\mathsf{KTCond}(x,y)$ can be computed in $\widetilde{O}(n)$ time.

        \item Given only $x \in \B^n$ and $y \in \B^{\ell}$, $\mathsf{KTCond}(x,y)$ can be computed in $\widetilde{O}(n+\sqrt{n}\cdot \log(1/\eps)^5)$ time.
        Note that this is $\Otilde(n)$ when $\eps\geq 2^{-C n^{0.1}}$.
    \end{itemize}

    In particular, if $\eps' = \sqrt{\eps}$ and $C_\alpha$ is chosen large enough compared to $C'_\alpha$, then for all $(n,k)$-sources $X$ and a $(1-\eps')$-fraction of seeds $y$ it holds that $\KTCond(X,y)\approx_{\eps'} Z_y$, where $Z_y$ is an $(m,k'-\ell\geq (1-\alpha)m)$-source.
\end{theorem}
\begin{proof}
    Fix a target $\alpha\in (0,1)$, and let $\beta\in(0,1)$ be a constant to be set appropriately small depending on $\alpha$.
    We invoke \cref{thm:F_q-ktcond} with $\beta$ in place of $\alpha$
    and $q$ a prime such that $\frac{C'_\beta}{2}\log(n/\eps)\leq \log q \leq C'_\beta\log(n/\eps)$
    to get $\KTCond'\colon\F_q^{n'}\times \F_q\to\F_q^{m'}$, a strong $(k-\log q,k'=k,\eps^2/3)$-condenser with $n' = \left\lceil \frac{n}{\lfloor \log q\rfloor} \right\rceil$ and $m' \leq (1+\beta)\frac{k}{\log q}$.
    Then, we apply \cref{lem:q-to-2} to $\KTCond'$. This yields our $\KTCond\colon \B^n\times \B^\ell\to\B^m$, a strong $(k,k'=k,\eps)$-condenser with $\ell=\lfloor \log q\rfloor = \Theta_\beta(\log(n/\eps))$ and
    \begin{equation*}
        m = m'\cdot \lceil\log q\rceil \leq (1+\beta)k+m' \leq (1+\beta)k + (1+\beta)k/\ell \leq (1+2\beta)k\leq (1+\alpha)k,
    \end{equation*}
    provided that $\beta\leq \alpha/2$.
    Furthermore, if $C_\beta$ is set sufficiently larger than $C'_\beta$, then $k'-\ell = k-\ell \geq (1-\beta)k$, and so
    \begin{equation*}
        \frac{k'-\ell}{m}\geq \frac{(1-\beta)k}{(1+2\beta)k}\geq 1-\alpha,
    \end{equation*}
    provided that $\beta>0$ is sufficiently smaller than $\alpha$.
    The first part of the theorem statement follows if we set the new $C_\alpha$ and $C'_\alpha$ to be $C_\beta$ and $C'_\beta$, respectively.
    For the running time, note that $\KTCond'$ is computable in time $\Otilde(n)$ given $p$ and an irreducible polynomial over $\F_p$ of degree $r$, and the transformation in \cref{lem:q-to-2} still runs in time $\Otilde(n)$.
    The remaining claim about the running time given only $x$ and $y$ was discussed in \cref{remark:KT-preproc}.
    
    To see the ``In particular'' part of the theorem statement, fix an $(n,k)$-source $X$ and note that $(Y,\KTCond(X,Y))\approx_\eps (Y, Z)$ for some $Z$ such that $\minH(Y, Z)\geq k'$.
    Let $Z_y = (Z|Y=y)$.
    Then, an averaging argument gives that for a $(1-\sqrt{\eps})$-fraction of seeds $y$ we have $\KTCond(X,y)\approx_{\sqrt{\eps}} Z_y$.
    Since $Y$ is uniformly random over $\B^\ell$, we get that $\minH(Z_y)\geq k'-\ell$, as desired.
\end{proof}

\subsubsection{The Lossy RS Condenser}

The downside of \cref{thm:ktcond} is that it requires the entropy in the source to be $\Omega(\log^{2}(n/\eps))$, instead of the optimal $\Omega(\log(n/\eps))$.
Instead, we can use a lossy condenser\footnote{Our extractor will lose a small constant fraction of the entropy, so losing a small constant fraction of the entropy in the condensing step will not make much difference.} based on Reed--Solomon codes. Unfortunately, this comes at the expense of computing a primitive element of a field of size $\poly(n/\eps)$, which we do not know how to do in nearly-linear time in $n$ for arbitrary $\eps$-s. 
As in \cref{sec:kt-cond}, we consider it a one-time preprocessing step, as it does not depend on the inputs to the condenser.
Luckily, we have freedom in choosing the field size, as long as it is large enough.
Therefore, if we choose the field size to be prime (as opposed to a power of $2$ as in~\cite{GUV09}), then we can implement this one-time preprocessing step in time $\polylog(n/\eps)$ using randomness, as discussed in \cref{remark:preproc3}.

We first state the lossy condenser over a prime field $\F_q$ and argue about its running time.
This condenser can be converted to a condenser over bits with only a negligible loss in parameters as in \cref{sec:kt-cond}.

\begin{theorem}[\protect{the lossy RS condenser over $\F_q$ \cite[adapted]{GUV09}}]\label{thm:F_q-rscond}
For any constant $\alpha \in (0,1)$ there exist constants $C_\alpha,C'_\alpha>0$ such that the following holds for every $n\in\N$, $\eps>0$, $k\geq C_\alpha \log\left(n/\eps\right)$, and $q$ a prime satisfying $\frac{C'_\alpha}{2} \log(n/\eps)\leq \log q \leq C'_\alpha \log(n/\eps)$.
There exists a strong $(k,k',\eps)$-condenser
\[
    \mathsf{RSCond}' \colon \F_q^{n'} \times \F_q \rightarrow \F_q^{m'}
\]
where $n' = \left\lceil \frac{n}{\lfloor \log q\rfloor} \right\rceil$, $m' = \lceil k/\log q\rceil$, and $k' \ge (1-\alpha)k$. 
Moreover, given $x \in \F_q^{n'}$, $y \in \F_q$, 
and a primitive element $\zeta$ for $\F_q$,
the output 
$\mathsf{RSCond}'(x,y)$ can be computed in time $\widetilde{O}(n)$.

\end{theorem}
\begin{proof}
    Given $x \in \F_q^{n'}$ and $y \in \F_q$, similarly to \cref{thm:ktcond}, interpret $x$ as a univariate polynomial $f\in\F_q[X]$ of degree at most $n'-1$.
    Let $\zeta$ be the primitive element of $\F_q$ given to us.
    The output $\RSCond'(x,y)$ is the sequence of evaluations
    \[
        \left(f(y),f(\zeta y),\ldots,f(\zeta^{m'-1}y) \right) \in \F_q^{m'}.
    \]
    The correctness proof, as well as the exact choice of parameters, are given
    in \cite[Section 6.1]{GUV09}.
    We focus on bounding the runtime.
    Computing the evaluation points $y,\zeta y,\ldots,\zeta^{m'-1}y$ can be done naively in time $m' \cdot M_{q}^{\mathsf{b}}(1) = \widetilde{O}(n'\log q)=\Otilde(n)$. Then, using \cref{lemma:fast}, the evaluation can be done in time $\widetilde{O}(n') \cdot \log q = \widetilde{O}(n'\log q)=\Otilde(n)$ as well.
\end{proof}

\begin{remark}[complexity of the RS condenser preprocessing]\label{remark:preproc3}
\em
We now discuss the complexity of the preprocessing step required by the RS condenser, 
which corresponds to finding a primitive element of a field $\F_q$ for a prime $q\leq \poly(n/\eps)$.
We do not know any algorithms for finding primitive elements of $\F_q$ running in time $\polylog(q)$, unless we have access to the prime factorization of $q-1$.\footnote{If we have access to the prime factorization of $q-1$, then we can find a primitive element of $\F_q$ in time $\polylog(q)$ using randomness by repeatedly sampling a uniformly random element $\alpha$ of $\F_q$ and checking whether it is primitive by seeing whether $\alpha^{\frac{q-1}{p}}\neq 1$ for every prime factor $p$ of $q-1$. An alternative algorithm is analyzed in~\cite[Section 11.1]{Sho05}.}
However, the RS condenser from \cite{GUV09} can be instantiated with \emph{any} field $\F_q$ of appropriately large order $q=\poly(n/\eps)$. We can leverage this to speed up the preprocessing considerably by moving away from the fields of characteristic $2$ used in~\cite{GUV09}.
More precisely, we can focus on prime $q$ and follow an approach used in cryptography for generating public parameters (prime/primitive element pairs $(q,g)$) for cryptographic schemes based on discrete logarithms over $\Z^*_q$, outlined in Shoup's excellent book~\cite[Section 11.1]{Sho05}.
This leads to a randomized algorithm running in time $\polylog(n/\eps)$ and succeeding with high probability, that only needs to be executed once for a given set of extractor parameters $(n,m,k,\eps)$.
More precisely, we can first efficiently sample the prime factorization of a uniformly random number $r$ in the set $\{1,\dots,L\}$ for an appropriate upper bound $L=\poly(n/\eps)$, following approaches of Bach~\cite{Bac88} or Kalai~\cite{Kal03} (as discussed in~\cite[Section 9.6]{Sho05}). Then, we check whether $r\geq L/2$ and $q=r+1$ is prime, and repeat if not.
Since roughly a $\frac{1}{\polylog(n/\eps)}$ fraction of such $q$'s is prime, we will succeed with high probability after $\polylog(n/\eps)$ trials.
After the process above we hold a suitably large prime $q$ and the prime factorization of $q-1$, which allows us to find a primitive element of $\F_q$ in randomized time $\polylog(q)=\polylog(n/\eps)$.
\end{remark}

Finally, combining \cref{thm:F_q-rscond} and \cref{lem:q-to-2} allows us to get a version of the ``prime $q$'' RS condenser over bits.
\begin{theorem}[the lossy RS condenser over bits, \cite{GUV09}]\label{thm:rscond}
    For any constant $\alpha\in(0,1)$ there exist constants $C_\alpha,C'_\alpha>0$ such that the following holds for every $n\in\N$, $\eps>0$, and $k\geq C_\alpha\log(n/\eps)$.
    There exists a strong $(k,k',\eps)$-condenser
    \begin{equation*}
        \RSCond\colon \bits^n\times \bits^\ell\to\bits^m
    \end{equation*}
    where $\ell \leq C'_\alpha \log(n/\eps)$, $k\leq m\leq (1+\alpha)k$, and $k'\geq (1-\alpha)m$.
    Moreover, given $x\in\bits^n$, $y\in \bits^\ell$, and a primitive element $\zeta$ for $\F_q$ with $q$ a prime in $[2^\ell, 2^{\ell+1}]$, the output $\RSCond(x,y)$ can be computed in time $\Otilde(n)$.

    In particular, if $\eps' = \sqrt{\eps}$ and $C_\alpha>0$ is chosen large enough compared to $C'_\alpha$, then for all $(n,k)$-sources $X$ and a $(1-\eps')$-fraction of seeds $y$ it holds that $\RSCond(X,y)\approx_{\eps'} Z_y$, where $Z_y$ is an $(m,k'-\ell\geq (1-2\alpha)m)$-source.
\end{theorem}
\begin{proof}
    We argue about the output length and the output min-entropy rate. The rest is analogous to the proof of \cref{thm:ktcond}.
    Fix a target $\alpha\in(0,1)$. We invoke \cref{thm:F_q-rscond} with an appropriately small constant $\beta\in(0,1)$ in place of $\alpha$.
    Regarding the output length $m$, note that $m= m' \cdot\lceil \log q\rceil \geq \frac{k}{\log q}\cdot \log q =k$, and that
    \begin{equation*}
        m= m' \lceil \log q\rceil\leq m'\log q+m' \leq k+\log q + m' \leq k + (\ell + 1) + (k/\ell+1) = \left(1+\frac{1}{\ell}+ \frac{\ell+2}{k}\right)k.
    \end{equation*}
    The quantity on the right hand side can be made at most $(1+\beta)k\leq (1+\alpha)k$ by setting $\beta<\alpha$ and $C_\beta$ to be sufficiently larger than $C'_\beta$, and by the lower bound on $\ell=\lfloor \log q\rfloor$ in \cref{thm:F_q-rscond}.
    Regarding the output entropy $k'$, since \cref{lem:q-to-2} implies a penalty of $\ell$ bits in the input min-entropy, from \cref{thm:F_q-rscond} instantiated with $\beta$ and \cref{lem:q-to-2} we get the guarantee that
    \begin{equation*}
        k' \geq (1-\beta)(k-\ell) \geq (1-\beta)k-\ell \geq (1-2\beta)k \geq \frac{(1-2\beta)m}{1+\beta} \geq (1-\alpha)m
    \end{equation*}
    provided that $C_\beta$ is sufficiently larger than $C'_\beta$ and that $\beta\leq \alpha/3$.
    
    The ``In particular'' part of the theorem statement follows analogously to that of \cref{thm:ktcond}, since $m\geq k$ and we can assume that $\ell/k\leq \alpha$ by setting $C_\alpha$ to be sufficiently large compared to $C'_\alpha$.
\end{proof}

\subsubsection{Towards removing preprocessing?}\label{sec:remove-preproc}

We discuss an approach, suggested by Jesse Goodman, towards removing the preprocessing steps required for the KT and RS condensers, or at least expanding the range of parameters for which preprocessing is not required, and some barriers we face when trying to fully implement this strategy.

\paragraph{KT condenser.}
As discussed in \cref{remark:KT-preproc}, the KT condenser in \cref{thm:ktcond} requires a one-time preprocessing step independent of the source and the seed whenever, roughly, $\eps\leq  2^{-C n^{0.12}}$.
If this holds, a sufficient preprocessing consists of generating a prime $q=\poly(n/\eps)$, which can be done in randomized time $\polylog(n/\eps)$.
This raises the following possibility: perhaps we can derandomize this algorithm so that it only uses $O(\log(n/\eps))$ bits of randomness without hurting the running time too much.
If this is the case, then the randomness used by the algorithm can be absorbed into the seed, and we get a condenser running in time $\Otilde(n+\polylog(n/\eps))$, without preprocessing.
When $\eps$ is not too small (depending on the exponent of the $\polylog(n/\eps)$ term), this is $\Otilde(n)$.

We can realize this approach by combining the randomized prime generation algorithm with an averaging sampler.
A similar strategy was used in~\cite{BIW06}.
We sketch how this can be done.
Let $N=\poly(n/\eps)$.
The randomized prime generation algorithm independently samples $O(\log N)$ uniformly random integers $q\in [N/2,N]$, and checks whether at least one of these numbers is prime.
Primality can be tested in deterministic time $\Otilde(\log^6 N)$ using an improved variant of the AKS primality test~\cite{Len02}, and, by the prime number theorem, with high enough probability there will be at least one prime number among the samples.
One way to derandomize this algorithm is by using a $(\gamma,\theta)$-averaging sampler $\Samp\colon \bits^r\to[N]^m$ with accuracy $\theta=\Theta(1/\log N)=\Theta(1/\log(n/\eps))$ and confidence $\gamma=\eps$ to generate $m$ candidate primes, and then run the AKS primality test on each candidate.
Denoting by $T$ the runtime of $\Samp$, this procedure runs in time $\Otilde(T+m\cdot \log^6(n/\eps))$, uses $r$ bits of randomness, and generates the desired prime except with probability at most $\eps$.

With some hindsight, it turns out that this running time is not enough to beat the runtime of the deterministic preprocessing from \cref{remark:KT-preproc}, even ignoring $T$ and using a sampler with optimal sample complexity $m$. 
However, it is possible to improve on the runtime by replacing the AKS primality test with a faster probabilistic test, such as Miller-Rabin~\cite[Section 10.2]{Sho05}, and use a second sampling round to derandomize these probabilistic tests.
More precisely, one iteration of the Rabin-Miller algorithm for testing primality of an integer $q$ samples $\alpha$ uniformly at random from $[q]$ and checks whether $\alpha$ passes a certain test in time $\Otilde(\log^2 q)$.\footnote{Shoup~\cite[Section 10.2]{Sho05} states that the test is computable in time $\Otilde(\log^3 q)$ via repeated squaring, but more sophisticated techniques yield the $\Otilde(\log^2 q)$ bound~\cite{Nar14}.}
When $q$ is prime all $\alpha$-s pass the test, while for $q$ composite $\alpha$ fails the test with probability at least $3/4$.
We use a $(\gamma',\theta')$-averaging sampler $\Samp'\colon \bits^{r'}\to[N]^{m'}$ with $\theta'=1/4$ and $\gamma'=\eps/m$ to generate $\alpha_1,\dots,\alpha_{m'}$, which will be used to run Miller-Rabin tests on the prime candidates $q_1,\dots,q_m$ generated by the first sampler.\footnote{Note that $\Samp'$ outputs samples from $[N]$, while the guarantees for the Miller-Rabin test on $q$ hold for $\alpha$ uniformly random over $[q]$. But since each $q_i\in[N/2,N]$, taking the test to be $\alpha\Mod q_i$ for $\alpha$ uniformly random over $[N]$ only decreases the failure probability from at least $3/4$ to at least $1/2$.}
Note that for any fixed composite $q_i$, the probability that all Miller-Rabin tests $\alpha_1\Mod q_i,\dots,\alpha_{m'}\Mod q_i$ pass is at most $\eps/m$.
By a union bound over $q_1,\dots,q_m$, the probability that this holds for at least one such $q_i$ is at most $m\cdot \eps/m=\eps$.
Therefore, if $\Samp$ and $\Samp'$ run in time $T$ and $T'$, respectively, this procedure runs in time $\Otilde(T+T'+m\cdot m'\cdot \log^2(n/\eps))$, uses $r+r'$ bits of randomness, and generates the desired prime except with probability at most $2\eps$.

We already know from \cref{remark:KT-preproc} that the KT condenser does not require preprocessing to run in $\Otilde(n)$ time when $\eps>2^{-C n^{0.12}}$, so it is relevant to explicitly work out the best possible improvement afforded by the approach above, assuming we are aiming for $\Otilde(n)$ runtime.
We can instantiate $\Samp$ and $\Samp'$ with the nearly-optimal efficient averaging samplers of Xun and Zuckerman~\cite[Theorem 2]{XZ24}.
Under the choices of $N$, $(\theta,\gamma)$, and $(\theta',\gamma')$ above, for any constant $\delta>0$, we can instantiate $\Samp$ with randomness complexity $r=O_\delta(\log(n/\eps))$ and sample complexity $m=O(\log^{3+\delta}(n/\eps))$, and $\Samp'$ with randomness complexity $r'=O_\delta(\log(n/\eps))$ and sample complexity $m'=O(\log^{1+\delta}(n/\eps))$.

Therefore, using these samplers,  for any constant $\delta>0$, the prime-generating procedure above uses $O_\delta(\log(n/\eps))$ bits of randomness, which can be absorbed into the seed of the condenser, runs in time $\Otilde(T+T'+\log^6(n/\eps))$, with $T$ and $T'$ the runtimes of $\Samp$ and $\Samp'$, and fails with probability at most $2\eps$, which can be absorbed into the error of the condenser.
One can verify that $T$ and $T'$ are $O(\log^6(n/\eps))$ in this regime.
Therefore, the time bound above is $\Otilde(n)$ when $\eps> 2^{-Cn^{\gamma}}$ for any constant $\gamma<1/6$, 
which improves on the simpler $\eps>2^{-C n^{0.12}}$ bound from \cref{remark:KT-preproc}.\footnote{In contrast, using the AKS primality test would lead to running time $\Otilde(\log^9(n/\eps))$, which is only $\Otilde(n)$ when $\eps> 2^{-Cn^{1/9}}$, and therefore worse than the bound from \cref{remark:KT-preproc}.}
We also note that we cannot hope to improve on this bound by picking a better averaging sampler, since any $(\theta,\gamma)$-averaging sampler must have sample complexity $m=\Omega(\log(1/\gamma)/\theta^2)$~\cite{CEG95}.

\paragraph{RS condenser.}
Unlike the KT condenser, the RS condenser from \cref{thm:rscond} always requires preprocessing.
Therefore, arguably the most interesting direction in this discussion would be to establish a statement of the form ``if $\eps\geq 2^{-n^\gamma}$ for some constant $\gamma>0$, then the RS condenser does not require preprocessing to run in time $\Otilde(n)$'', in analogy with the statement we already have for the KT condenser.
However, it is not clear to us how to implement a sampler-based strategy in this case.
We elaborate on this below.

The preprocessing algorithm considered in \cref{remark:preproc3} has two stages: (1) repeatedly sample an integer $q$ alongside its prime factorization uniformly at random from $[N/2,N]$, with $N = \poly(n/\eps)$, until $q+1$ is prime, as in~\cite[Section 9.6]{Sho05}, and (2) find a primitive element of $\F_q$.
We focus on the first stage. 
The first part of this stage requires sampling a ``random non-increasing sequence'' in $[1,N]$ (see~\cite[Section 9.5 (Algorithm RS)]{Sho05}). 
However, this procedure requires too much randomness, which blows up the seed length of the sampler. Indeed, following the analysis in~\cite[Section 9.5.2]{Sho05} gives that Algorithm RS requires $\Theta(\log^2 N)=\Theta(\log^2(n/\eps))$ bits of randomness in expectation. More precisely, let $O_i$ be the random variables denoting the number of times the integer $i$ appears in the non-increasing sequence, as defined in~\cite[Section 9.5.2]{Sho05}.
There, it is shown that $\E[O_i]=\frac{1}{i-1}$ for all integers $i\in [2,N]$.
Every time $i$ is sampled requires using at least $\log i$ bits of randomness in that iteration of Algorithm RS.
Therefore, the expected number of bits of randomness required is at least
\begin{equation*}
    \sum_{i=1}^N \log i \cdot \E[O_i] \geq \sum_{i=1}^N \frac{\log i}{i} \geq \int_1^N \frac{\log x}{x}\; dx - O(1) = \Theta(\log^2 N).
\end{equation*}
Therefore, the sampler must output samples of bitlength $\Omega(\log^2(n/\eps))$, and so also requires a seed of length $\Omega(\log^2(n/\eps))$.

\medskip

We note that the barriers outlined above are specific to the preprocessing steps we considered in this work, and to the use of samplers. It may be possible to improve on the current results by considering alternative preprocessing steps and by using other pseudorandom objects. We leave this as a natural direction for future work.

\section{Nearly-Linear Time Extractors with Order-Optimal Seed Length}

\subsection{A Non-Recursive Construction}\label{sec:non-rec}

In this section, we use the sampler based on bounded independence from \cref{sec:bounded-indep-sampler} and the nearly-linear time KT condenser from \cref{sec:fast-condensers}
to construct a seeded extractor with order-optimal seed length $O(\log(n/\eps))$ computable in time $\Otilde(n)$.
We remark that the goal of this section is to give a low-error, relatively simple (and in particular, non-recursive) construction. Among the non-recursive ``sample-then-extract'' extractor constructions, two notable ones are \cite{NZ96} and \cite{Zuc96}: The \cite{NZ96} construction uses poly-logarithmic seed (see \cref{fn:nz}); The \cite{Zuc96} construction bears some resemblance to our construction and also utilizes sub-sampling, but only works in the high-error regime, namely has seed length $\Theta(\eps^{-2}+\log n)$. 
Constructions in that framework that support low error and short seed include \cite{SZ99} and followup works such as \cite{GUV09,Zuc97}. We adapt those recursive constructions in \cref{sec:recursive} to get our nearly linear time construction of \cref{thm:rec-ext-intro}. Finally, we note that other block source conversion techniques have been used for constructing pseudorandomness primitives, such as in \cite{T02,DKSS13}, but they are less relevant in our context.

In a nutshell, our extractor proceeds as follows on input an $(n,k)$-source $X$:
\begin{enumerate}
    \item\label{it:condense} Using a fresh seed, apply the lossless KT condenser from \cref{thm:ktcond} to $X$.
    This yields an $(n',k)$-source $X'$ of length $n'\approx k$ and constant entropy rate $\delta$ which can be arbitrarily close to $1$.
    Note that in the parameter regime considered in this section the KT condenser does not require the one-time preprocessing step.

    \item\label{it:sampling1} Using the fact that $X'$ has high min-entropy rate, use the bounded-independence sampler from \cref{lem:k-sampler2} to sample subsources from $X'$ using a fresh seed.
    Specific properties of the bounded-independence sampler allow us to obtain a block source $Z=(Z_1, Z_2,\dots, Z_t)$ with a seed of length only $O(\log(1/\eps))$.
    The number of blocks is $t=O(\log n)$ and the blocks $Z_i$ have geometrically \emph{increasing} lengths, up to an $n^\alpha$ length threshold.

    \item\label{it:sampling2} Now, to prepare for the hash-based iterative extraction, we need to make our blocks \emph{decreasing}. Again, using a short seed, of length $O(\log(n/\eps))$, we transform
    $Z$ into $S = (S_1, \dots, S_t)$, where the blocks are now geometrically decreasing. The blocks lengths will vary from $n^{\beta_1}$ to some $n^{\beta_2}$, for some constants $\beta_1 > \beta_2$. (Here, we do not use the ``prefix samplers'' property of \cref{lem:k-sampler2}, so transforming $Z$ into $S$ can be done using an existing sampler, and specifically the one from \cref{lem:extendedSampler}.)

    \item\label{it:sampling3} Using a fresh seed, apply the fast hash-based extractor from \cref{lem:fasthash} to perform  block source extraction from $S$. Noting that the first block has length $n^{\Omega(1)}$, the block source extraction only outputs $n^{\Omega(1)}$ bits. 
    
    While the seed length of \cref{lem:fasthash} requires at least $m$ random bits,
    we are still able to use only $O(\log(n/\eps))$ bits, since we do not output $n^{\Omega(1)}$ bits already at the beginning of the iterative extraction process, but instead first output logarithmically many bits, and then gradually increase the output length.
\end{enumerate}
Finally, once we have extracted $n^{\Omega(1)}$ random bits, outputting almost all the entropy can be done using standard techniques (see \cref{sec:improving}).
These steps will culminate in the following theorem.

\begin{theorem}[non-recursive construction]\label{thm:main-nonrecursive}
There exists a constant $c \in (0,1)$ such that for every positive integers $n$ and $k \le n$, any $2^{-k^{c}} \le \eps \le \frac{1}{n}$, and any constant $\eta \in (0,1)$, there exists a strong $(k,\eps)$ extractor
\[
\Ext \colon \B^n \times \B^d \rightarrow \B^{m},
\]
where $d = O(\log(n/\eps))$, and 
$m = (1-\eta) k$. Moreover, given inputs $x \in \B^n$
and $y \in \B^d$, we can compute $\Ext(x,y)$ in time $\widetilde{O}(n)$.
\end{theorem}

\subsubsection{\cref{it:sampling1}: Generating the block source}

Because of the initial condensing step, we will assume from here onwards that our input source $X$ is an $(n,k=\delta n)$-source with constant $\delta$.
In order to generate the desired block source, 
we first use a fresh seed $Y$ as input to an appropriate instantiation of the bounded-independence sampler $\Gamma$ from \cref{lem:k-sampler2}.
This yields a tuple of coordinates $\Gamma(Y)=j_1,\dots,j_{m_t}$ from $[n]$, such that $\Gamma(Y)|_{[1,m_i]}$ is an appropriate averaging sampler for every $i$.
Then, we use these coordinates to sample subsources from $X \sim \B^n$, and get a block source with \emph{increasing} blocks. We recall that getting increasing blocks is only an intermediate step towards our final goal. Indeed, in this step, we sample from the source $X$, which will guarantee that a typical prefix leaves enough entropy in the next, larger, blocks. In \cref{subsample} we will need to subsample from those blocks and argue that each block still has entropy even after a typical fixing of all the blocks that precedes it. The latter property is the one needed for block-source extraction.

\begin{lemma}[sampling a block source]\label{thm:block-source}
    There exists a deterministic procedure that given an $(n,k)$-source $X$ with $k\geq  \delta n$, $\delta$ being constant, and:
    \begin{itemize}
        \item A constant loss parameter $\zeta \in (0,1)$,
         \item A closeness parameter $\eps \in (0,1)$,
        \item Number of desired blocks $t \in \mathbb{N}$,
        \item A final, maximal, block length $\Delta_t \le c_{\mathsf{t}} \cdot n$ where $c_{\mathsf{t}}=c_{\mathsf{t}}(\zeta,\delta)$ is constant, and,
    \end{itemize}
    takes an independent and uniform random seed $Y \sim \B^{d_{\mathsf{samp}}}$  and outputs 
    a random variable $Z$ such that, for every $y$, $(Z|Y=y)$ is $\eps$-close to an exact 
    \[
    \left( (\Delta_1,\ldots,\Delta_t),(1-\zeta)\delta \right)
    \]
    block-source, where each $\Delta_{i-1} = \alpha \cdot \Delta_i$ for $\alpha = \frac{\zeta\delta}{16}$ and assuming that $\eps \ge 2^{-(c_{\eps}\Delta_1 - t)}$ for some constant $c_{\eps}=c_{\eps}(\zeta,\delta)$. Moreover,
    the seed length $d = O\left(\frac{\log n}{\log\Delta_1} \cdot \log\frac{t}{\eps}\right)$, and the procedure runs in time $\widetilde{O}(n)+\polylog(t/\eps)$.

    Note that for any constants $0 < \theta_1 < \theta_2 < 1$, and any $\eps \ge 2^{-n^c}$ where $c > 0$ is a small enough constant, we can have $\Delta_t = n^{\theta_2}$
    and $\Delta_1 = n^{\theta_1}$ for some $t = O(\log n)$, with
    seed length $O(\log(t/\eps))$ and runtime $\widetilde{O}(n)$.
\end{lemma}

\begin{proof}
Given our $\Delta_1,\ldots,\Delta_t$, we let $m_i = \sum_{j=1}^{i}\Delta_j$ for $j \in [t]$. Note that for $i \in [t-1]$,
each $m_i = \sum_{j=1}^{i}\Delta_j \le \frac{\alpha}{1-\alpha}\Delta_{i+1}$, so in particular
\[
m_t = m_{t-1} + \Delta_t \le \frac{\alpha}{1-\alpha}\Delta_t + \Delta_t \le n,
\]
by choosing the constant $c_{\mathsf{t}}$ appropriately. 
Let $\Gamma \colon \B^{d_{\mathsf{samp}}} \rightarrow [n]^{m_t}$ be the $(\gamma,\eps_{\Gamma})$-averaging sampler of \cref{lem:k-sampler2}, set with $\eps_{\Gamma} = \frac{1}{\log(\frac{6}{\zeta\delta})} \cdot \frac{\zeta\delta}{6} = O(1)$ and $\gamma = \frac{\eps}{8t}$. Note that then, \[d_{\mathsf{samp}} = O\left( \frac{\log n}{\log m_1} \cdot \log\frac{1}{\gamma} \right) = O\left(\frac{\log n}{\log\Delta_1} \cdot \log\frac{t}{\eps}\right),\]
    and indeed $m_t \le \frac{\eps_{\Gamma}}{16} \cdot n$ can be met by, again, setting the constant $c_{\mathsf{t}} \in (0,1)$ appropriately.
Moreover, we have that for any $i \in [t]$,
\[
W_i = \Gamma(Y)|_{[1,m_i]}
\]
is a $(\gamma,\eps_{\Gamma})$ sampling distribution, where $w \sim W_i$ has distinct symbols. Set $\beta = \frac{\zeta}{2}$.

For each $i\in[t]$, let $A_i=X_{W_i}$.
We may write $A_t=(Z_1,\dots,Z_t)$ with $Z_j=(A_t)_{[m_{j-1}+1,m_j]}$.
Note that under this perspective we have $A_i=(Z_1,\dots,Z_i)$ for each $i\in[t]$.
We claim that $(Z_1,\dots,Z_t)$ is close to an appropriate (exact) block source $Z'=(Z'_1,\dots,Z'_t)$, even conditioned on the seed $Y$.
This follows by an induction argument similar to that of~\cite[proof of Lemma 17]{NZ96} which we detail below.

First, \cref{lem:sample-entropy} instantiated with $\tau = \frac{\beta\delta}{3}$ (notice that indeed $\eps_{\Gamma} \le \frac{\tau}{\log(1/\tau)}$) tells us that
\begin{equation}\label{eq:base-case1}
    (Y, Z_1,\dots,Z_i)= (Y,A_i) \approx_{\gamma + 2^{-\Omega(\tau n)}} (Y, A'_i),
\end{equation}
where $(A'_i|Y=y)$ has entropy rate at least $\delta-3\tau \ge (1-\beta)\delta$ for every $y$.
\cref{eq:base-case1} with $i=1$ and $Z'_1=A'_1$ is our base case.
Now, fix an arbitrary $i\geq 2$ and suppose that we already know that
\begin{equation}\label{eq:IH1}
    (Y, Z_1,\dots,Z_{i-1} ) \approx_{2(i-1)(\gamma + 2^{-\Omega(\tau n)}+\xi)} ( Y, Z'_1,\dots,Z'_{i-1}),
\end{equation}
where $(Z'_j|Y=y,Z'_1=z_1,\dots,Z'_{j-1}=z_{j-1})$ has entropy rate at least $(1-\zeta)\delta$ for every $1\leq j \le  i-1$ and every $y,z_1,\dots,z_{j-1}$, and $\xi = \frac{\eps}{4t}$.

Write $A'_i = (Z''_1,\dots,Z''_i)$ with $|Z''_j|=|Z'_j|$.
Applying \cref{lem:tailboundchainrule} to $(A'_i|Y=y)$ with $\delta = \xi$ yields
\begin{align}
    \minH(Z''_i|Y=y,Z''_1=z_1,\dots,Z''_{i-1}=z_{i-1}) &=\minH(A'_i|Y=y,Z''_1=z_1,\dots,Z''_{i-1}=z_{i-1})\nonumber\\
    &\geq \minH(A'_i) - \sum_{j=1}^{i-1}\Delta_j - \log(1/\xi) \nonumber\\
    &= \minH(A'_i)-m_{i-1} -\log(1/\xi) \nonumber\\
    &\geq (1-\beta)\delta m_i -\frac{\alpha}{1-\alpha}\Delta_i-\log(1/\xi) \label{eq:LB-ent-Z''}
\end{align}
except with probability at most $\xi$ over the choice of $(z_1,\dots,z_{i-1})$. Using the fact that $m_i \ge \Delta_i$, to get entropy rate at least $(1-\zeta)\delta$ it is left to verify that
\[
\left( (\zeta - \beta)\delta - \frac{\alpha}{1-\alpha} \right) \Delta_i = \left( \frac{\beta\delta}{2} - \frac{\alpha}{1-\alpha} \right) \Delta_i  \ge \log(1/\xi).
\]
Using our bound on $\alpha$, we get that $\frac{\alpha}{1-\alpha} \le \frac{\beta\delta}{4}$. Thus, $\frac{\beta\delta}{2}\Delta_i \ge \log(1/\xi)$ holds whenever $\log(1/\eps) \le c_{\eps}\Delta_1 - t$, where both $c_{\eps}$ depend only on $\delta$ and $\zeta$.

Call a vector $\vec{v}=(y,z_1,\dots,z_{i-1})$ \emph{good} if \cref{eq:LB-ent-Z''} holds.
Suppose that we already have blocks $B_1,\dots,B_{i-1}$, arbitrarily distributed.
We generate one more block $B_i$ as follows.
First, sample $\vec{v}\sim (Y,B_1,\dots,B_{i-1})$.
If $\vec{v}$ is good as defined above, we set $B_{i,\vec{v}}$ (the random variable $B_i$ conditioned on $(Y,B_1,\dots,B_{i-1})=\vec{v}$) to be $B_{i,\vec{v}} = (Z''_i|(Y,Z''_1,\dots,Z''_{i-1})=\vec{v})$.
Otherwise, including when $\vec{v}$ is not in the support of $(Y,Z''_0,\dots,Z''_{i-1})$, we set $B_{i,\vec{v}}$ to be uniformly distributed over $\B^{\Delta_i}$ and independent of everything else.
Note that by construction we have $\minH(B_{i,\vec{v}})\geq (1-\zeta)\delta \Delta_i$ for all $\vec{v}$.

Then, it follows from \cref{eq:base-case1} that
\begin{equation}\label{eq:triangle2}
    \left(Y,Z_1,\dots,Z_{i-1},B^{(1)}_i\right) \approx_{\gamma + 2^{-\Omega(\tau n)}} \left(Y,Z''_1,\dots,Z''_{i-1}, B^{(2)}_i\right),
\end{equation}
where on the left-hand side we take $B^{(1)}_i$ to be sampled based on $Y$ and $B_j=Z_j$ for $1\leq j \le i-1$ as described for $B_i$ in the previous paragraph, and on the right-hand side we take $B^{(2)}_i$ to be sampled based on $Y$ and $B_j=Z''_j$ for $0\leq j \le i-1$.
Since $\vec{v}\sim(Y,Z''_1,\dots,Z''_{i-1})$ is good with probability at least $1-\xi$, in which case $B_{i,\vec{w}}$ is sampled identically to $Z''_{i,\vec{w}}$, we get from \cref{eq:triangle2} and a triangle inequality that
\begin{equation}\label{eq:triangle3}
    \left(Y,Z_1,\dots,Z_{i-1},B^{(1)}_i\right) \approx_{\gamma + 2^{-\Omega(\tau n)}+\xi} \left(Y,Z''_1,\dots,Z''_{i-1}, Z''_i\right).
\end{equation}
By \cref{eq:IH1}, we also have that
\begin{equation}\label{eq:triangle4}
    \left( Y,Z_1,\dots,Z_{i-1},B^{(1)}_i \right) \approx_{2(i-1)(\gamma + 2^{-\Omega(\tau n)}+\xi)} \left(Y,  Z'_1,\dots,Z'_{i-1},B_i^{(3)} \right),
\end{equation}
where, again, $B_i^{(3)}$ is sampled based on $Y$ and $B_j=Z'_j$ for $1\leq j \le i-1$ as described for $B_i$ above.
Combining \cref{eq:base-case1} (recall that $A'_i=(Z''_1,\dots,Z''_i)$) with \cref{eq:triangle3,eq:triangle4} via a triangle inequality, we get that
\begin{equation*}
    \left(Y,Z_1,\dots,Z_i \right) \approx_{2i(\gamma + 2^{-\Omega(\tau n)}+\xi)} \left(Y, Z'_1,\dots,Z'_{i-1},B^{(3)}_i\right).
\end{equation*}
Note that the sampling of $B^{(3)}_i$ on the right-hand side of this equation guarantees that $\minH(B^{(3)}_i|Y=y,Z'_1=z_1,\dots,Z'_{i-1}=z_{i-1})\geq (1-\zeta)\delta \Delta_i$ for all $(y,z_1,\dots,z_{i-1})$.
Therefore, $(Z'_1,\dots,Z'_{i-1},Z'_i=B^{(3)}_i)$ is indeed the target block source with $i$ blocks. Setting $i=t$, and by inspection of the sampling process for the $Z'_j$-s,  gives that $(Z|Y=y)$ is 
\[
2t \cdot \left( \gamma+2^{-\Omega(\tau n)} + \xi\right) \le \eps
\]
close to an exact $((\Delta_1,\ldots,\Delta_t),(1-\zeta)\delta)$ block source for every $y$.

The bound on the runtime follows easily, recalling that $\Gamma$ runs in time $\widetilde{O}(n)+\polylog(1/\gamma)$.
\end{proof}

\subsubsection{\cref{it:sampling2}: Subsampling from the block source}\label{subsample}

To apply iterative extraction, we will need our block source to have \emph{decreasing} blocks. Here, we will use a sampler to sample from each block, using the same seed across the blocks.
\begin{lemma}[subsampling from a block source]\label{lem:subsampling}
There exists a deterministic procedure that given a
\[
\left( (\Delta_1,\ldots,\Delta_t),\delta \right)
\]
block-source $Z = (Z_1,\ldots,Z_t)$, where $\Delta_1 \le \ldots \le \Delta_t$ and $\delta$ is a constant, 
\begin{itemize}
    \item A constant shrinkage parameter $\alpha \in (0,1)$,
    \item A constant loss parameter $\zeta \in (0,1)$,
        \item A closeness parameter $\eps \in (0,1)$,
    \item An initial, maximal, block length $\ell_1 \le \Delta_1$, and,
    \item An independent and uniform random seed $Y \sim \B^{d_{\mathsf{samp}}}$,
\end{itemize}
satisfies the following. Assuming that $\ell_t \ge c_{1}\log(t/\eps)$ for some constant $c_1 = c_1(\zeta,\delta)$, it outputs a random variable $S$ such that, for every $y$, $(S|Y=y)$ is $\eps$-close to an exact
\[
\left( (\ell_1,\ldots,\ell_t),(1-\zeta)\delta \right)
\]
block-source, where each $\ell_{i+1} = \alpha \cdot \ell_i$
Moreover, the seed length $d_{\mathsf{samp}} = \log\frac{\Delta_t}{\ell_1} + O\left(t+\log\frac{1}{\eps}\right)$,
and the procedure runs in time $t \cdot \polylog(\Delta_t) \left( \ell_1 + \log^{2}(t/\eps)\right)$.

Note that when $\Delta_t = n^{\theta_1}$ and $\ell_1 = n^{\beta}$ for some constants $\theta_1, \beta \in (0,1)$, $d_{\mathsf{samp}} = O(\log(n/\eps))$,  the procedure runs in time $O(n)$, and we can take any $\eps \ge 2^{-c \cdot \ell_t}$ for some constant $c$ that depends on $\zeta$ and $\delta$. 
\end{lemma}
\begin{proof}
For $i \in [t]$, let $m_i = \sum_{j=1}^{i}\ell_i$, recalling that $\ell_i = \alpha^{i-1}\ell_1$.
For each $i \in [t]$, let $\Gamma_i \colon \B^{d_i} \rightarrow [\Delta_i]^{\ell_i}$ be the
$(\gamma,\eps_{\Gamma})$ distinct-samples sampler of \cref{lem:extendedSampler}, where
$\gamma = \frac{\eps}{2t}$ and $\eps_{\Gamma} = \frac{1}{\log(\frac{6}{\zeta\delta})} \cdot \frac{\zeta\delta}{6}  = O(1)$. We need to make sure that each $\ell_i \ge c \cdot \frac{\log(1/\gamma)}{\eps^2_{\Gamma}}$ for some universal constant $c$, and indeed that is the case, by our constraint on $\ell_t$.
Also, 
$
d_i = \log(\Delta_i/\ell_i) + O(\log\frac{1}{\gamma} \cdot \poly(1/\eps_{\Gamma}))
$
and we set $d_{\mathsf{samp}}$ to be the maximum over the $d_i$-s, so 
\[
d_{\mathsf{samp}} = d_t = \log\frac{\Delta_t}{\ell_1} + t \cdot \log\frac{1}{\alpha} + O\left(\log\frac{t}{\eps}\right).
\]
We denote the corresponding samples by $W_i = \Gamma_i(Y|_{[1,d_i]})$, and let $S_i = (Z_i)_{W_i}$. 
Setting $\eps_i' = 2^{-(\zeta/2)\delta\Delta_i}$ and observing that $\delta \Delta_i = (1-\frac{\zeta}{2})\delta \Delta_i + \log(1/\eps'_i)$, we get that $Z$ is $\eps'= \sum_{i}\eps'_i$ close to some $Z'$, an exact $((\Delta_1,\ldots,\Delta_t),(1-\zeta)\delta)$-source.
From here onwards, assume that $Z$ is the exact block source, and aggregate the error.

Next, we invoke \cref{lem:sample-entropy} with $\tau = \frac{\zeta\delta}{6}$ (notice that indeed $\eps_{\Gamma} \le \frac{\tau}{\log(1/\tau)}$), and get that for every $i \in [t]$,
and $z_{\mathsf{pre}} \in \B^{\Delta_1+\ldots+\Delta_{i-1}}$,
\begin{equation*}
\left( Y,S_{i,z_{\mathsf{pre}}} \right) \approx_{\eps''_i = \gamma+2^{-\Omega(\tau\Delta_i)}} \left( Y, S'_{i,z_{\mathsf{pre}}}\right),
\end{equation*}
where $S_{i,z_{\mathsf{pre}}}$ denotes $S_i$ conditioned on $(Z_1,\dots,Z_{i-1})=z_{\mathsf{pre}}$ and $S'_{i,z_{\mathsf{pre}}}$ satisfies $\minH(S'_{i,z_{\mathsf{pre}}}|Y=y)\geq (1-\frac{\zeta}{2})^2\delta \cdot \ell_i \ge (1-\zeta)\delta \cdot \ell_i$ for all $y$.
Thus, in particular, this holds if we condition on 
any sample from $(S_1,\ldots,S_{i-1})$, and so we have that for every $i \in [t]$,
\begin{equation}\label{eq:induction-block}
\left(Y, S_1,\ldots,S_{i-1},S_i \right) \approx_{\eps''_i}\left(Y, S_1,\ldots,S_{i-1},S'_i \right),
\end{equation}
where $\minH(S'_i|Y=y,S_1=s_1,\dots,S_{i-1}=s_{i-1})\geq (1-\zeta)\delta \cdot \ell_i$ for all $(y,s_1,\dots,s_{i-1})$.

This implies
that, conditioned on the seed $Y$, $(S_1,\ldots,S_t)$
has distance
\[
\eps' + \sum_{i=1}^{t}\eps_i'' \le t \cdot (\eps'_1 + \eps''_1) \le \eps
\]
from an (exact) $((\ell_1,\ldots,\ell_t),(1-\zeta)\delta)$ block source, where we used the fact that the
$2^{-\Omega(\tau\Delta_1)}$ and $2^{-(\zeta/2)\delta\Delta_1}$ terms are at most $\frac{\eps}{4t}$, which follows from the fact that $c_1\log(t/\eps) \le \Delta_1$ for a suitable choice of $c_1$, and where $\eps'$ accounts for the assumption that $Z$ is the exact block source above.
This can be shown by induction on the number of blocks using \cref{eq:induction-block} analogously to (and even in a simpler way than)~\cite[proof of Lemma 17]{NZ96}, and similarly to the proof of \cref{thm:block-source}.
Since we believe this proof is easier to follow than the proof of \cref{thm:block-source}, we give details here too for the sake of exposition.

First, the base case is given by \cref{eq:induction-block} for $i=1$.
Now, fix an arbitrary $i\geq 2$ and suppose that we already know that
\begin{equation}\label{eq:IH2}
    (Y,S_1,\dots,S_{i-1})\approx_{\sum_{j=1}^{i-1}\eps''_j} (Y,S'_1,\dots,S'_{i-1}),
\end{equation}
where $\minH(S'_j|Y=y,S'_1=s_1,\dots,S'_{j-1}=s_{j-1})\geq (1-\zeta)\delta \ell_j$ for every $1\leq j \le i-1$ and all $(y,s_1,\dots,s_{j-1})$.

We now show how to extend this by one block.
Generally speaking, suppose that we already have blocks $B_1,\dots,B_{i-1}$, arbitrarily distributed.
Analogously to the proof of \cref{thm:block-source}, we generate one more block $B_i$ by first sampling $\vec{v}\sim (Y,B_1,\dots,B_{i-1})$.
If $\vec{v}$ is in the support of $(Y,S_1,\dots,S_{i-1})$, we set $B_{i,\vec{v}}$, the random variable $B_i$ conditioned on $(Y,B_1,\dots,B_{i-1})=\vec{v}$, to be $B_{i,\vec{v}}=(S'_i|(Y,S_1,\dots,S_{i-1})=\vec{v})$.
Otherwise, we set $B_{i,\vec{v}}$ to be uniformly distributed over $\B^{\ell_i}$ and independent of everything else.
By construction, $\minH(B_{i,\vec{v}})\geq (1-\zeta)\delta \ell_i$ for all $\vec{v}$.

From \cref{eq:IH2}, it follows that
\begin{equation}\label{eq:IH-triangle}
    \left(Y,S_1,\dots,S_{i-1},B^{(1)}_i \right) \approx_{\sum_{j=1}^{i-1}\eps''_j} \left(Y,S'_1,\dots,S'_{i-1},B^{(2)}_i\right),
\end{equation}
where $B^{(1)}_i$ is sampled by following the procedure in the paragraph above with $B_j=S_j$ for all $j \le i-1$, and $B^{(2)}_i$ is sampled from $B_j=S'_j$ for all $j \le i-1$.
Now, note that $(Y,S_1,\dots,S_{i-1},B^{(1)}_i)$ is distributed exactly like $(Y,S_1,\dots,S_{i-1},S'_i)$, because when $B_j=S_j$ for all $j \le i-1$ we get that $\vec{v}$ above is always in the correct support.
Therefore, combining this observation with \cref{eq:induction-block,eq:IH-triangle} and a triangle inequality yields
\begin{equation*}
    (Y,S_1,\dots,S_{i-1},S_i) \approx_{\eps''_i+\sum_{j=1}^{i-1}\eps''_j = \sum_{j=1}^i \eps''_j}\left(Y,S'_1,\dots,S'_{i-1},B^{(2)}_i\right).
\end{equation*}
To conclude the argument, it suffices to note that $(S'_1,\dots,S'_{i-1},S'_i=B^{(2)}_i|Y=y)$ is the desired exact block source by inspection of the sampling process for the $S'_j$-s, and take $i=t$.

To establish the runtime, note that we simply apply $\Gamma_i$ for each $i \in [t]$, which takes
\[
\sum_{i=1}^{t} \log^{2}(1/\gamma) \cdot \polylog(\Delta_i) + O(\ell_i \log\Delta_i) \le t \cdot \polylog(\Delta_t) \left( \ell_1 + \log^{2}(t/\eps)\right)
\]
time.
This concludes our lemma.
\end{proof}

\subsubsection{\cref{it:sampling3}: Applying a block source extractor}

We now wish to extract from our decreasing-blocks block source, and for that we combine
\cref{lem:subsampling,thm:block-source} with the block source extraction of \cref{lem:block-ext}.
This will give us a nearly linear-time logarithmic-seed extractor that outputs $n^{\Omega(1)}$ bits. For the $\Ext_i$-s in \cref{lem:block-ext}, we will
use the fast hash-based extractors from \cref{lem:fasthash}.

\begin{lemma}\label{lem:short}
There exists a small constant $c > 0$ such that the following holds. For every large enough $n$, any constant $\delta \in (0,1)$, any $k \ge \delta n$, and any $\eps \ge 2^{-n^{c}}$, there exists a strong $(k,\eps)$ extractor
\[
\Ext_{\mathsf{short}} \colon \B^{n} \times \B^{d} \rightarrow \B^{m}
\]
where $d = O(\log(n/\eps))$, and $m = n^{c}$. Moreover, given inputs $x \in \B^n$ and $y \in \B^d$, we can compute $\Ext_{\mathsf{short}}(x,y)$ in time $\widetilde{O}(n)$.
\end{lemma}
\begin{proof}
Let $X$ be an $(n,k=\delta n)$-source. Set $\eps' = \eps/3$, $\theta_1 = 8/10$, $\theta_2 = 9/10$, and $\zeta = 1/10$.  We first apply \cref{thm:block-source} with $\Delta_t = n^{\theta_2}$, $\Delta_1 = n^{\theta_1}$, and error $\eps'$, where
$t = O(\log n)$ is as guaranteed from \cref{thm:block-source}. This requires a seed of length $d_1 = O(\log(1/\eps')) = O(\log(1/\eps))$, and in time $\widetilde{O}(n)$ we output a random variable $Z_1$
which is $\eps'$-close to an exact $((\Delta_1,\ldots,\Delta_t),(1-\zeta)\delta)$ block source for every fixing of the seed. 

Set $\beta = 7/10$, and $\gamma = 6/10 < \beta$. Set $\alpha$ to be the constant such that $n^{\beta} \cdot \alpha^{t-1} = n^{\gamma}$. 
We then apply
\cref{lem:subsampling} on $Z_1$ with that $\alpha$, the same $\zeta$, closeness $\eps'$ and an initial
block length $\ell_1 = n^{\beta}$. This gives us a random variable $Z_2$ that is $2\eps'$-close to a
\[
\left( (\ell_1 = n^{\beta},\ldots,\ell_t = n^{\gamma}), \delta' \triangleq (1-\zeta)^2\delta  \right)
\]
block source, requires a seed of length
$d_2 = O(\log(n/\eps'))=O(\log(n/\eps))$, and runs in time  $t \cdot \polylog(\Delta_t) \left( \ell_1 + \log^{2}(t/\eps)\right) = O(n)$, assuming $c$ is small enough. Again, $Z_2$ is $\eps'$-close to an exact block source for every fixing of the seed. 

For our next and final step, of performing the block-source extraction itself, set $d_3 = c_{\mathsf{E}}\log(\ell_t/\eps_{\Ext})$ where $c_{\mathsf{E}}$ is the constant guaranteed by \cref{lem:fasthash}. Also, let $\eps_{\Ext} = \frac{\eps'}{6t}$, and $\theta$ will be a constant whose value will be later determined. We will use the following extractors:
\begin{itemize}
    \item Let $\Ext_{t} \colon \B^{\ell_t} \times \B^{d_3} \rightarrow \B^{m_{t} = (1+\theta)d_3}$ be the $(k_t = \delta'\ell_t,\eps_{\Ext})$ extractor guaranteed to us by \cref{lem:fasthash}. Notice that we need to satisfy $k_t \ge \theta d_3+c_{\mathsf{E}}\log(1/\eps_{\Ext})$. This can be satisfied making sure that $\eps$ is at most $2^{-\Omega(\ell_t)}$, where the hidden constant depends on $c_{\mathsf{E}}$. 
    \item For each $i \in [t-1]$, let 
    \[
    \Ext_{i} \colon \B^{\ell_i} \times \B^{m_{i+1}} \rightarrow \B^{m_i}\] be the $(k_{i} = \delta' \ell_i,\eps_{\Ext})$ extractor guaranteed to us by \cref{lem:fasthash},
    where $m_{i} = (1+\theta)m_{i+1}$. We need to make sure that $m_{i+1} \ge c_{\mathsf{E}}\log(\ell_i/\eps_{\Ext})$ and that $k_i \ge \theta m_{i+1}+c_{\mathsf{E}}\log(1/\eps_{\Ext})$. To see that the latter holds, note that 
    $
    k_i = \delta' \cdot \ell_1  \alpha^{i-1}
    \ge n^{\gamma/2}
    $
    and that
    $
    \theta m_{i+1}+c_{\mathsf{E}}\log(1/\eps_{\Ext}) = \theta(1+\theta)^{t-i}d_3 + c_{\mathsf{E}}\log(1/\eps_{\Ext}) < n^{\gamma/2},
    $
    if we choose $\theta$ to be a small enough constant (with respect to the constant $\frac{\log n}{t}$) and $\eps$ to be, again, at most $2^{-\Omega(\ell_t)}$. 
\end{itemize}
Everything is in place to apply our block source extraction, \cref{lem:block-ext}, on $Z_2$ and an independent and uniform seed of length $d_3$. We get that $\mathsf{BExt}$ outputs $Z_3$ of length $m_1 = n^{\Omega(1)}$, which is $2t\eps_{\Ext} \le \eps'$ close to uniform, and runs in time
$
O\left(\sum_{i=1}^{t}\ell_i\log \ell_i \right) = O(n).
$
Recall that indeed, as \cref{lem:block-ext} requires, all the $\Ext_{i}$-s output their seed.

To conclude, note that the overall error of our extractor is at most $3\eps' = \eps$, and the seed length is $d_1 + d_2 + d_3 = O(\log(n/\eps))$.
\end{proof}

\subsubsection{Improving the output length}\label{sec:improving}

The extractor $\Ext_{\mathsf{short}}$ from \cref{lem:short} only outputs $n^{\Omega(1)}$ bits. Here, we will use an extractor $\Ext_{\mathsf{out}}$ that outputs a linear fraction of the entropy but requires a (relatively) long seed, and use \cref{cor:boosting2} to boost the output length. For $\Ext_{\mathsf{out}}$, we will again use a sample-then-extract extractor, however this time, we can use \emph{independent} samples to create a block source with exponentially decreasing blocks. This setting is easier, and we can simply use the original \cite{NZ96} construction. Since a similar construction will be analyzed later in the paper (including a time complexity analysis), we choose to employ it instead of revisiting \cite{NZ96}. We state it formally as a corollary below.

\begin{corollary}\label{cor:out}
There exists a constant $C \ge 1$ such that for any
constants $\tau,c \in (0,1)$, any large enough positive integer $n$ and any $\eps \ge 2^{-n^c}$, there exists a strong $(k = (1-\tau) n,\eps)$ extractor
\[
\Ext_{\mathsf{out}} \colon \B^n \times \B^d \rightarrow \B^{m}
\]
where $d = O(\log n \cdot \log(n/\eps))$, and 
$m = k/C$. Moreover, given inputs $x \in \B^n$
and $y \in \B^d$, we can compute $\Ext_\mathsf{out}(x,y)$ in time $\widetilde{O}(n)$.
\end{corollary}
The correctness of \cref{cor:out} follows from \cref{coro:Ext-bad-output} applied with $i=1$ (which is indeed non-recursive), without the need for a preliminary condensing step.\footnote{\label{fn:nz}As mentioned earlier, \cref{cor:out} can already be deduced from \cite{NZ96} (modulo the tight runtime analysis), with a slightly worse seed of $d = O(\log^{2}n \cdot \log(n/\eps))$, which would not change the parameters of our overall construction. There, they first convert $X$ into a block source $Z$ using $\log n$ \emph{independent} samples from a $k$-wise independent sample space, for $k \approx \log(n/\eps)$. The block source $Z$ has decreasing blocks, so the standard block source extraction scheme can then be employed. The fact that this procedure can be implemented in time $\widetilde{O}(n)$ follows easily from the runtime analysis of primitives in our work (specifically, \cref{arithmetic:it1} of \cref{lemma:fast}, and  \cref{lem:fasthash}).}

Plugging-in $\Ext_{\mathsf{out}}$ and $\Ext_{\mathsf{short}}$ into
\cref{cor:boosting2} readily gives the following result.
\begin{lemma}\label{lem:pre-final}
There exist constants $\tau,c \in (0,1)$ such that for every positive integer $n$, and any $2^{-n^{c}} \le \eps \le \frac{1}{n}$, there exists a $(k = (1-\tau)n,\eps)$ extractor
$
\Ext \colon \B^n \times \B^d \rightarrow \B^{m}
$
where $d = O(\log(n/\eps))$, and 
$m = c k$. Moreover, given inputs $x \in \B^n$
and $y \in \B^d$, we can compute $\Ext(x,y)$ in time $\widetilde{O}(n)$.
\end{lemma}

To boost the output length in \cref{lem:pre-final} from $\Omega(k)$ to $(1-\eta)k$
for any constant $\eta >0$, we apply \cref{lem:boost-output} a constant number of times depending only on $\eta$ (that is, we simply apply $\Ext$ with independent seeds and concatenate the outputs). 
To then go from any min-entropy requirement $k$ to entropy rate $1-\tau$, we first apply the KT condenser from \cref{thm:ktcond}.
Since $\eps\geq 2^{-k^{c}}$ we also have that $\eps\geq 2^{-n^{0.1}}$ if $c<0.1$, and so the KT condenser does not require preprocessing.
Furthermore, we can ensure that $k\geq C'\log^2(n/\eps)$ with $C'>0$ a sufficiently large constant so that the conditions for applying the KT condenser are satisfied whenever $n$ is larger than some constant.

This finally gives us our main theorem for this section,
\cref{thm:main-nonrecursive}, apart from the strongness property, which we now discuss.

\paragraph{The non-recursive construction is strong.}  In what follows, we refer  to the itemized list in the beginning of the section. The condensing step, \cref{it:condense}, is strong, since we use strong condensers.
Next, the block source creators of \cref{thm:block-source,lem:subsampling} are strong, so \cref{it:sampling1,it:sampling2} hold in a strong manner as well. 
\cref{it:sampling3} readily gives a strong extractor.
For the output-extending phase, \cref{cor:boosting2} tells us that the extractor from \cref{lem:pre-final} is strong. Finally, we apply that extractor several times with independent seeds, and the strongness of that procedure is guaranteed from \cref{lem:boost-output}.

\subsection{A Recursive Construction}\label{sec:recursive}

In this section, we prove the following.

\begin{theorem}[recursive construction]\label{thm:rec-ext}

For any constant $\eta>0$ there exists a constant $C>0$ such that the following holds.
    For any positive integers $n$ and $k\leq n$ and any $\eps>0$ satisfying $k\geq C \log(n/\eps)$ there exists a strong $(k,\eps)$-seeded extractor \[\Ext \colon \bits^n\times\bits^d\to\bits^m\] with seed length $d\leq C\log(n/\eps)$ and output length $m\geq (1-\eta)k$.
    Furthermore, 
    \begin{enumerate}
        \item if $k \geq 2^{C\log^*\!n}\cdot \log^2(n/\eps)$ and $\eps\geq 2^{-Cn^{0.1}}$, then $\Ext$ is computable in time $\Otilde(n)$;

        \item if $k\geq 2^{C\log^*\!n}\cdot \log^2(n/\eps)$ and $\eps< 2^{-Cn^{0.1}}$, then $\Ext$ is computable in time $\Otilde(n)$ after a preprocessing step, corresponding to generating $O(\log^*\!n)$ primes $q\leq \poly(n/\eps)$;

        \item if $k < 2^{C\log^*\!n}\cdot \log^2(n/\eps)$, then $\Ext$ is computable in time $\Otilde(n)$ after a preprocessing step, corresponding to generating $O(\log\log n)$ primes $q\leq \poly(n/\eps)$ and a primitive element for each field $\F_q$.        
    \end{enumerate}

\end{theorem}

In a nutshell, our construction behind \cref{thm:rec-ext} works by considering two cases.
If $\eps> C n^3\cdot 2^{-k/\log k}$, then we instantiate the recursive approach of Srinivasan and Zuckerman~\cite{SZ99} appropriately.
Otherwise, we apply the recursive approach of Guruswami, Umans, and Vadhan~\cite{GUV09}.

\subsubsection{The (extremely) low-error case}

In this section, we consider the lower error case of \cref{thm:rec-ext} where $\eps\leq Cn^3\cdot 2^{-k/\log k}$.
We instantiate the recursive approach from~\cite[Section 4.3.3]{GUV09} appropriately, and analyze its time complexity.
Crucially, because of our upper bound on $\eps$, we will only need to run $O(\log\log n)$ levels of their recursive approach.

In order to obtain the statement of \cref{thm:rec-ext} for output length $m\geq (1-\eta)k$ with $\eta$ an arbitrarily small constant, it suffices to achieve output length $m=\Omega(k)$ and then apply \cref{lem:boost-output} a constant number of times depending only on $\eta$.
Therefore, we focus on achieving output length $m=\Omega(k)$.

\begin{theorem}
    There exist constants $c,C>0$ such that the following holds.
    For any positive integers $n$ and $k\leq n$ and any $\eps\in (0,Cn^3\cdot 2^{-k/\log k}]$  further satisfying $k>C\log(n/\eps)$, there exists a strong $(k,\eps)$-seeded extractor $\Ext \colon \bits^n\times\bits^d\to\bits^m$ with seed length $d\leq C\log(n/\eps)$ and output length $m\geq k/3$.
    
    Furthermore, $\Ext$ is computable in time $\Otilde(n)$ after a preprocessing step that corresponds to finding primitive elements of $O(\log\log n)$ fields $\F_q$ with prime  orders $q\leq \poly(n/\eps)$.
\end{theorem}

\begin{proof}
We discuss our instantiation of the recursive approach from~\cite{GUV09} in detail, as it will be relevant to the time complexity analysis.
Let $\eps_0 = \eps/ \log^C n$ and $d=C\log(n/\eps_0)=O(\log(n/\eps))$ for a large enough constant $C>0$ to be determined later.
For an integer $k\geq 0$, let $i(k)=\left\lceil\log\left(\frac{k}{8d}\right)\right\rceil$, which determines the number of levels in our recursion.
It will be important for bounding the time complexity of this construction to observe that
\begin{equation}\label{eq:bound-ik}
    i(k)=O(\log\log n)    
\end{equation}
because $\eps\leq C n^3 \cdot 2^{-k/\log k}$.
For each $k$, we define a family of strong $(k,\eps_{i(k)})$-seeded extractors $\Ext_{i(k)} \colon \B^n\times\B^d\to\B^m$ with $\eps_{i(k)}\leq 9\eps_{i(k/3)}+63\eps_0$ when $i(k)>0$ by induction on $i(k)$.
Solving this recursion yields $\eps_{i(k)}=2^{O(i(k))}\cdot \eps_0\leq \eps$, provided that $\eps_0=\eps/\log^C n$ for a sufficiently large constant $C>0$.

\paragraph{Base case.}
For the base case $i(k)=0$, which holds when $k\leq 8d$, we choose $\Ext_0$ to be the $(k,\eps_0)$-seeded extractor defined below. 
Before we define and analyze it formally, we informally discuss how the extractor works.
Recall that we are aiming for seed length $d$, which in this base case satisfies $d\geq k/8$, output length $m\geq k/3$, and nearly-linear time complexity.
Roughly speaking, on input an $(n,k)$-source $X$ we first apply the fast RS strong condenser to obtain an output $X'$ that is close to a source with high min-entropy rate.
Then, we apply the fast hash-based extractor from \cref{lem:fasthash-short}, which can be made to require only seed length $\approx k/t$ for a large constant $t$, to $X'$ with a fresh seed.
More formally,
\begin{enumerate}
    \item Apply the lossy RS strong condenser $\RSCond$ (\cref{thm:rscond}) on $X$, instantiated with $\alpha=1/400$ and error $\eps'_0=\eps_0/2$. This requires a seed $Y_1$ of length $d_1\leq C_0\log(n/\eps'_0)$, for some constant $C_0>0$, and is a valid invocation since $k\geq C\log(n/\eps'_0)$ for a sufficiently large constant $C>0$.
    The corresponding output $X'$ satisfies $(Y_1, X') \approx_{\eps'_0} (Y_1, Z)$, for some  $(n',k')$-source $Z$ with $k'\geq(1-2\alpha)n'= (1-1/200)n'$.

    \item Let $\Ext'_0 \colon \bits^{n'}\times\bits^{d_2}\to\bits^{m'}$ be the average-case strong $(k',\eps'_0)$-seeded extractor from \cref{lem:fasthash-short} instantiated with $t=10$, which requires a seed $Y_2$ of length $d_2\leq k'/10 + C'_0\log(n'/\eps'_0)$ for some constant $C'_0>0$ and has output length $m'\geq k'/2$.
    The conditions for the invocation of \cref{lem:fasthash-short} with $t=10$ are satisfied since $k'\geq (1-1/200)n'=(1-\frac{1}{20t})n'$ and
    \begin{equation*}
        2^{-n'/500}\leq 2^{-k/500}\leq (\eps_0/n)^{C/500}\leq \eps'_0,
    \end{equation*}
    where the second inequality uses the theorem's hypothesis that $k\geq C\log(n/\eps)$ with $C>0$ a sufficiently large constant.
\end{enumerate}

We set $Y=(Y_1, Y_2)$ and define $\Ext_0(X,Y)=\Ext'_0(\RSCond(X,Y_1),Y_2)$.
We now argue that $\Ext_0$ is an extractor with the desired properties.
From the discussion above, we have
\begin{equation*}
    \left( Y, \Ext_0(X,Y) \right) = \left(Y_1, Y_2, \Ext'_0(\RSCond(X,Y_1),Y_2 \right) \approx_{\eps'_0} \left( Y_1, Y_2, \Ext'_0(Z,Y_2) \right) \approx_{\eps'_0} (Y_1, Y_2, U_{m'}).
\end{equation*}
Therefore, the triangle inequality implies that $\Ext_0$ is an average-case strong $(k,2\eps'_0=\eps_0)$-seeded extractor.
It remains to argue about the seed length, output length, and time complexity of $\Ext_0$.
The seed length of $\Ext_0$ is 
\begin{equation*}
    d_1+d_2\leq k'/10 + (C_0+C'_0)\log(n'/\eps'_0) \leq 0.8d + (C_0+C'_0)\log(n'/\eps'_0) \leq d,
\end{equation*}
provided that $d= C\log(n/\eps)$ with $C$ a sufficiently large constant.
The output length of $\Ext_0$ is $m'\geq k'/2\geq k/3$, since $k'\geq (1-1/200)k$.
Finally, both steps above take time $\Otilde(n)$, and so $\Ext_0$ can be computed in time $\Otilde(n)$ after a one-time preprocessing step.

\paragraph{Induction step.}
When $i(k)>0$, we assume the existence of the desired average-case strong extractors $\Ext_{i(k')}$ for all $i(k')<i(k)$ as the induction hypothesis.
More precisely, we assume that for all $k'$ such that $i(k')<i(k)$ there exists a family of average-case strong $(k',\eps_{i(k')})$-seeded extractors $\Ext_{i(k')}\colon \B^n\times\B^d\to\B^{k'/3}$ parameterized by $n$ computable in time $\Otilde(n)$ after a one-time preprocessing step. 

Intuitively, we will first use these extractors to construct an extractor $\Ext'_{i(k)}$ with all the desired properties (min-entropy requirement $k$, small error $\eps'=\eps_{i(k/3)}+7\eps_0$, small seed length $d'\approx d/16$) \emph{except} that the output length is not large enough. Specifically, we will obtain output length $m=k/9$, but would like to get output length $m\geq k/3$ for the induction step.
Roughly speaking, we proceed as follows on input an $(n,k)$-source $X$.
First, we apply the fast RS strong condenser to $X$ to obtain a source $X'$ with high min-entropy rate.
Then, we split $X'$ in half to create a block source $(X'_1, X'_2)$ with two blocks.
This split decreases the entropies of $X'_1$ and $X'_2$, so we apply the fast RS strong condenser to $X'_2$ to replace this second block by a block $X''_2$ with high min-entropy rate.
Finally, we perform block source extraction on $(X'_1, X''_2)$.
Recall that all we are missing after this is a sufficiently large output length.
The output length can be boosted via standard techniques, at the expense of slightly larger error and seed length.

More formally,
\begin{enumerate}
    \item\label{it:condense-GUV} Apply the lossy RS strong $(k,k',\eps_1=\eps_0^2)$-condenser $\RSCond$ (\cref{thm:rscond}) on $X$ with $\alpha=1/20$ and a seed $Y_{\mathsf{RS}}$ of length $d_{\mathsf{RS}}\leq C_{\mathsf{RS}}\log(n/\eps_0)$. 
    This is valid since 
    $k>8d\geq C_\alpha\log(n/\eps)$ with $C_\alpha$ the constant from \cref{thm:rscond} with $\alpha=1/20$
    if the constant $C$ in the theorem statement is large enough.
    By the second part of \cref{thm:rscond} we get that with probability at least $1-\eps_0$ over the choice of $Y_{\mathsf{RS}}=y$ it holds that the corresponding condenser output $X'$ is $\eps_0$-close to some $(n',k')$-source $Z$ with $k'\geq (1-2\alpha)n'=0.9n'$.
    For the sake of exposition, from here onwards we work under such a good choice of the seed $Y_{\mathsf{RS}}$, and we will add the $\eps_0$ slack term to the final error.

    \item\label{it:split-X} 
    Let $(X'_1, X'_2)$ correspond to the first two blocks of $\lfloor n'/2\rfloor\triangleq n''$ bits of $X'$.
    By \cref{lem:block-chain-rule} instantiated with $n''$ and $\Delta=0.1n'$ and the fact that $X'$ is $\eps_0$-close to an $(n',k')$-source, we get that $(X'_1, X'_2)$ is $(\eps_{0}+2\eps_0=3\eps_0)$-close to an exact $((n'',n''),k''/n'')$-block-source $(W_1, W_2)$ with
    \begin{equation}\label{eq:LBk''}
        k'' \geq n'' -\Delta -\log(1/\eps_0)\geq 0.4n' - \log(1/\eps_0)-1\geq k/3,
    \end{equation}
    since $n'\geq k > d=C\log(n/\eps_0)$ for a sufficiently large constant $C>0$.

    \item\label{it:condense-again} Apply the lossy RS strong $(k'',k''',\eps_1=\eps_0^2)$-condenser $\RSCond'$ (\cref{thm:rscond}) to $X'_2$ with $\alpha=1/800$ and a seed $Y'_{\mathsf{RS}}$ of length at most $d'_{\mathsf{RS}}=C'_{\mathsf{RS}}\log(n''/\eps_1)\leq C'_{\mathsf{RS}}\log(n/\eps_0)$ to get $X_2''$.
    From \cref{it:split-X} and the data-processing inequality, we know that
    \begin{equation}\label{eq:replace-Ws}
        (Y'_{\mathsf{RS}}, X'_1, X''_2)= \left( Y'_{\mathsf{RS}}, X'_1, \RSCond(X'_2,Y'_{\mathsf{RS}})\right) \approx_{3\eps_0} \left( Y'_{\mathsf{RS}},  W_1, \RSCond(W_2,Y'_{\mathsf{RS}}) \right).
    \end{equation}
    Since $(W_2|W_1=w_1)$ is a $k''$-source for any $w_1$ in the support of $W_1$,
    we conclude from \cref{thm:rscond} and \cref{eq:replace-Ws} that
    \begin{equation*}
        \left( Y'_{\mathsf{RS}},  W_1, \RSCond(W_2,Y'_{\mathsf{RS}}) \right) \approx_{\eps_1} \left( Y'_{\mathsf{RS}},  W_1, \widetilde{W_2} \right),
    \end{equation*}
    where $\widetilde{W_2}\sim \B^{n'''}$ and $\minH(Y'_{\mathsf{RS}}, \widetilde{ W_2}|W_1=w_1)\geq k'''+d'_{\mathsf{RS}}$ for all $w_1$ in the support of $W_1$, with $n'''\geq k''\geq k'''\geq (1-1/400) n'''$.
    This is a valid invocation since $k''\geq k/3>8d/3> C\log(n/\eps)$ for a large enough constant $C>0$ by \cref{eq:LBk''}.
    Therefore, by the second part of \cref{thm:rscond}, with probability at least $1-\eps_0$ over the choice of  $Y'_{\mathsf{RS}}=y'$ we get that
    \begin{equation}\label{eq:replace-W'2}
        \left(W_1, \widetilde{W_2} \right) | \set{Y'_{\mathsf{RS}}=y'} \approx_{\eps_0} (W_1, W'_2),
    \end{equation}
    where $W'_2\sim\B^{n'''}$ satisfies $\minH(W'_2|W_1=w_1)\geq k'''\geq (1-1/400) n'''$.
    Fix such a good fixing of $Y'_{\mathsf{RS}}$ from now onwards.
    As before, we will account for the probability $\eps_0$ of fixing a bad seed in the final extractor error.
    Then, by combining \cref{eq:replace-Ws,eq:replace-W'2} we get that $(X'_1, X''_2)$ is $(\eps_{\mathsf{BS}}=4\eps_0)$-close to an $((n'',n'''),k'',k''')$-block source.

    \item\label{it:block-ext-GUV} 
    We will now apply block source extraction to $(X'_1, X''_2)$, which we recall is $(\eps_{\mathsf{BS}}=4\eps_0)$-close to an exact $((n'',n'''),k'',k''')$-block source.
    We instantiate \cref{lem:block-ext} with $\Ext_2$ being the strong extractor from \cref{lem:fasthash-short} with source input length $n'''$, min-entropy requirement $k'''$, error $\eps_{\mathsf{BExt}}=\eps_0$, output length $d$, and $t=16$.
    This requires a seed of length $d_{\mathsf{BExt}}\leq d/16 + C'_0\log(n/\eps_0)$. 
    This instantiation of \cref{lem:fasthash-short} is valid since $k'''\geq (1-1/400)n'''>(1-\frac{1}{20 t})n'''$ and
    \begin{equation*}
        k'''\geq 0.95n'''\geq 0.95k''\geq \frac{0.95k}{3}>\frac{0.95\cdot 8d}{3}>2d,
    \end{equation*}
    where we used the fact that $i(k)>0$, and so $k>8d$.
    For $\Ext_1$ we choose the average-case strong extractor $\Ext_{i(k/3)}$ (recall that $k''\geq k/3$ and note that $i(k/3)<i(k)$) 
    with input length $n''$, entropy requirement $k/3$, error $\eps_{i(k/3)}$, output length at least $(k/3)/3=k/9$, and seed length $d$ guaranteed by the induction hypothesis above. 

\end{enumerate}

\cref{it:condense-GUV,it:split-X,it:condense-again,it:block-ext-GUV} above yield a strong seeded extractor $\Ext'_{i(k)} \colon \bits^n\times\bits^{d'}\to\bits^{m'}$ with min-entropy requirement $k$, error $\eps'=\eps_{i(k/3)}+\eps_{\mathsf{BExt}}+\eps_{\mathsf{BS}}+2\eps_0=\eps_{i(k/3)}+7\eps_0$ (where the $2\eps_0$ term comes from the two fixings of the seeds in the two condensing steps in \cref{it:condense-GUV,it:condense-again}), seed length
\begin{equation*}
    d' = d_{\mathsf{BExt}} + d'_{\mathsf{RS}} + d_{\mathsf{RS}} \leq d/16 + C'\log(n/\eps_0),
\end{equation*}
for some constant $C'>0$,
and output length $m'=k/9$. 

\paragraph{Boosting the output length of $\Ext'_{i(k)}$.}
To conclude the definition of $\Ext_{i(k)}$, we need to increase the output length of $\Ext'_{i(k)}$ from $k/9$ to $k/3$. To that end, we use \cref{lem:boost-output}.
Applying \cref{lem:boost-output} once with $\Ext_1=\Ext'_{i(k_1)}$ with $k_1=k$ and $\Ext_2=\Ext'_{i(k_2)}$ with  $k_2=k-k/9-1=8k/9-1$ and $g=1$ yields a strong $(k,3\eps')$-seeded extractor $\Ext''_{i(k)}$ with output length $(k_1+k_2)/9 \geq k(1-(8/9)^2)-1$ and seed length $2(d/16+C'\log(n/\eps_0))=d/8+2C'\log(n/\eps_0)$.
Applying \cref{lem:boost-output} again with $\Ext_1=\Ext''_{i(k_1)}$ for $k_1=k$ and $\Ext_2=\Ext''_{i(k_2)}$ for $k_2=(8/9)^2 k$ and $g=1$ yields a strong $(k,9\eps')$-seeded extractor with output length $m\geq k(1-(8/9)^4)-1\geq k/3$ and seed length $2(d/8+2C'\log(n/\eps_0))=d/4+4C'\log(n/\eps_0)\leq d$, which we set as $\Ext_{i(k)}$. 
This second invocation of \cref{lem:boost-output} is also valid, since $k_2 = (8/9)^2 k = k - (k(1-(8/9)^2)-1)-1 = k_1-m_1-g$.
Note that the error $\eps_{i(k)}=9\eps'=9\eps_{i(k/3)}+63\eps_0$, as desired.

\paragraph{Time complexity and final error.}
It remains to analyze the time complexity and the overall error of the recursive procedure above.
Evaluating $\Ext_{i(k)}$ requires at most eight evaluations of the condenser from \cref{thm:rscond}, four evaluations of the fast hash-based extractor from \cref{lem:fasthash-short}, four evaluations of $\Ext_{i(k'')}$ for some $i(k'')<i(k)$, and simple operations that can be done in time $\Otilde(n)$. 
This means that the overall time complexity is $4^{i(k)}\cdot\Otilde(n)=\Otilde(n)$ after a one-time preprocessing step independent of the source and seed, since $4^{i(k)}=\poly(\log n)$ by \cref{eq:bound-ik}.
This preprocessing step corresponds to finding primitive elements for $O(\log\log n)$ fields $\F_q$ with prime orders $q\leq \poly(n/\eps_0)=\poly(n/\eps)$.
Furthermore, $\eps_{i(k)}=O(\eps_0+\eps_{i(k/3)})$ for all $k$, and so $\eps_{i(k)}=2^{O(i(k))}\eps_0=\poly(\log n)\cdot \eps_0\leq \eps$ provided that $\eps_0\leq \eps/\log^C n$ for a large enough constant $C>0$. 
\end{proof}

\subsubsection{The (relatively) high-error case} 

In this section, we consider the higher error case where $\eps\geq C n^3\cdot 2^{-k/\log k}$.
We instantiate the recursive approach of Srinivasan and Zuckerman~\cite[Section 5.5]{SZ99} appropriately with the fast condensers from \cref{sec:fast-condensers}, the sampler from \cref{lem:extendedSampler}, and the fast hash-based seeded extractors from \cref{lem:fasthash},  and analyze its complexity.

The next lemma shows how we can recursively decrease the seed length of an extractor.
We complete the construction by instantiating the base extractor in this recursion appropriately, and then increasing its output length.
\begin{lemma}[\protect{analogous to~\cite[Corollary 5.10]{SZ99}, with different instantiation and additional complexity claim}]\label{lem:reducelog}
    There exist constants $c,C>0$ such that the following holds.
    Suppose that for any positive integers $n_0$, $k_0=0.7n_0$, and some $\eps_0=\eps_0(n_0)\geq 2^{-ck_0}$ and $m_0=m_0(n_0)$ there exists a strong $(k_0,\eps_0)$-seeded extractor $\Ext_0 \colon \bits^{n_0}\times\bits^{d_0}\to\bits^{m_0}$ with seed length $d_0 \leq u\cdot \log(n_0/\eps_0)\leq k_0$.
    Then, for any positive integers $n$ and $k\leq n$ there exists a family of strong $(k,\eps)$-seeded extractors $\Ext \colon \bits^n\times\bits^d\to\bits^m$ 
    with error $\eps\leq C \log u\cdot \eps_0(ck)$, seed length $d\leq C\log u \cdot \log(n/\eps_0(ck))$, and output length $m\geq m_0(ck)$. 
    Furthermore,
    \begin{enumerate}
        \item If $\Ext_0$ is computable in time $T(n_0)$ and $k\geq C\log^2(n/\eps_0(ck))$, then $\Ext$ is computable in time $T(n)+\Otilde(n+\sqrt{n}\cdot \log(1/\eps_0(ck))^5)$;
        \label{item:nopreproc}

        \item If $\Ext_0$ is computable in time $T(n_0)$ after a preprocessing step, then $\Ext$ is computable in time $T(n)+\Otilde(n)$ after a 
        preprocessing step.\footnote{We discuss the precise preprocessing step in more detail in \cref{remark:preproc-rec}.}
        \label{item:preproc}
    \end{enumerate}
\end{lemma}
\begin{proof}
    We begin by discussing the high-level approach in this proof.
    On input an arbitrary $(n,k)$-source $X$, $\Ext$ proceeds as follows.
    First, it applies a fast strong condenser to $X$ to obtain a new source $X'$ with high min-entropy rate.
    If $k\geq C\log^2(n/\eps_0)$ then we can apply the KT condenser, which does not require preprocessing unless $\eps_0$ is very small.
    Otherwise, we apply the RS condenser.
    Then, we use $X'$ to generate a block source $Z=(Z_0,\dots, Z_t)$ with geometrically decreasing block lengths.
    The way we achieve this depends on the regime we are in. If we are in a regime where we must anyway resort to the RS condenser, then we generate each block by applying an appropriately instantiated RS condenser with a fresh seed to $X'$.
    Otherwise, if we are in a regime where we can use the KT condenser, then we use the expander random walks averaging sampler from \cref{lem:extendedSampler}, which runs in time $\Otilde(n)$ in this regime.\footnote{The sole reason for this case analysis is that by using the sampler instead of the RS condenser in the ``KT regime'' we can avoid a preprocessing step unless $\eps$ is tiny.}
    Finally, we apply block source extraction to $Z$.
    More concretely, we begin by applying the fast hash-based extractor from \cref{lem:fasthash} to the shorter blocks at the end of $Z$, up until the second block $Z_1$ of $Z$. This generates a sufficiently large (but still short) seed that we can use to extract from the first block $Z_0$ using the base extractor $\Ext_0$.

    We now formally analyze the approach above. 
    We begin by setting up relevant parameters:
    \begin{itemize}
    
        \item Let $C_{\mathsf{blocks}}\geq 1$ be a constant to be determined.
        Set $\ell_0=\frac{k}{100\cdot C_{\mathsf{blocks}}}$ and $k_0=0.7\ell_0$. 
        For $\eps_0=\eps_0(\ell_0)$ and $m_0=m_0(\ell_0)$, we define $\ell_1=C_{\mathsf{blocks}}\cdot u \log(\ell_0/\eps_0)$.  
        Then, we define $\ell_i=0.9\ell_{i-1}$ for all $i\geq 2$.
        The $\ell_i$'s will be block lengths for a block source $Z$.
        In particular, when performing block source extraction from $Z$ we will instantiate $\Ext_0$ with input length $n_0=\ell_0$.

        \item Define $m_1=u\cdot \log(\ell_0/\eps_0)$ and $m_i=0.9m_{i-1}$ for all $i\geq 2$.
        The $m_i$'s will be output lengths for block source extraction from $Z$.
        
        \item Set $t=1+\left\lceil\frac{\log\left(u/\log u\right)}{\log(1/0.9)}\right\rceil$. This will be the number of blocks of $Z$.
        We have $m_t=0.9^{t-1}m_1\in[0.9\log u\cdot \log(\ell_0/\eps_0), \log u\cdot \log(\ell_0/\eps_0)]$.
        Furthermore, since $\ell_1=C_{\mathsf{blocks}}\cdot m_1$, we also have that $\ell_i=C_{\mathsf{blocks}}\cdot m_i$ for all $i\geq 1$.

    \end{itemize}

    Let $X$ be an arbitrary $(n,k)$-source.
    The extractor $\Ext \colon \bits^n\times\bits^d\to\bits^m$ proceeds as follows on input $X$:
    \begin{enumerate}
        \item\label{it:condensing-SZ}
        Using a fresh seed $Y_{\Cond}$ of length $C_{\Cond}\log(n/\eps_0)$, apply a strong $(k,k',\eps_0^2)$-condenser $\Cond$ to $X$.
        If $k\geq C\log^2(n/\eps_0)$ for an apropriately large constant $C>0$, then we instantiate $\Cond$ with the KT strong $(k,k'=k,\eps_0^2)$-condenser (\cref{thm:ktcond}) instantiated with $\alpha=0.05$.
        Otherwise, we instantiate $\Cond$ with the lossy RS $(k,k'\geq 0.975k,\eps_0^2)$-condenser (\cref{thm:rscond}) instantiated with $\alpha=0.025$.  
        By the second part of either \cref{thm:ktcond} or \cref{thm:rscond}, we get that with probability at least $1-\eps_0$ over the choice of $Y_{\Cond}=y$ it holds that $X'=\Cond(X,y)$ is $\eps_0$-close to an $(n',k')$-source with $k'\geq 0.95n'$.
        
        From here onwards we work under such a good fixing $Y_{\Cond}=y$ and also assume that $X'$ is an $(n',k')$-source.
        We account for the resulting $2\eps_0$ error term in the final extractor error later.

        \item \label{it:block-source-SZ}
        We use $X'$ to generate a block source $Z=(Z_0, Z_1,\dots, Z_t)$ with geometrically decreasing block lengths.
        Our procedure depends on the regime of parameters we are in:
        \begin{enumerate}
            \item 
            \label{it:large-k}
            If $k\geq C\log^2(n/\eps_0)$ for an appropriately large constant $C>0$, then for each $i=0,1,\dots,t$ let $\Samp_i \colon \bits^{r_i}\to [n']^{\ell_i}$ be the $(\gamma=\eps_0,\theta=1/100)$-averaging sampler from \cref{lem:extendedSampler} with input length $r_i=\gamma_{\Samp}\log(n'/\eps_0)$ for some constant $\gamma_{\Samp}>0$.
        We choose the constant $C_{\mathsf{blocks}}$ above to be large enough so that
        $n'\geq \ell_i\geq \ell_t\geq  C_{\Samp}\log(1/\eps_0)/\theta^2$ for all $i\in [t]$, 
        where $C_{\Samp}$ is the constant $C$ from \cref{lem:extendedSampler}.
        To see that $\ell_i\leq n'$ for $i=0,1,\dots,t$ (and so indeed \cref{lem:extendedSampler} can be applied to obtain $\ell_i$ samples), note that
        \begin{equation}\label{eq:sum-ell-is}
            \ell_0+\sum_{i=1}^t \ell_i \leq \ell_0+\sum_{i=1}^\infty \ell_i = \ell_0+10\ell_1\leq k/9< n'.
        \end{equation}
        The second-to-last inequality uses the fact that 
        \begin{equation*}
            \ell_1=C_{\mathsf{blocks}} \cdot u\log(\ell_0/\eps_0)\leq C_{\mathsf{blocks}}\cdot k_0\leq C_{\mathsf{blocks}}\cdot \ell_0= k/100,
        \end{equation*}
        where the first inequality holds since $u\log(\ell_0/\eps_0)\leq k_0$ is an hypothesis in the lemma statement.

        By \cref{lem:sample-entropy} instantiated with $X'$ and $\Samp_0$, we conclude that
        \begin{equation}\label{eq:base-case}
            (Y_0, Z_0) \approx_{\eps_0 + 2^{-\beta_{\Samp} k}} (Y_0, Z'_0),
        \end{equation}
        with $\beta_{\Samp}>0$ the constant guaranteed by \cref{lem:sample-entropy}, 
        where $(Z'_0|Y_0=y_0)$ is an $(\ell_0,0.9\ell_0)$-source for every $y_0$.
        We now argue how this guarantee extends to more blocks.
        
        For each $Z_j$, define $Z_{j,\vec{y}}=(Z_j|(Y_0,\dots,Y_{i-1})=\vec{y})$.
        Consider any fixing $(Y_0,\dots,Y_{i-1})=\vec{y}$.
        Then, \cref{lem:tailboundchainrule} with $\delta = 2^{-\beta_{\Samp} k}$,  where $\beta_{\Samp}>0$ is taken to be a small enough constant, and $\ell=k/9$ (from the upper bound in \cref{eq:sum-ell-is}) implies that
        \begin{align}\label{eq:LB-ent-X'}
            \minH(X'|Z_{0,\vec{y}}=z_0,\dots,Z_{i-1,\vec{y}}=z_{i-1})&\geq k' - \ell - \beta_{\Samp} k \nonumber\\
            &\geq 0.95n' - k/9 - \beta_{\Samp} k \nonumber\\
            &\geq 0.95n'-n'/9-\beta_{\Samp} n' \nonumber\\
            &\geq  0.8n'
        \end{align}
        except with probability at most $2^{-\beta_{\Samp} k}$ 
        over the choice $(z_0,\dots,z_{i-1})\sim(Z_{0,\vec{y}},\dots,Z_{i-1,\vec{y}})$.
        Call a fixing $\vec{v}=(y_0,z_0,\dots,y_{i-1},z_{i-1})$ for which \cref{eq:LB-ent-X'} holds \emph{good}.
        Define $X'_{\vec{v}}=(X'|(Y_0,Z_0,\dots,Y_{i-1},Z_{i-1})=\vec{v})$.
        Then, by \cref{eq:LB-ent-X'} and \cref{lem:sample-entropy} we know that for all good $\vec{v}$-s we have
        \begin{equation}\label{eq:new-block}
            \left(Y_i, Z_{i,\vec{v}} = (X'_{\vec{v}})_{\Samp_i(Y_i)}\right)\approx_{\eps_0+2^{-\beta_{\Samp}k}} \left( Y_i, Z'_{i,\vec{v}}\right),
        \end{equation}
        with $(Z'_{i,\vec{v}}|Y_i=y_i)$ an $(\ell_i,0.7\ell_i)$-source for all $y_i$.

        Analogously to~\cite[proof of Lemma 17]{NZ96} and the proof of \cref{lem:short}, we can use \cref{eq:new-block} and the fact that  $\vec{v}\sim (Y_0,Z_{0,\vec{y}},\dots,Y_{i-1},Z_{i-1,\vec{y}})$ is good with probability at least $1-2^{-\beta_{\Samp} k}$ to show by induction on the number of blocks that
        \begin{equation}\label{eq:close-block}
           (Y_0,\dots,Y_t,Z) \approx_{\eps_{\mathsf{block}}} (Y_0,\dots,Y_t, Z''),
        \end{equation}
        where for every $(y_0,\dots,y_t)$,  $(Z''|Y_0=y_0,\dots,Y_t=y_t)$ is an exact $((\ell_0,\dots,\ell_t),0.7)$-block-source, and $\eps_{\mathsf{block}}=(t+1)(\eps_0+2\cdot 2^{-\beta_{\Samp} k})$.

        \item \label{it:small-k}
        If $k< C\log^2(n/\eps_0)$, then for each $i=0,1,\dots,t$ let $\Cond_i\colon\bits^{n'}\to\bits^{m_i}$ be the strong RS $(k_i=0.9\ell_i,k'_i,\eps_0^2)$-condenser from \cref{thm:rscond} instantiated with $\alpha=0.01$.
        Note that $k_i\leq m_i\leq (1+\alpha)k_i\leq \ell_i$ and $k'_i\geq (1-\alpha)m_i\geq 0.99k_i$.
        Using a fresh seed $Y_i$ of length at most $C'_\alpha \log(n'/\eps_0^2)\leq 2C'_\alpha\log(n'/\eps_0)$ with $C'_\alpha$ the constant guaranteed by \cref{thm:rscond} for $\alpha=0.01$, we compute $W_i =\Cond_i(X',Y_i)$ and obtain $Z_i$ by padding $W_i$ to get length exactly $\ell_i$, i.e., $Z_i=(W_i, 0^{\ell_i-m_i})$. This is valid since, as already pointed out, $m_i\leq\ell_i$.
        Later we argue that, despite this padding, $Z_i$ will be statistically close to a source $Z'_i$ with sufficiently large min-entropy rate.

        The argument showing that $Z_0, Z_1,\dots, Z_t$ is close to an exact $((\ell_0,\ell_1,\dots,\ell_t),0.7)$-block-source conditioned on the seeds $Y_1,\dots,Y_t$ is very similar to that of the previous case.
        Nevertheless, we do need to check that the choices of $\ell_0,\dots,\ell_t$ allow the desired applications of \cref{thm:rscond}.

        First, to apply \cref{thm:rscond} we need that $k_i\geq C_\alpha \log(n'/\eps_0^2)$ for all $i$, with $C_\alpha>0$ the constant from \cref{thm:rscond} for $\alpha=0.01$. Since $k_i=0.9\ell_i\geq 0.9\ell_t=k_t$ for all $i$, it suffices to show that $k_t\geq C_\alpha \log(n'/\eps_0^2)$.
        Since $n'\leq 2k$ and
        \begin{equation*}
            k_t=0.9\ell_t=0.9^t \ell_1 \geq 0.9^2 C_{\mathsf{blocks}}\log(\ell_0/\eps_0)= 0.9^2 C_{\mathsf{blocks}}\log\left(\frac{k}{100 C_{\mathsf{blocks}}\eps_0}\right),
        \end{equation*}
        it is enough to guarantee that
        \begin{equation*}
            0.9^2 C_{\mathsf{blocks}}\log\left(\frac{k}{100 C_{\mathsf{blocks}}\eps_0}\right) \geq 2 C_\alpha \log(2k/\eps_0).
        \end{equation*}
        If we take $C_{\mathsf{blocks}}$ to be large enough so that $C_{\mathsf{blocks}}\geq 5
        C_\alpha$, the desired inequality holds 
        provided that, say, $k\geq 2(100 C_{\mathsf{blocks}})^2$, which is a constant lower bound on $k$.

        Analogously to \cref{eq:LB-ent-X'} in the previous case, for any $i=0,1,\dots,t$ and an arbitrary fixing $\vec{y}=(y_0,\dots,y_{i-1})$ we have that
        \begin{equation*}
            \minH(X'|Z_{0,\vec{y}}=z_0,\dots,Z_{i-1,\vec{y}}=z_{i-1})\geq 0.8n'
        \end{equation*}
        except with probability at most $2^{-\beta k}$ over the choice of $z_0,\dots,z_{i-1}$, for a sufficiently small constant $\beta>0$.
        Therefore, under such a good fixing $\vec{v}=(y_0,z_0,\dots,y_{i-1},z_{i-1})$, and defining $W_{i,\vec{v}}$ to be $W_i$ conditioned on $(Y_0,Z_0,\dots,Y_{i-1},Z_{i-1})=\vec{v}$, \cref{thm:rscond} guarantees that  $(Y_i,W_{i,\vec{v}})\approx_{\eps_0} (Y_i,W'_{i,\vec{v}})$ with $(W'_{i,\vec{v}}|Y_i=y_i)$ an $(m_i,k'_i)$-source for all $y_i$ provided that $0.8n'\geq k_i=0.9\ell_i$.
        This holds, since
        \begin{equation*}
            0.8n' \geq 0.8k \geq \frac{k}{100 C_{\mathsf{blocks}}} = \ell_0 \geq 0.9\ell_i = k_i
        \end{equation*}
        for any $i=0,1,\dots,t$.
        By padding $W_{i,\vec{v}}$ and $W'_{i,\vec{v}}$ with $0^{\ell_i-m_i}$, we get that $(Y_i,Z_{i,\vec{v}})\approx_{\eps_0} (Y_i,Z'_{i,\vec{v}})$ with $(Z'_{i,\vec{v}}|Y_i=y_i)$ an $(\ell_i,k'_i)$-source for all $y_i$.
        Furthermore, the min-entropy of $Z'_{i,\vec{v}}$ satisfies
        \begin{equation*}
            k'_i\geq (1-\alpha)k_i= 0.99 k_i = 0.99\cdot 0.9\ell_i\geq 0.7\ell_i,
        \end{equation*}
        and so $(Z'_{i,\vec{v}}|Y_i=y_i)$ is an $(\ell_i,0.7\ell_i)$-source for all $y_i$.
        As in the previous case, 
        this can be used to conclude by induction that
        \begin{equation*}
           (Y_0,\dots,Y_t,Z) \approx_{\eps'_{\mathsf{block}}} (Y_0,\dots,Y_t, Z''),
        \end{equation*}
        where for every fixing $(y_0,\dots,y_t)$ we have that $(Z''|Y_0=y_0,\dots,Y_t=y_t)$ is an exact $((\ell_0,\dots,\ell_t),0.7)$-block-source, and $\eps'_{\mathsf{block}}=(t+1)(\eps_0+2^{-\beta k})$. 
        We choose $\beta=\beta_{\Samp}$ as in \cref{it:large-k}, and so $\eps'_{\mathsf{block}}\leq \eps_{\mathsf{block}}$.
        
        \end{enumerate}

        \item \label{it:block-ext-SZ}
        We apply block source extraction (\cref{lem:block-ext}) to $Z=(Z_0, Z_1, \dots, Z_t)$.
        More precisely, let $\BExt \colon \bits^{\ell_0}\times\cdots\times\bits^{\ell_t}\times\bits^{d_t}\to\bits^{m_0}$ be the strong $(k_0,k_1,\dots,k_t,(t+1)\eps_0)$-block-source extractor with $k_i=0.7\ell_i$ obtained via \cref{lem:block-ext} as follows. 
        We instantiate $\Ext_0$ with the strong extractor promised by the lemma statement with seed length $d_0\leq u\cdot \log(\ell_0/\eps_0) = m_1$.
        For $i\in[t]$, we instantiate $\Ext_i \colon \bits^{\ell_i}\times \bits^{d_i}\to\bits^{m_i}$ as the strong $(k_i=0.7\ell_i,\eps_0)$-seeded extractor from \cref{lem:fasthash} with seed length $d_i= 2m_i + 4\log(\ell_i/\eps_0)+8$. 
        We choose the constant $C_{\mathsf{blocks}}$ to be large enough so that
        \begin{equation*}
            m_i = \ell_i/C_{\mathsf{blocks}} \leq 0.7\ell_i -16\log(4/\eps_0)=k_i-16\log(4/\eps_0),
        \end{equation*}
        as required by \cref{lem:fasthash}. 
        This is possible since by choosing $C_{\mathsf{blocks}}$ large enough we have
        \begin{equation*}
            \ell_i\geq \ell_t = C_{\mathsf{blocks}} \cdot m_t\geq 0.9\cdot C_{\mathsf{blocks}}\log u\cdot \log(\ell_0/\eps_0) \geq 100\log(4/\eps_0)
        \end{equation*}
        for all $i\in[t]$, and so $0.7\ell_i -16\log(4/\eps_0)\geq \ell_i/2$ for all $i\in[t]$.
        Furthermore, for any $i\geq 2$ the output length $m_i$ of $\Ext_i$ satisfies 
        \begin{equation*}
            d_i+m_i = 3m_i + 4\log(n/\eps_0)+8 \geq 2m_{i-1}+4\log(n/\eps_0)+8 \geq d_{i-1},
        \end{equation*}
        where we recall that $m_i=0.9 m_{i-1}$ for $i\geq 2$.
        Finally, the output length of $\Ext_1$ satisfies $d_1+m_1\geq m_1\geq d_0$, where we recall that $d_0$ is the seed length of $\Ext_0$.

        Let $Y_{\mathsf{BExt}}$ be a fresh seed of length $d_t$.
        With the desired upper bound on the seed length $d$ from the lemma's statement in mind, we note that
        \begin{equation}\label{eq:dt-UB}
            d_t \leq 2m_t +4\log(\ell_t/\eps_0)+8 \leq 2\log u\cdot \log(\ell_0/\eps_0) + 4\log(\ell_0/\eps_0)\leq 6\log u\cdot \log(n/\eps_0),
        \end{equation}
        since $\ell_0\leq k\leq n$.
        By \cref{lem:block-ext}, we get that
        \begin{align*}
            \left(Y_0,\dots,Y_t,Y_{\mathsf{BExt}},\BExt(Z,Y_{\mathsf{BExt}})\right) &\approx_{\eps_{\mathsf{block}}} \left(Y_0,\dots,Y_t,Y_{\mathsf{BExt}}, \BExt(Z'',Y_{\mathsf{BExt}})\right) \\
            &\approx_{(t+1)\eps_0}
            \left(Y_0,\dots,Y_t,Y_{\mathsf{BExt}},U_{m_0}\right).
        \end{align*}
        Applying the triangle inequality, we conclude that
        \begin{equation*}
            \left(Y_0,\dots,Y_t,Y_{\mathsf{BExt}}, \BExt(Z,Y_{\mathsf{BExt}})\right) \approx_{\eps_{\mathsf{block}}+(t+1)\eps_0} \left( Y_0,\dots,Y_t,Y_{\mathsf{BExt}},U_{m_0}\right).
        \end{equation*}
    \end{enumerate}

    We now define our final strong extractor $\Ext \colon \bits^n\times\bits^d\to\bits^{m_0}$ (recall that we abbreviate $m_0=m_0(\ell_0)$). 
    Choose our overall seed to be $Y=(Y_{\mathsf{Cond}}, Y_0, Y_1, \dots, Y_t, Y_{\mathsf{BExt}})$ and set $\Ext(X,Y)=\BExt(Z,Y_{\mathsf{BExt}})$.
        By the discussion above, $\Ext$ is a strong $(k,\eps)$-extractor with error (recall that we abbreviate $\eps_0=\eps_0(\ell_0)$)
        \begin{equation*}
            \eps = 2\eps_0+\eps_{\mathsf{block}}+(t+1)\eps_0 \leq (2t+4)(\eps_0+2\cdot 2^{-\beta_{\Samp} k}).
        \end{equation*}
        As discussed above, the $2\eps_0$ accounts for fixing the seed in the condensing step of \cref{it:condensing-SZ} and for assuming that $X'$ is an $(n',k')$-source under this fixing.
        Since $t=O(\log u)$, if we pick $C>0$ to be a sufficiently large constant and $c$ to be smaller than $\beta_\Samp$ so that $\eps_0\geq 2^{-ck_0}\geq 2^{-ck}\geq 2^{-\beta_\Samp k}$, we get
        \begin{equation*}
            \eps \leq C \log u \cdot \eps_0.
        \end{equation*}
        The seed length is
        \begin{equation*}
            d = |Y_{\mathsf{Cond}}|+|Y_{\mathsf{BExt}}|+\sum_{i=0}^t |Y_i| \leq C_\Cond \log(n/\eps_0)+d_t + t\cdot \gamma\log(n'/\eps_0)
            \leq
            C\log u\cdot \log(n/\eps_0),
        \end{equation*}
        where $\gamma=\max(C'_\alpha,\gamma_{\Samp})$,
        provided that $C$ is large enough (again since $t=O(\log u)$), as desired.
        We used \cref{eq:dt-UB} to bound $d_t$ and obtain the last inequality.

        \paragraph{Time complexity.}
        It remains to analyze the time complexity of $\Ext$.
        We proceed by cases:
        \begin{itemize}
            \item If $k\geq C\log^2(n/\eps_0)$ 
            with $C$ a sufficiently large constant, then by \cref{thm:ktcond} we get that \cref{it:condensing-SZ} either takes time $\Otilde(n+\sqrt{n}\cdot \log(1/\eps_0)^5)$, or time $\Otilde(n)$ after a one-time preprocessing step.
            Regarding \cref{it:block-source-SZ}, each averaging sampler $\Samp_i$ from \cref{lem:extendedSampler} runs in time $\log^2(1/\eps_0)\cdot\polylog n + O(\ell_i \log n)$. 
            This is $\Otilde(n)$ when $\eps_0\geq 2^{-\Otilde(\sqrt{n})}$, which is implied by the constraint $k\geq C\log^2(n/\eps_0)$.
            Since there are $t=O(\log u)=O(\log n)$ blocks, \cref{it:block-source-SZ} runs in $\Otilde(n)$ times.
            \cref{it:block-ext-SZ} takes time $T(\ell_0)+t\cdot\Otilde(n)=T(\ell_0)+\Otilde(n)\leq T(n)+\Otilde(n)$, since $\Ext_0$ is computable in time $T(\ell_0)$, each $\Ext_i$ from \cref{lem:fasthash} are computable in time $\Otilde(n)$, and $\ell_0\leq n$.
            Therefore, in this case $\Ext$ is computable in overall time $T(n)+\Otilde(n+\sqrt{n}\cdot \log(1/\eps_0)^5)$, or time $T(n)+\Otilde(n)$ after a preprocessing step.

            \item  If $k< C\log^2(n/\eps_0)$, then \cref{it:condensing-SZ} takes time $\Otilde(n)$ after a preprocessing step.
            In this case, \cref{it:block-source-SZ} amounts to $t=O(\log n)$ applications of \cref{thm:rscond}, and so runs in time $\Otilde(n)$ after a preprocessing step.
        \cref{it:block-ext-SZ} takes time $T(\ell_0)+\Otilde(n)$ after a preprocessing step,
        and so $\Ext$ is computable in overall time $T(\ell_0)+\Otilde(n)\leq T(n)+\Otilde(n)$ after a
        preprocessing step.
        \end{itemize}

\end{proof}

Denote by $\log^{(i)}$ the function that iteratively applies $\log$ a total of $i$ times (so $\log^{(1)}\!n=\log n$, $\log^{(2)}\!n=\log\log n$, and so on).
Denote by $\log^{*}$ the iterated logarithm.
Then, we have the following corollary.

\begin{corollary}\label{coro:Ext-bad-output}
    There exists a constant $C>0$ such that the following holds.
    Let $n$ be any positive integer and $i$ any positive integer such that $\log^{(i)}\!n\geq 6C$.
    Then, for any $k\leq n$ and any $\eps\geq n^3\cdot 2^{-k/2^{C\cdot i}}$ there exists a strong $(k,\eps)$-seeded extractor $\Ext \colon \bits^n\times\bits^d\to\bits^m$ with seed length $d\leq C\log^{(i)}\!n\cdot \log(n/\eps)$ and output length $m\geq k/2^{C\cdot i}$.
    Furthermore,
    \begin{enumerate}
        \item if $k \geq 2^{C\cdot i }\cdot \log^2(n/\eps)$ and $\eps\geq 2^{-Cn^{0.1}}$, then $\Ext$ is computable in time $\Otilde(n)$;

        \item if $k < 2^{C\cdot i }\cdot \log^2(n/\eps)$ or $\eps< 2^{-Cn^{0.1}}$, then $\Ext$ is computable in time $\Otilde(n)$ after a 
        preprocessing step.
    \end{enumerate}

    Consequently, if we choose $i$ to be the largest integer such that $\log^{(i)}\!n\geq 6C$ (which satisfies $i\leq \log^{*}\!n$) we get a strong $(k,\eps)$-seeded extractor $\Ext \colon \bits^n\times\bits^d\to\bits^m$ with seed length $d\leq 6C^2 \log(n/\eps)$ 
    and output length $m\geq k/2^{C\log^{*}\!n}$ for any error $\eps\geq n^3\cdot 2^{-k/2^{C\log^{*}\!n}}$.
    If $k \geq 2^{C\log^{*}\!n}\cdot \log^2(n/\eps)$ and $\eps\geq n^3 \cdot 2^{-Cn^{0.1}}$, then $\Ext$ is computable in time $\Otilde(n)$.
    Otherwise, $\Ext$ is computable in time $\Otilde(n)$ after a preprocessing step. 
\end{corollary}
\begin{proof}
    This is a consequence of iteratively applying \cref{lem:reducelog} $i$ times.
    Note that here part of the relevant condition for the preprocessing is $k\geq 2^{Ci}\log^2(n/\eps)$.
    The reason behind this is that each application of \cref{lem:reducelog} reduces the min-entropy requirement by a constant factor.

    Let $c,C>0$ be the constants guaranteed by \cref{lem:reducelog}.
    For the first application of the lemma, we take $\Ext_0 \colon \bits^{n}\times\bits^{d_0}\to\bits^{m_0}$ to be the strong $(k_0=0.7n,\eps_0)$ extractor from \cref{lem:fasthash} with $m_0=k_0/20$ and $\eps_0\geq 2^{-ck_0/100}$ to be defined later.
    The corresponding seed length is $d_0 \leq 2m_0 + 4\log(n/\eps_0)+4$, which satisfies $d_0\leq k_0$, and so the initial value of $u$ is
    $u_0=d_0/\log(n/\eps_0)\leq k_0$.
    Denote by $\Ext_1$ the resulting strong seeded extractor.
    In the second application of \cref{lem:reducelog}, we instantiate $\Ext_0$ with $\Ext_1$ instead to obtain a new strong seeded extractor $\Ext_2$, and so on.
    For each $j\in[i]$, we obtain a family of strong $(k,\eps_j)$-seeded extractors $\Ext_j \colon \B^n\times \B^{d_j}\to \B^{m_j}$ parameterized by $k$ with output length $m_j=m_{j-1}(ck)$,
    error
    \begin{equation*}
        \eps_j = C \log u_{j-1}\cdot \eps_{j-1}(ck)
    \end{equation*}
    and seed length
    \begin{align*}
        d_j = C \log u_{j-1}\cdot \log(n/\eps_{j-1}(ck)) =C \log u_{j-1}\cdot \log\left(\frac{n \cdot C\log u_{j-1}}{\eps_j}\right),
    \end{align*}
    where
    \begin{align*}
        u_j &= \frac{d_j}{\log(n/\eps_j)} \\
        &= C \log u_{j-1} \cdot \left(1+ \frac{\log C}{\log(n/\eps_j)}+\frac{\log\log u_{j-1}}{\log(n/\eps_j)}\right) \\
        &\leq C \log u_{j-1} \cdot \left(1+ \frac{\log C}{\log n}+\frac{\log\log u_{j-1}}{\log n}\right) \\
        &\leq 3C \log u_{j-1}.
    \end{align*}
    The last inequality uses the fact that $u_{j-1}\leq u_0\leq n$ for all $j$.
    
    Recall that from the corollary statement that $i$ is such that $\log^{(i)}\!n \geq 6C$.
    We show by induction that $u_j \leq 3C \log^{(j)} n+3C\log(6C)$ for all $j=0,\dots,i$.
    This is immediate for the base case $j=0$, since $u_0\leq k_0\leq n$.
    For the induction step, note that
    \begin{multline*}
        u_{j+1}\leq 3C\log u_j \leq 3C\log(3C\log^{(j)}n+3C\log(6C))\\ \leq 3C\log(2\cdot 3C\log^{(j)}n)= 3C\log^{(j+1)}n+3C\log(6C),
    \end{multline*}
    as desired.
    This implies that
    \begin{equation*}
        d_j = u_j \cdot \log(n/\eps_j) \leq 6C\log^{(j)} n\cdot \log(n/\eps_j)
    \end{equation*}
    and
    \begin{equation*}
        \eps_j = C\log u_{j-1}\cdot \eps_{j-1}(ck) \leq (6C)^j \left(\prod_{j'=0}^{j-1} \log^{(j')} n\right)\cdot \eps_0(c^j k) 
    \end{equation*}
    for all $j\in[i]$.
    We may assume that $C$ is large enough that $\log a \leq \sqrt{a}$ for all $a\geq C$, in which case $\prod_{j'=0}^{j-1} \log^{(j')} n\leq \prod_{j'=0}^{j-1} n^{2^{-j'}}\leq n^2$ since $\log^{(j')} n\geq C$ for all $j'\leq i$ by hypothesis.
    Therefore, we obtain final output length
    \begin{equation*}
        m_i = m_0(c^i k) = k/2^{O(i)},
    \end{equation*}
    final error $\eps_i$ satisfying
    \begin{equation*}
        \eps_0(ck)\leq \eps_i \leq (6C)^i \cdot n^2 \cdot \eps_0(c^i k) \leq n^3 \cdot \eps_0(c^i k),
    \end{equation*}
    where the last inequality uses that $\log^{(i)}\!n \geq 6C$,
    and final seed length
    \begin{equation*}
        d_i \leq 6C\log^{(i)}\!n\cdot \log(n/\eps_i).
    \end{equation*}

     We now instantiate $\eps_0(c^i k) = \eps/n^3$.
    Note that $\eps_0(c^i k)\geq 2^{-0.7c^{i+1} k/100}$ as required for the choice of $\Ext_0$ above so long as $\eps\geq n^3 \cdot 2^{-0.7c^{i+1}k}$, which holds by the corollary's hypothesis if $C$ is a large enough constant. 
    With this choice of $\eps_0(c^i k)$ we get final error $\eps_i\leq n^3 \cdot \eps_0(c^i k)=\eps$.
    In fact, we can make $\eps_i$ larger so that $\eps_i=\eps$, in which case the final seed length satisfies
    \begin{equation*}
        d_i \leq 6C \log^{(i)}\!n\cdot \log(n/\eps),
    \end{equation*}
    as desired.

    \paragraph{Time complexity.}
    Finally, we discuss the time complexity of $\Ext$.
    Note that the initial choice for $\Ext_0$ is computable in time $\Otilde(n_0)$.
    Therefore, if $k\geq 2^{C \cdot i}\log^2(n/\eps)$ then each application of \cref{lem:reducelog} runs in time $\Otilde(n+\sqrt{n}\log(1/\eps_0(c^i k))^5)$.
    This uses the fact that the error increases in each application of \cref{lem:reducelog}.
    Since $\eps_0(c^i k) \geq \eps/n^3$, then each application of \cref{lem:reducelog} runs in time $\Otilde(n)$ without preprocessing when $\eps\geq 2^{-C n^{0.1}}$.
    Otherwise, the condition in \cref{item:preproc} of \cref{lem:reducelog} holds and so $\Ext$ is computable in time $\Otilde(n)$ after a preprocessing step, since we always have $u\leq n$ in each application of the lemma.
\end{proof}

\begin{remark}[the preprocessing in \cref{coro:Ext-bad-output}]
\label{remark:preproc-rec}
    \em
    For the sake of readability we were not explicit about the precise preprocessing in the statement of Case 2 of \cref{coro:Ext-bad-output}.
    We expand on that now.
    \cref{coro:Ext-bad-output} recursively invokes \cref{lem:reducelog} $i$ times.
    All these recursive calls use the same type of condenser (either the KT condenser or the RS condenser), depending on the initial choices of $k$ and $\eps$.
    Therefore:
    \begin{itemize}
        \item If $k \geq 2^{C\cdot i }\cdot \log^2(n/\eps)$ and $\eps< 2^{-Cn^{0.1}}$, then the preprocessing corresponds to the preprocessing for $i$ calls of the KT condenser -- for example, generating $i$ primes $q\leq \poly(n/\eps)$.

        \item If $k < 2^{C\cdot i }\cdot \log^2(n/\eps)$ and $\eps< 2^{-Cn^{0.1}}$,
        then each invocation of \cref{lem:reducelog} requires one preprocessing step for the RS condenser call in \cref{it:condensing-SZ},
        and, naively, $t$ preprocessing steps for the $t=O(\log n)$ RS condenser calls in \cref{it:small-k}.
        We can further improve this by noting that all $t$ calls of \cref{thm:rscond} in \cref{it:small-k} use the same input length, error, and $\alpha$, and so we only need to run one preprocessing step that suffices for all $t$ calls simultaneously.
        So, overall, the preprocessing corresponds to the preprocessing for $2i$ calls of the RS condenser -- generating $2i$ primes $q\leq \poly(n/\eps)$ along with a primitive element for each $\F_q$. 
    \end{itemize}
\end{remark}

Note that \cref{coro:Ext-bad-output} only guarantees output length $k/2^{Ci}$ for each $i$.
In particular, when $i=\log^*\! n$ we get output length $k/2^{C\log^*\! n}$.
This is slightly sublinear, and we would like to aim for output length $ck$ for some constant $c>0$.
To obtain our final theorem, we use block source extraction to increase the output length of the extractor from \cref{coro:Ext-bad-output}, following a strategy of Zuckerman~\cite{Zuc97}.
\begin{theorem}
    There exist constants $c,C>0$ such that the following holds.
    For any integers $n$ and $k\leq n$ and any $\eps\geq Cn^3\cdot 2^{-k/\log k}$ 
    there exists a strong $(k,\eps)$-seeded extractor $\Ext \colon \bits^n\times\bits^d\to\bits^m$ with seed length $d\leq C\log(n/\eps)$ and output length $m\geq ck$.
    Furthermore,
    \begin{enumerate}
        \item if $k \geq 2^{C\log^{*}\!n}\cdot \log^2(n/\eps)$ and $\eps\geq 2^{-C n^{0.1}}$, then $\Ext$ is computable in time $\Otilde(n)$;

        \item if $k < 2^{C\log^{*}\!n}\cdot \log^2(n/\eps)$ or $\eps < 2^{-C n^{0.1}}$, then $\Ext$ is computable in time $\Otilde(n)$ after a preprocessing step. 
    \end{enumerate}
\end{theorem}
\begin{proof}
    We begin by providing an informal discussion of the proof.
    On input an arbitrary $(n,k)$-source $X$, we begin by applying a fast condenser to $X$ to obtain another source $X'$ with high min-entropy rate.
    Then, we split $X'$ in half to obtain a block source $(X_1, X_2)$ with two blocks.
    Then, we perform block source extraction on $(X_1, X_2)$.
    More concretely, we apply the extractor obtained by instantiating \cref{coro:Ext-bad-output} with $i=\log^*\! n$ to $X_2$, and then use its output as the seed to extract from $X_1$ using the extractor obtained by instantiating \cref{coro:Ext-bad-output} with $i=2$ (which has output length $\Omega(k)$).

    More formally,
    define $\eps' = \eps/6$ and let $X$ be an arbitrary $(n,k)$-source.
    The extractor $\Ext$ behaves as follows on input $X$:
    \begin{enumerate}
        \item \label{it:condensing-SZ-final}
        
        Apply a strong $(k,k',(\eps')^2)$-condenser $\Cond\colon \B^n\times \B^{d_{\mathsf{Cond}}}\to\B^{n'}$ to $X$, with output min-entropy $k'\geq 0.95n'$ and seed length $d_{\mathsf{Cond}}=C_\Cond\log(n/\eps')$.
        If $k \geq 2^{C\log^{*}\!n}\cdot \log^2(n/\eps)$, we instantiate $\Cond$ with the KT strong $(k,k',\eps')$-condenser (\cref{thm:ktcond} instantiated with $\alpha=0.05$).
        Otherwise, we instantiate $\Cond$ with the RS strong $(k,k',\eps')$-condenser (\cref{thm:rscond} instantiated with $\alpha=0.025$).
        By the second part of either \cref{thm:ktcond} or \cref{thm:rscond}, we get that with probability at least $1-\eps'$ over the choice of the seed $y$ we obtain an output $X'$ that is $\eps'$-close to an $(n',k')$-source with $k'\geq 0.95n'$.
        As in previous arguments, we work under such a good fixing of $y$ from here onwards and account for the probability $\eps'$ of selecting a bad seed in the final extractor error later on.

        \item 
        Let $(X_1, X_2)$ correspond to the first two blocks of $\lfloor n'/2\rfloor$ bits of $X'$.
        Choose the constant $c>0$ in the theorem statement small enough so that $\log(1/\eps')\leq \log(1/\eps)+3\leq ck+3\leq 0.05k-1$, which means that $\lfloor n'/2\rfloor -0.05k-\log(1/\eps') \geq 0.4n'$.
        Then, combining \cref{it:condensing-SZ-final} with \cref{lem:block-chain-rule} (instantiated with $t=2$, $\Delta=0.05k$, and $\eps=\eps'$) via the triangle inequality, $(X_1,X_2)$ is $3\eps'$-close to an exact $((n_1=\lfloor n'/2\rfloor,n_2=\lfloor n'/2\rfloor),0.8)$-block-source.

        \item \label{it:block-ext-SZ-final}
        Apply block source extraction to $(X_1, X_2)$.
        More precisely, let $\Ext_1\colon \bits^{n_1}\times\bits^{d_1}\to\bits^{m_1}$ be the strong $(k_1=0.8n_1,\eps_1=\eps')$-seeded extractor from \cref{coro:Ext-bad-output} instantiated with $i=2$ and 
        $n_1=\lfloor n'/2\rfloor $, which requires $\eps_1=\eps'\geq n_1^3 \cdot 2^{-c_1 k_1}$ and guarantees $d_1\leq C_1\log\log k_1\cdot \log(n'/\eps)$ and $m_1\geq c_1 k_1$, for constants $c_1,C_1>0$ guaranteed by \cref{coro:Ext-bad-output}.
        Furthermore, let $\Ext_2\colon \bits^{n_2}\times\bits^{d_2}\to\bits^{m_2}$ be the strong $(k_2=0.8n_2,\eps_2=\eps')$-seeded extractor from the ``Consequently'' part of \cref{coro:Ext-bad-output} and $n_2=\lfloor n'/2\rfloor $, which requires $\eps_2\geq n_2^3 \cdot 2^{-k_2/2^{C_2\log^{*}\! k_2}}$ and guarantees $d_2\leq C_2\log(n'/\eps)$ and $m_2\geq k_2/2^{C_2\log^{*}\! k_2}$, for a constant $C_2>0$ guaranteed by \cref{coro:Ext-bad-output}. 
        This choice of parameters ensures that $m_2\geq d_1$ and is valid by the lower bound on $\eps$ in the theorem statement, recalling that $\eps'=\eps/6$.
        Indeed, since $k\geq k_1=k_2\geq 0.4n'$, to see that $m_2\geq d_1$ it suffices to check that
        \begin{equation*}
            \frac{0.4k}{2^{C_2\log^{*}\! k}}\geq d_1=C_1 \log\log k \cdot \log(n'/\eps_1).
        \end{equation*}
        Since $\eps_1=\eps'=\eps/6$ and $\log(n'/\eps_1)=O(\log(k/\eps'))=O(\log k + k/\log k)=O(k/\log k)$, it is enough that
        \begin{equation*}
            k \geq C'_1\cdot 2^{C_2\log^{*}\! k}\log\log k\cdot \frac{k}{\log k}
        \end{equation*}
        for a sufficiently large constant $C'_1>0$, which holds whenever $k$ is larger than some appropriate absolute constant.
        Instantiating \cref{lem:block-ext} with $\Ext_1$ and $\Ext_2$ above yields a strong $(k_1=0.8n_1,k_2=0.8n_2,\eps_1+\eps_2)$-block-source extractor
        $\BExt\colon \bits^{n_1}\times\bits^{n_2}\times\bits^{d_2}\to\bits^{m_1}$.

        Since $X'$ is $3\eps'$-close to an exact $(n_1,n_2,0.8)$-block source, we conclude that
        \begin{equation}\label{eq:final-ext-SZ}
            \left( Y_{\mathsf{BExt}},\BExt(X',Y_{\mathsf{BExt}}) \right) \approx_{3\eps' + \eps_1+\eps_2} U_{d_2+m_1}.
        \end{equation}
    \end{enumerate}

    We define the output of our final strong extractor $\Ext \colon \bits^n\times \bits^d\to\bits^{m_1}$ to be $\BExt(X',Y_{\mathsf{BExt}})$.
    Since $\eps_1=\eps_2=\eps'$, \cref{eq:final-ext-SZ} implies that
    \begin{equation*}
       \left( Y_{\mathsf{Cond}}, Y_{\mathsf{BExt}}, \Ext(X,Y_{\mathsf{Cond}}, Y_{\mathsf{BExt}}) \right) \approx_{5\eps'} U_{d+m_1}.
    \end{equation*}
    This means that $\Ext$ is a strong $(k,\eps'+5\eps'=\eps)$-seeded extractor with seed length $d=|Y_{\mathsf{Cond}}|+|Y_{\mathsf{BExt}}|=O(\log(n/\eps))$ and output length $m_1\geq c_1 k_1 \geq c'_1 k$ for an absolute constant $c'_1>0$, where one of the $\eps'$ terms in the error comes from fixing the seed in the condensing step of \cref{it:condensing-SZ-final}.

    \paragraph{Time complexity.}
    Finally, we analyze the time complexity of $\Ext$.
    If $k \geq 2^{C\log^{*}\!n}\cdot \log^2(n/\eps)$ and $\eps\geq 2^{-C n^{0.1}}$, then
    \cref{it:condensing-SZ-final} runs in time $\Otilde(n)$.
    In \cref{it:block-ext-SZ-final}, $\Ext_1$ and $\Ext_2$ are both computable in time $\Otilde(n)$ under these lower bounds on $k$ and $\eps$,
    and thus so is $\BExt$.
    We conclude that $\Ext$ runs in time $\Otilde(n)$.

    Otherwise, if $k < 2^{C\log^{*}\!n}\cdot \log^2(n/\eps)$, then \cref{it:condensing-SZ-final} runs in time $\Otilde(n)$ after a 
    preprocessing step. And if $\eps< 2^{-C n^{0.1}}$ then $\Ext_1$ and $\Ext_2$ in \cref{it:block-ext-SZ-final} run in time $\Otilde(n)$ after a 
    preprocessing step.
    Therefore, overall, $\Ext$ runs in time $\Otilde(n)$ after a 
    preprocessing step in this case.
\end{proof}

\section{A Faster Instantiation of Trevisan's Extractor}\label{sec:fast-trevisan}

We first recall Trevisan's extractor \cite{Tre01,RRV02}, $\Tre \colon \B^{n} \times \B^{d} \rightarrow \B^{m}$,
set to some designated error $\eps > 0$. We will need the notion of weak designs, due to Raz, Reingold, and Vadhan~\cite{RRV02}.

\begin{definition}[weak design]
    A collection of sets $S_1,\dots,S_m\subseteq[d]$ is an \emph{$(\ell,\rho)$-weak design} if for all $i \in[m]$ we have $|S_i|=\ell$ and
    \begin{equation*}
        \sum_{j<i}2^{|S_i\cap S_j|}\leq \rho(m-1).
    \end{equation*}
\end{definition}
We will also need a $\delta$-balanced code $\mathcal{C} \colon \B^n \rightarrow \B^{\bar{n}}$.
The parameters of the weak design affect the extractor's parameters and can be set in a couple of different ways. The parameter $\ell$ is set to be $\log\bar{n}$, typically $\rho$ is chosen according to $m$, $\eps$, and the desired entropy $k$, and then $d$ is chosen as a function of $\ell$, $m$, and $\rho$ according to the weak design (see \cite{RRV02}).
Given $x \in \B^n$ and $y \in \B^d$, Trevisan's extractor outputs
\begin{equation}\label{eq:trevisan}
\Tre(x,y) = (\bar{x}|_{y_{S_1}}, \dots, \bar{x}|_{y_{S_m}}),
\end{equation}
where we denote $\bar{x} = \mathcal{C}(x)$ and interpret each length-$\log\bar{n}$ bit-string $y_{S_i}$ as a location in $[\bar{n}]$.
For the runtime analysis, it will be important to recall that $\delta$ is set to be $\frac{\eps}{cm}$ for some universal constant $c$.

\begin{theorem}
Trevisan's extractor of \cref{eq:trevisan}, set to extract $m$ bits with any error $\eps > 0$, is computable in time $\widetilde{O}(n+m\log(1/\eps))$. 

On a RAM in the logarithmic cost model, Trevisan's extractor is computable in time $O(n)+m\log(1/\eps) \cdot \polylog(n)$ with a preprocessing time of $\widetilde{O}(m\log(n/\eps))$. In particular, there exists a universal constant $c$, such that whenever $m \le \frac{n}{\log^{c}(n/\eps)}$, it runs in time $O(n)$, without the need for a separate preprocessing step.
\end{theorem}
\begin{proof}
Looking at \cref{eq:trevisan}, note that we only need to compute $m$ coordinates of $\mathcal{C}(x)$.
To compute those $m$ coordinates, $y_{S_1},\ldots,y_{S_m}$, we first need to compute the weak design itself. Note that this can be seen as a preprocessing step, since it only depends on the parameters of the extractor, and not on $x$ or $y$. We will use the following result.
\begin{claim}[\cite{FYEC24}, Section A.5]
For every $\ell,m \in \mathbb{N}$ and $\rho > 1$, there exists an $(\ell,\rho)$-weak design $S_1,\ldots,S_{m} \subseteq [d]$ with $d = O(\frac{\ell^2}{\log\rho})$, computable in time $\widetilde{O}(m\ell)$.
\end{claim}

Once we have our preprocessing step, we are left with computing the code. By \cref{cor:biased},
we can choose $\bar{n}$ so that $n/\bar{n} = \delta^{c}$ for some universal constant $c$,
and so $\bar{n} = n \cdot \poly(m,1/\eps)$ and $\ell = \log\bar{n}=O(\log(n/\eps))$. 
Generating the design can then be done in time $\widetilde{O}(m\log(n/\eps))$.
Now, \cref{cor:biased} tells us that any $m$
bits of $\mathcal{C}(x)$ can be computed in time 
\[
\widetilde{O}(n)+ m\log(1/\delta) \cdot \polylog(n)= \widetilde{O}(n+m\log(1/\eps)).
\]
On a RAM in the logarithmic cost model, we can use the variant of $\mathcal{C}$ that uses Spielman's code as a base code (see \cref{remark:spielman}) and get a runtime of $O(n) + m\log(1/\eps) \cdot \polylog(n)$. This gives a truly linear time  construction whenever $m$ is at most $\frac{n}{\log(1/\eps)\polylog(n)}$.
\end{proof}

We conclude by noting that there is a natural setting of parameters under which Trevisan's extractor
gives logarithmic seed and linear (or near-linear) time. When $m = k^{\Omega(1)}$, the parameters can be set so that $d = O\left(\frac{\log^{2}(n/\eps)}{\log k}\right)$. We thus have the following corollary.
\begin{corollary}
For every $n \in \mathbb{N}$, any constant $c > 1$, and any constants $\alpha,\beta \in (0,1)$,
Trevisan's extractor $\Tre \colon \B^n \times \B^d \rightarrow \B^m$ can be instantiated as a $(k = n^{\alpha},\eps = n^{-c})$ extractor with $d = O(\log n)$, $m = k^{\beta}$, and 
given $x \in \B^n$ and $y \in \B^d$, $\Tre(x,y)$ is computable in time $\widetilde{O}(n)$ (or $O(n)$ time, depending on the model).
\end{corollary}

\bibliographystyle{alpha}
\bibliography{refs}

\appendix

\end{document}